\documentclass[preprint,12pt]{elsarticle}
\usepackage{pepa}
\usepackage{color,graphicx}
\usepackage{subfigure}
\usepackage{amssymb,amsmath,amsthm,amsfonts,eepic,graphicx}
\usepackage{amscd}
\usepackage{longtable}
\usepackage{algorithmic}
\usepackage{algorithm}

\newtheorem{theorem}{Theorem}
\newtheorem{proposition}{Proposition}
\newtheorem{lemma}{Lemma}
\newtheorem{remark}{Remark}
\newtheorem{corollary}{Corollary}
\newtheorem{definition}{Definition}

\newtheorem{conjecture}{Conjecture}
\newtheorem{model}{Model}

\journal{Performance Evaluation}

\begin{document}
\begin{frontmatter}

\title{Fundamental Results on Fluid Approximations of Stochastic Process Algebra Models}


\author{Jie Ding} \ead{jieding@yzu.edu.cn}
\address{School of Information Engineering, The University of Yangzhou, Yangzhou, China}

\author{Jane Hillstion} \ead{jane.hillston@ed.ac.uk}
\address{ Laboratory for Foundations of Computer Science,
School of Informatics, The University of Edinburgh, Edinburgh, UK.}

\begin{abstract}

%Performance modelling provides an important route to gaining insight
%about how systems will perform, particularly with respect to
%scalability. However, the size and complexity of large scale systems
%challenge the capabilities of discrete state-based modelling
%formalisms.  Recently, for the stochastic process algebra PEPA, a
%novel approach to this problem has been developed, making a
%continuous state space approximation as a set of ordinary
%differential equations (ODEs).  In this paper we establish some
%fundamental properties of this approximation.  In particular we show
%the existence, uniqueness and boundedness of solutions of the ODEs.
%Furthermore, we present the convergence of the solutions for models
%without synchronisation and their relationship to the steady state
%distributions of the corresponding Markov chains. Moreover, a case
%study on the convergence for a model with synchronisations is
%demonstrated.

\par In order to avoid the state space explosion problem encountered
in the quantitative analysis of large scale PEPA models, a fluid
approximation approach has recently been proposed, which results in
a set of ordinary differential equations (ODEs) to approximate the
underlying continuous time Markov chain (CTMC). This paper presents
a mapping semantics from PEPA to ODEs based on a numerical
representation scheme, which extends the class of PEPA models that
can be subjected to fluid approximation. Furthermore, we have
established the fundamental characteristics of the derived ODEs,
such as the existence, uniqueness, boundedness and nonnegativeness
of the solution. The convergence of the solution as time tends to
infinity for several classes of PEPA models, has been proved under
some mild conditions. For general PEPA models, the convergence is
proved under a particular condition, which has been revealed to
relate to some famous constants of Markov chains such as the
spectral gap and the Log-Sobolev constant. This thesis has
established the consistency between the fluid approximation and the
underlying CTMCs for PEPA, i.e.\ the limit of the solution is
consistent with the equilibrium probability distribution
corresponding to a family of underlying density dependent CTMCs.
\end{abstract}


\end{frontmatter}

\tableofcontents

\section{Introduction}

Stochastic process algebras, such as PEPA~\cite{Jane1},
TIPP~\cite{TIPP}, EMPA~\cite{EMPA}, IMC~\cite{IMC-Thesis}, are
powerful modelling formalisms for concurrent systems which have
enjoyed considerable success over the last decade. Such modeling can
help designers and system managers by allowing aspects of a system
which are not readily tested, such as scalability, to be analysed
before a system is deployed. However, both model construction and
analysis can be challenged by the size and complexity of large scale
systems. This problem, called \emph{state space explosion}, is
inherent in the discrete state approach employed in stochastic
process algebras and many other formal modelling approaches. To
overcome this problem, many work devoted exploiting the
compositionality of the process algebra to decompose or simplify the
underlying CTMC e.g.\
\cite{Towards-Product-Form,HillstonNigel99-product-form,
ClarkHillston-Product-Form-Insensitive,Harrison-Turning-back,
Mertsiotakis97-PhDThesis,Semi-numerical-Solution-SPA}. Another
technique %which has been taken to exploit the model compositionality
is the use of Kronecker
algebra~\cite{Kloul-Efficient-Kronecker-Representation,Kloul-Formal-Techniques,
Buchholz-Compositional-Markovian-ProcessAlgebra,Rettelbach94-compositionalminimal}.
In addition, abstract Markov chains and stochastic bounds techniques
have also been used to analyse large scale PEPA
models~\cite{MichaelSmith-Abstraction-Model-Checking,MichaelSmith-Composition-Abstraction}.

The techniques reported above are based on the discrete state space.
Therefore as the size of the state space is extremely large, these
techniques are not always strong enough to handle the state space
explosion problem. For example, in the modelling of biochemical
mechanisms using the stochastic $\pi$-calculus
\cite{PI-Calculus-Bio-Sto-Extension,Pi-bio-lambda-switch} and PEPA
\cite{Calder-06strongercomputational,RKIPTCSB06}, the state space
explosion problem becomes almost insurmountable. Consequently in
many cases models are analysed by discrete event simulation rather
than being able to abstractly consider all possible behaviours.

To avoid this problem Hillston proposed a radically different
approach in \cite{Jane2} from the following two perspectives:
choosing a more abstract state representation in terms of state
variables, quantifying the types of behaviour evident in the model;
and assuming that these state variables are subject to continuous
rather than discrete change. This approach results in a set of ODEs,
leading to the evaluation of transient, and in the limit, steady
state measures.

\par  However, there
are not many discussions on the fundamental problems, such as the
existence, uniqueness, boundedness and nonnegativeness of the
solution, as well as the its asymptotic behaviour as time tends to
infinity, and the relationship between the derived ODEs and the
underlying CTMCs for general PEPA models. Solving these problems can
not only bring confidence in the new approach, but can also provide
new insight into, as well as a profound understanding of,
performance formalisms. This paper will focus on these topics  and
give answers to these problems.

\par The remainder of this paper is structured as follows.
Section~2 will give a brief introduction to PEPA as well as a
numerical representation scheme developed for PEPA. Based on this
scheme, the fluid approximation of PEPE models will be introduced in
Section~3. In this section, the existence and uniqueness of
solutions of the derived ODEs will be presented. In addition, we
will show that for a PEPA model without synchronisations, the
solution of the ODEs converges as time goes to infinity and the
limit coincides with the steady-state probability distribution of
the underlying CTMC\@. In Section~4, we demonstrate the consistency
between the fluid approximation and the Markov chains underlying the
same PEPA model. This relationship will be utilised to investigate
the long-time behaviour of the ODEs' solutions in Section~5. The
convergence of the solutions will be proved under a particular
condition, which relates the convergence problem to some well-known
constants of Markov chains such as the spectral gap and the
Log-Sobolev constant. Section~6 and~7 present an analytic approach
to analyse the fluid approximation. For several classes of PEPA
models, the convergence will be demonstrated under some mild
conditions, and the coefficient matrices of the derived ODEs have
been exposed to have the following property: all eigenvalues are
either zeros or have negative real parts. In addition, the
structural property of invariance in PEPA models will be shown to
play an important role in the proof of convergence. Finally, after
presenting some related work in Section~8, we conclude the paper in
Section~9.

\section{The PEPA modelling formalism}

This section will briefly introduce the PEPA language and its
numerical representation scheme. The numerical representation scheme
for PEPA was developed by Ding in his thesis~\cite{JieThesis}, and
represents a model numerically rather than syntactically supporting
the use of mathematical tools and methods to analyse the model.

\subsection{Introduction to PEPA}\label{section:PEPA}

PEPA (Performance Evaluation Process Algebra)~\cite{Jane1},
developed by Hillston in the 1990s, is a high-level model
specification language for low-level stochastic models, and
describes a system as an interaction of the components which engage
in activities.  In contrast to classical process algebras,
activities are assumed to have a duration which is a random variable
governed by an exponential distribution. Thus each activity in PEPA
is a pair $(\alpha,r)$ where $\alpha$ is the action type and $r$ is
the activity rate. The language has a small number of combinators,
for which we provide a brief introduction below; the structured
operational semantics can be found in~\cite{Jane1}.  The grammar is
as follows:
\begin{eqnarray*}
  S & ::= & (\alpha, r). S \mid S + S \mid C_S \\
  P & ::= & P \sync{L} P \mid P/L \mid C
\end{eqnarray*}
where $S$ denotes a \emph{sequential component} and $P$ denotes a
\emph{model component} which executes in parallel.  $C$ stands for a
constant which denotes either a sequential component or a model
component as introduced by a definition.  $C_S$ stands for constants
which denote sequential components.  The effect of this syntactic
separation between these types of constants is to constrain legal
PEPA components to be cooperations of sequential processes.

\textbf{Prefix}: The prefix component $(\alpha,r).S$ has a
designated first activity $(\alpha,r)$, which has action type
$\alpha$ and a duration which satisfies exponential distribution
with parameter $r$, and subsequently behaves as $S$.

\textbf{Choice}: The component $S+T$ represents a system which may
behave either as $S$ or as $T$. The activities of both $S$ and $T$
are enabled.  Since each has an associated rate there is a
\emph{race condition} between them and the first to complete is
selected.  This gives rise to an implicit probabilistic choice
between actions dependent of the relative values of their rates.

\textbf{Hiding}: Hiding provides type abstraction, but note that the
duration of the activity is unaffected.  In $P/L$ all activities
whose action types are in $L$ appear as the ``private'' type $\tau$.

\textbf{Cooperation}:$P\sync{L}Q$ denotes cooperation between $P$
and $Q$ over action types in the cooperation set $L$. The cooperands
are forced to synchronise on action types in $L$ while they can
proceed independently and concurrently with other enabled activities
(\emph{individual} activities).  The rate of the synchronised  or
\emph{shared} activity is determined by the slower cooperation (see
\cite{Jane1} for details). We write $P \parallel Q$ as an
abbreviation for $P\sync{L}Q$ when $L = \emptyset$ and $P[N]$ is
used to represent $N$ copies of $P$ in a parallel, i.e.\ $P[3] = P
\parallel P \parallel P$.

\textbf{Constant}: The meaning of a constant is given by a defining
equation such as $A\rmdef P$.  This allows infinite behaviour over
finite states to be defined via mutually recursive definitions.

On the basis of the operational semantic rules (please refer
to~\cite{Jane1} for details),
 a PEPA model may be regarded as a labelled multi-transition system
$$\left(\mathcal{C},\mathcal{A}ct,
\left\{\mathop{\longrightarrow}\limits^{(\alpha,r)}|(\alpha,r)\in
\mathcal{A}ct\right\}\right)$$ where $\mathcal{C}$ is the set of
components, $\mathcal{A}ct$ is the set of activities and the
multi-relation $\mathop{\longrightarrow}\limits^{(\alpha,r)}$ is
given by the rules. If a component $P$ behaves as $Q$ after it
completes activity $(\alpha, r)$, then denote the transition as
$P\mathop{\longrightarrow}\limits^{(\alpha,r)}Q$.

The memoryless property of the exponential distribution, which is
satisfied by the durations of all activities, means that the
stochastic process underlying the labelled transition system has the
Markov property. Hence the underlying stochastic process is a
CTMC\@. Note that in this representation the states of the system
are the syntactic terms derived by the operational semantics.  Once
constructed the CTMC can be used to find steady state or transient
probability distributions from which quantitative performance  can
be derived.

\subsection{Numerical Representation of PEPA Models}
\label{subsec:Background-NumericalRepresentation}

As explained above there have been two key steps in the use of fluid
approximation for PEPA models: firstly, the shift to a numerical
vector representation of the model, and secondly, the use of
ordinary differential equations to approximate the dynamic behaviour
of the underlying CTMC\@.  In this paper we are only concerned with
the former modification.

This section presents the numerical representation of PEPA models
developed in~\cite{JieThesis}. For convenience, we may represent a
transition $P\stackrel{(\alpha,r)}{\longrightarrow}Q$ as
$P\stackrel{(\alpha,r_\alpha^{P\rightarrow Q})}{\longrightarrow}Q$,
or often simply as $P\stackrel{\alpha}{\longrightarrow}Q$ since the
rate is not pertinent to structural analysis, where $P$ and $Q$ are
two local derivatives.  Following~\cite{Jane2}, hereafter the term
\emph{local derivative} refers to the local state of a single
sequential component. In the standard structured operational
semantics of PEPA used to derive the underlying CTMC, the state
representation is syntactic, keeping track of the local derivative
of each component in the model.  In the alternative \emph{numerical
vector representation} some information is lost as states only
record the number of instances of each local derivative:

\begin{definition} (\textbf{{Numerical Vector
Form}~\cite{Jane2}})\label{def:Ch3-NumVectorForm}. For an arbitrary
PEPA model $\mathcal{M}$ with $n$ component types ${C}_i,
i=1,2,\cdots,n$, each with $d_i$ distinct local derivatives, the
numerical vector form of $\mathcal{M}$, $\mathbf{m}(\mathcal{M})$,
is a vector with $d=\sum_{i=1}^nd_i$ entries. The entry
$\mathbf{m}[{C}_{i_j}]$ records how many instances of the $j$th
local derivative of component type ${C}_i$ are exhibited in the
current state.
\end{definition}

In the following we will find it useful to distinguish derivatives
according to whether they enable an activity, or are the result of
that activity:

\begin{definition}(\textbf{Pre and post local
derivative})\label{def:Ch3-pre-post-derivative}
\begin{enumerate}
  \item If a local derivative $P$ can enable an activity $\alpha$, that
is $P\mathop{\longrightarrow}\limits^{\alpha}\cdot$, then $P$ is
called a \emph{pre local derivative} of $\alpha$. The set of all pre
local derivatives of $\alpha$ is denoted by $\mathrm{pre}(\alpha)$,
called the
\emph{pre set} of $\alpha$.\\

  \item  If $Q$ is a
local derivative obtained by firing an activity $\alpha$, i.e.\
$\cdot\mathop{\longrightarrow}\limits^{\alpha}Q$, then $Q$ is called
a \emph{post local derivative} of $\alpha$. The set of all post
local derivatives is denoted by $\mathrm{post}(\alpha)$, called the
\emph{post
set} of $\alpha$.\\

  \item  The set of all the local derivatives derived from
$P$ by firing $\alpha$, i.e.\
$$
      \mathrm{post}(P,\alpha)=\{Q \mid P\stackrel{\alpha}{\longrightarrow}Q\},
$$
is called the \emph{post set of $\alpha$ from
  $P$}.
\end{enumerate}

\end{definition}

Within a PEPA model there may be many instances of the same activity
type but we will wish to identify those that have exactly the same
effect within the model.  In order to do this we additionally label
activities according to the derivatives to which they relate, giving
rise to \emph{labelled activities}:

\begin{definition}\label{definition:LabelledAcitvity}(\textbf{Labelled
Activity}).

\begin{enumerate}
  \item For any individual activity $\alpha$, for
each $P\in \mathrm{pre}(\alpha), Q\in \mathrm{post}(P,\alpha)$,
label $\alpha$ as
$\alpha^{P\rightarrow Q}$.\\

  \item For a shared activity $\alpha$, for each
  $(Q_1,Q_2,\cdots,Q_k)$ in
$$
\mathrm{post}(\mathrm{pre}(\alpha)[1],\alpha)\times
\mathrm{post}(\mathrm{pre}(\alpha)[2],\alpha)\times\cdots\times
\mathrm{post}(\mathrm{pre}(\alpha)[k],\alpha),
$$
label $\alpha$ as $\alpha^{w}$, where
$$w=(\mathrm{pre}(\alpha)[1]\rightarrow Q_1,\mathrm{pre}(\alpha)[2]\rightarrow
Q_2,\cdots,\mathrm{pre}(\alpha)[k]\rightarrow Q_k).$$
\end{enumerate}
Each $\alpha^{P\rightarrow Q}$ or $\alpha^{w}$ is called a
\textbf{labelled activity}. The set of all labelled activities  is
denoted by $\mathcal{A}_{\mbox{\small label }}$. For the above
labelled activities $\alpha^{P\rightarrow Q}$ and $\alpha^{w}$,
their respective pre and post sets are defined as
$$\mathrm{pre}(\alpha^{P\rightarrow Q})=\{P\},\;\mathrm{post}(\alpha^{P\rightarrow
Q})=\{Q\},$$
$$\mathrm{pre}(\alpha^{w})=\mathrm{pre}(\alpha),\;\mathrm{post}(\alpha^{w})
=\{Q_1,Q_2,\cdots,Q_k\}.$$
\end{definition}

In the numerical representation scheme, the transitions between
states of the model are represented by a matrix, termed the
\emph{activity matrix} --- this records the impact of the labelled
activities on the local derivatives.

\begin{definition}\label{definition: JieActivityMatrix}(\textbf{{Activity
Matrix, Pre Activity Matrix, Post Activity Matrix}}). For a model
with $N_{\mathcal{A}_{\mbox{\small label }}}$ labelled activities
and $N_\mathcal{D}$ distinct local derivatives, the activity matrix
$\mathbf{C}$ is an $N_\mathcal{D}\times N_{\mathcal{A}_{\mbox{\small
label }}}$ matrix, and the entries are defined as following
\begin{equation*}
     \mathbf{C}(P_i,\alpha_j)=\left\{
             \begin{array}{ll}
               +1 \hspace*{3mm} & \mbox{if }\;\; P_i\in \mathrm{post}(\alpha_j)\\
               -1 & \mbox{if }\;\; P_i\in \mathrm{pre}(\alpha_j) \\
               0 & \mbox{otherwise}
             \end{array}
     \right.
\end{equation*}
where $\alpha_j$ is a labelled activity. The pre activity matrix
$\mathbf{C^{pre}}$ and post activity matrix $\mathbf{C^{post}}$  are
defined as
\begin{equation*}
     \mathbf{C^{Pre}}(P_i,\alpha_j)=\left\{
             \begin{array}{ll}
               +1 \hspace*{3mm}  & \mbox{if } \;\; P_i\in \mathrm{pre}(\alpha_j) \\
                0 & \mbox{otherwise}.
             \end{array}
     \right.,
\end{equation*}
\begin{equation*}
     \mathbf{C^{Post}}(P_i,\alpha_j)=\left\{
             \begin{array}{ll}
               +1 \hspace*{3mm} & \mbox{if } \;\;P_i\in \mathrm{post}(\alpha_j)\\
                0 & \mbox{otherwise}.
             \end{array}
     \right.
\end{equation*}
\end{definition}

From Definitions~\ref{definition:LabelledAcitvity}
and~\ref{definition: JieActivityMatrix}, each column of the activity
matrix corresponds to a system transition and each transition can be
represented by a column of the activity matrix. The activity matrix
equals the difference between the pre and post activity matrices,
i.e.\ $\mathbf{C=C^{Post}-C^{Pre}}$. The rate of the transition
between states is specified by a \emph{transition rate function},
but we omit this detail here since we are concerned with qualitative
analysis.  See \cite{JieThesis} for details.

We first give the definition of the apparent rate
of an activity in a local derivative.

\begin{definition}(\textbf{{Apparent Rate of $\alpha$ in $P$}})\label{def:DingApparentRate}
 Suppose $\alpha$ is an activity of a PEPA model and $P$ is a local derivative
 enabling $\alpha$ (i.e.\ $P\in \mathrm{pre}(\alpha)$). Let $\mathrm{post}(P,\alpha)$
be the set of all the local derivatives derived from $P$ by firing $\alpha$, i.e.\ $
     \mathrm{post}(P,\alpha)=\{Q\mid P\stackrel{(\alpha,r_\alpha^{P\rightarrow
Q})}{\longrightarrow}Q\}.
$ Let
\begin{equation}\label{eq:r_l^P}
       r_\alpha(P)=\sum_{Q\in \mathrm{post}(P,\alpha)}r_\alpha^{P\rightarrow Q}.
\end{equation}
The \emph{apparent rate} of $\alpha$ in $P$ in  state $\mathbf{x}$,
denoted by $r_\alpha(\mathbf{x},P)$, is defined as
\begin{equation}\label{eq:r_l(P)}
         r_\alpha(\mathbf{x},P)=\mathbf{x}[P]r_\alpha(P).%=\sum_{Q\in post(P,\alpha)}r_\alpha^{P\rightarrow Q},
\end{equation}
\end{definition}
The above definition is used to define the following transition rate
function.
\begin{definition}(\textbf{{Transition Rate Function}})
\label{def:TransitionRateFunction} Suppose $\alpha$ is an activity of a
PEPA model and
 $\mathbf{x}$ denotes a state vector.

\begin{enumerate}
  \item  If $\alpha$ is individual, then for each
  $P\stackrel{(\alpha,r^{P\rightarrow Q})}{\longrightarrow}Q$,
the transition rate function of labelled activity $\alpha^{P\rightarrow
Q}$ in state $\mathbf{x}$ is defined as
\begin{equation}\label{eq:Ch3-RateFunction1}
     f(\mathbf{x}, \alpha^{P\rightarrow Q})=\mathbf{x}[P]r_\alpha^{P\rightarrow Q}.
\end{equation}

\item If $\alpha$ is synchronised, with
$\mathrm{pre}(\alpha)=\{P_1,P_2,\cdots,P_k\}$, then for each
$$(Q_1,Q_2,\cdots,Q_k)\in
\mathrm{post}(P_1,\alpha)\times \mathrm{post}(P_2,\alpha)\times\cdots\times
\mathrm{post}(P_k,\alpha),$$
 let $w=(P_1\rightarrow Q_1,P_2\rightarrow
Q_2,\cdots,P_k\rightarrow Q_k)$. Then the transition rate function
of labelled activity $\alpha^{w}$ in state $\mathbf{x}$ is defined as
\begin{equation*}
f(\mathbf{x}, \alpha^w)=\left(\prod_{i=1}^k\frac{r_\alpha^{P_i\rightarrow
Q_i}}{r_\alpha(P_i)}\right)\min_{i\in\{1,\cdots,k\}}\{r_\alpha(\mathbf{x},P_i)\},
\end{equation*}
where $r_\alpha(\mathbf{x},P_i)=\mathbf{x}[P_i]r_\alpha(P_i)$ is the apparent
rate of $\alpha$ in $P_i$ in state $\mathbf{x}$. So
\begin{equation}\label{eq:Ch3-RateFunction20}
f(\mathbf{x}, \alpha^w)=\left(\prod_{i=1}^k\frac{r_\alpha^{P_i\rightarrow
Q_i}}{r_\alpha(P_i)}\right)\min_{i\in\{1,\cdots,k\}}\{\mathbf{x}[P_i]r_\alpha(P_i)\}.
\end{equation}
\end{enumerate}
\end{definition}

Note that Definition~\ref{def:TransitionRateFunction} accommodates the
passive or unspecified rate $\infty$. An algorithm for automatically
deriving the numerical representation of a PEPA model was presented
in~\cite{JieThesis}.

\begin{remark}\label{remark:Ch3-Top*0=0}
Definition~\ref{def:TransitionRateFunction} accommodates the passive
or unspecified rate $\infty$. If there are some $r^{U\rightarrow
V}_l$ which are $\infty$, then the relevant calculation in the rate
functions (\ref{eq:Ch3-RateFunction1}) and
(\ref{eq:Ch3-RateFunction20}) can be made  according to the
following inequalities and equations that define the comparison and
manipulation of unspecified activity rates (see Section 3.3.5
in~\cite{Jane1}):
\begin{displaymath}
\begin{array}{cl}
 r<w\infty             & \emph{\mbox{for all $r\in \mathbb{R}^{+}$ and for all $w\in \mathbb{N}$}} \\
  w_1\infty<w_2\infty    & \emph{\mbox{if $w_1<w_2$ for all $w_1,w_2\in \mathbb{N}$}} \\
  w_1\infty+w_2\infty=(w_1+w_2)\infty & \emph{\mbox{for all $w_1,w_2\in \mathbb{N}$}}\\
  \frac{w_1\infty}{w_2\infty}=\frac{w_1}{w_2} & \emph{\mbox{for all $w_1,w_2\in \mathbb{N}$}}
\end{array}
\end{displaymath}
Moreover, we assume that $0\cdot\infty=0$. So the terms such as
``$\min\{A\infty, rB\}$'' are interpreted as~\cite{WormAttacks}:
\begin{equation*}
    \min\{A\infty, rB\}=\left\{\begin{array}{cc}
                                 rB, & A>0, \\
                                 0, & A=0.
                               \end{array}
    \right.
\end{equation*}
\end{remark}

The transition rate function has the following
properties~\cite{JieThesis}:
\begin{proposition}\label{proposition:Ch3-RateFunctionComparison}
The transition rate function is nonnegative; if $P$ is a pre local
derivative of $l$, i.e. $P\in \mathrm{pre}(l)$, then the transition
rate function of $l$ in a state $\mathbf{x}$ is less than the
apparent rate of $l$ in $U$ in this state, that is
$$
   0\leq f(\mathbf{x},l)\leq r_l(\mathbf{x},P)=\mathbf{x}[P]r_l(P),
$$
where $r_l(P)$ is the apparent rate of $l$ in $P$ for a single
instance of $P$.
\end{proposition}

\begin{proposition}\label{proposition:Ch3-HomogeousAndLipschitz}
Let $l$ be an labelled activity, and $\mathbf{x,y}$ be two states.
The transition rate function $f(\mathbf{x},l)$ defined in
Definition~\ref{def:TransitionRateFunction} satisfies:
\begin{enumerate}
  \item For any $H>0$, $Hf(\mathbf{x}/H,l)=f(\mathbf{x},l)$.
  \item There exists $M>0$ such that $|f(\mathbf{x},l)-f(\mathbf{y},l)|\leq
  M\|\mathbf{x}-\mathbf{y}\|$ \;for any $\mathbf{x},\mathbf{y}$ and $l$.
\end{enumerate}
\end{proposition}
Hereafter $\|\cdot\|$ denotes any matrix norm since all finite
matrix norms are equivalent. The first term of this proposition
illustrates a homogenous property of the rate function, while the
second indicates the Lipschtiz continuous property, both with
respect to states.

\section{Fluid approximations for PEPA models}

The section will introduce the fluid-flow approximations for PEPA
models, which leads to some kind of nonlinear ODEs. The existence
and uniqueness of the solutions of the ODEs will be established.
Moreover, a conservation law satisfied by the ODEs will be shown.

\subsection{State space explosion problem: an illustration by a tiny example}
\label{section:Ch3-StateSpaceExplosion} Let us first consider the
following tiny example.

\begin{model}\label{model:Uer-Provider}
\begin{equation*}
\begin{split}
User_1 \rmdef &(task_1, a).User_2\\
User_2 \rmdef &(task_2, b).User_1\\
Provider_1 \rmdef &(task_1, a).Provider_2\\
Provider_2 \rmdef &(reset, d).Provider_1\\
 \underbrace{User_1||\cdots||User_1}_{\mbox{$M$ copies}}
&\sync{\{task1\}}
 \underbrace{ Provider_1||\cdots|| Provider_1}_{\mbox{$N$
 copies}}
\end{split}
\end{equation*}
\end{model}
%
%
%\begin{model}\label{model:Uer-Provider}\textbf{PEPA Model of User-Provider System}\\
%{PEPA Definition for $User$}:
%\begin{equation*}
%\begin{split}
%User_1 \rmdef &(task_1, a).User_2\\
%User_2 \rmdef &(task_2, b).User_1\\
%\end{split}
%\end{equation*}
%{PEPA Definition for $Provider$}:
%\begin{equation*}
%\begin{split}
%Provider_1 \rmdef &(task_1, a).Provider_2\\
%Provider_2 \rmdef &(reset, d).Provider_1\\
%\end{split}
%\end{equation*}
%{System Equation}:
%\begin{equation*}
%\begin{split}
% \underbrace{User_1||\cdots||User_1}_{\mbox{$M$ copies}}
%\sync{\{task1\}}
% \underbrace{ Provider_1||\cdots|| Provider_1}_{\mbox{$N$
% copies}}
%\end{split}
%\end{equation*}
%
%\end{model}

%\par The system equation in a PEPA model
%specifies how many copies of each entity  are presented in the
%system, and how the components interact, by forcing cooperation on
%some of the activity types.  In Model~\ref{model:Uer-Provider}, the
%exact numbers of independent copies of the $User$ and $Provider$,
%i.e.\ $M$ and $N$ respectively, are both considered as variables.
According to the semantics of PEPA  originally defined
in~\cite{Jane1}, the size of the state space of the CTMC underlying
Model~\ref{model:Uer-Provider} is $2^{M+N}$. That is, the size of
the state space increases exponentially with the numbers of the
users and providers in the system. Consequently, the dimension of
the infinitesimal generator of the CTMC is $2^{M+N}\times 2^{M+N}$.
The computational complexity of solving the global balance equation
to get the steady-state probability distribution and thus  derive
the system performance, is therefore exponentially increasing with
the numbers of the components. When $M$ and/or $N$ are large, the
calculation of the stationary probability distribution will be
infeasible due to limited resources of memory and time. The problem
encountered here is the so-called \emph{state-space explosion}
problem.

\par A model aggregation technique, i.e. representing the states by
numerical vector forms, which is introduced by Gilmore \emph{et al.}
in~\cite{Ribaudo-AggregatingAlgorithm} and by Hillston
in~\cite{Jane2}, can help to relieve the state-space explosion
problem. It has been proved in~\cite{JieThesis} that by employing
numerical vector forms the size of the state space can be reduced to
$2^{M}\times 2^N$ to $(M+1)\times (N+1)$, without relevant
information and accuracy loss. However, this does not imply there is
no complexity problem.
%In practice, when hundreds of components
%exist in the system, for example $M=N=999$, then $(M+1)\times
%(N+1)=10^6$. This is still a large number, so that even the storage
%may become a problem with limited memory, let alone the analysis of
%the state space.
The following table,
Table~\ref{table:Ch3-RunningTimeStateSpaceDerivation}, gives the
runtimes of deriving the state space in several different scenarios.
All experiments were carried out using the PEPA Plug-in (v0.0.19)
for Eclipse Platform (v3.4.2), on a 2.66GHz Xeon CPU with 4Gb RAM
running Scientific Linux 5. The runtimes here are elapsed times
reported by the Eclipse platform.
\begin{table*}[htbp]
\begin{center}
%\begin{minipage}{15cm}
\begin{center}\caption{Elapsed time
of state pace
derivation}\label{table:Ch3-RunningTimeStateSpaceDerivation}
\begin{tabular}{|c|c|c|c|p{3.8cm}|} \hline\hline
  $(M,N)$ &  (300,300) & (350,300) &  (400,300)        & (400,400)  \\\hline
  time    & $2879$ ms & 4236 ms  & ``Java heap space" & ``GC overhead limit exceeded"\\\hline
\end{tabular}\end{center}
%\end{minipage}
\end{center}
\end{table*}

\par If there are 400 users and 300 providers in the system, the
Eclipse platform reports the error message of ``Java heap space",
while 400 users and 400 providers result in the error information of
``GC overhead limit exceeded". These experiments show that the
state-space explosion problem cannot be  completely solved by just
using the technique of numerical vector form, even for a tiny PEPA
model. That is, in order to do practical analysis for large scale
PEPA models we need new approaches.

%in terms of both qualitative and quantitative aspects,
%we need to go further to investigate PEPA and develop associated
%efficient computational methods and tools. The study of these topics
%constitutes the content of the next three chapters, whilst some
%basic technical preparation for the further research is given in
%this chapter. In particular, in the following sections the
%\emph{activity matrices} and \emph{transition rate functions} are
%defined to capture, especially in  numerical forms, the structural
%and timing information of PEPA models respectively.

\subsection{Fluid approximation of PEPA models}

In the numerical representation of PEPA presented in
Section~\ref{subsec:Background-NumericalRepresentation}, a numerical
vector form is introduced to capture the state information of models
with repeated components. In this vector form there is one entry for
each local derivative of each component type in the model. The
entries in the vector are no longer syntactic terms representing the
local derivative of the sequential component, but the number of
components currently exhibiting this local derivative. Each
numerical vector represents a single state of the system. The rates
of the transitions between states are specified by the transition
rate functions. For example, the transition from state $\mathbf{x}$
to $\mathbf{x}+l$ can be written as
$$
   \mathbf{x}\mathop{\longrightarrow}\limits^{(l,f(\mathbf{x},l))}\mathbf{x}+l,
$$
where $l$ is a transition vector corresponding to the labelled
activity $l$ (for convenience, hereafter each pair of transition
vectors and corresponding labelled activities shares the same
notation),  and $f(\mathbf{x},l)$ is the transition rate function,
reflecting the intensity of the transition from $\mathbf{x}$ to
$\mathbf{x}+l$.

\par The state space is inherently discrete with the entries within
the numerical vector form always being non-negative integers and
always being incremented or decremented in steps of one. As pointed
out in~\cite{Jane2}, when the numbers of components are large these
steps are relatively small and we can approximate the behaviour by
considering the movement between states to be continuous, rather
than occurring in discontinuous jumps. In fact, let us consider the
evolution of the numerical state vector. Denote the state at time
$t$ by $\mathbf{x}(t)$. In a short time $\Delta t$, the change to
the vector $\mathbf{x}(t)$ will be
$$
       \mathbf{x}(\cdot,t+\Delta t)-\mathbf{x}(\cdot,t)=F(\mathbf{x}(\cdot,t))\Delta t
       =\Delta t\sum_{l\in\mathcal{A}_{\mathrm{label}}}lf(\mathbf{x}(\cdot,t),l).
$$
Dividing by $\Delta t$ and taking the limit, $\Delta t\rightarrow
0$, we obtain  a set of ordinary differential equations (ODEs):
\begin{equation}\label{eq:ChFP-DingDerivedODEs}
   \frac{\mathrm{d}\mathbf{x}}{\mathrm{d}t}=F(\mathbf{x}),
\end{equation}
where
\begin{equation}\label{eq:ChFP-DingF(x)}
  F(\mathbf{x})=\sum_{l\in\mathcal{A}_{\mathrm{label}}}lf(\mathbf{x},l).
\end{equation}

\par Once the  activity matrix and the transition rate functions are
generated, the ODEs are immediately available. All of them can be
obtained automatically by a derivation algorithm presented
in~\cite{JieThesis}.

\par Let $U$ be a local derivative. For any transition vector $l$,
$l[U]$ is either $\pm 1$ or $0$. If $l[U]=-1$ then $U$ is in the pre
set of $l$, i.e.\ $U\in\mathrm{pre}(l)$,  while $l[U]=1$ implies
$U\in\mathrm{post}(l)$.
 According to
(\ref{eq:ChFP-DingDerivedODEs}) and (\ref{eq:ChFP-DingF(x)}),
\begin{equation}\label{eq:ChFP-ODEsDerivativeCentric}
\begin{split}
\frac{\mathrm{d}\mathbf{x}(U,t)}{\mathrm{d}t}&=\sum_ll[U]f(\mathbf{x},l)\\
&=-\sum_{l:l[U]=-1}f(\mathbf{x},l)+\sum_{l:l[U]=1}f(\mathbf{x},l)\\
&=-\sum_{\{l\mid
U\in\mathrm{pre}(l)\}}\!\!\!\!f(\mathbf{x},l)+\sum_{\{l\mid
U\in\mathrm{post}(l)\}}\hspace{-3mm}f(\mathbf{x},l).
\end{split}
\end{equation}
The term $\sum_{\{l\mid
U\in\mathrm{pre}(l)\}}\hspace{-0mm}f(\mathbf{x},l)$  represents the
``exit rates'' in the local derivative $U$, while the term
$\sum_{\{l\mid U\in\mathrm{post}(l)\}}\hspace{-0mm}f(\mathbf{x},l)$
reflects the ``entry rates'' in $U$. The formulae
(\ref{eq:ChFP-DingDerivedODEs}) and (\ref{eq:ChFP-DingF(x)}) are
activity centric while (\ref{eq:ChFP-ODEsDerivativeCentric}) is
local derivative centric. Our approach to derive ODEs has extended
previous results presented in the
literature~\cite{Jane2,WormAttacks,Jane4}, by relaxing restrictions
such as allowing shared activities may have different local rates,
each action name may appear in different local derivatives within
the definition of a sequential component, and may occur multiple
times with that derivative definition, etc.

% see Table~\ref{table:ChFP-MappingSemanticsComparison}.

%\begin{table}[htbp]
%  \begin{center}
%%\begin{minipage}{14.5cm}
%\begin{center}
%\caption{Comparison with  respect to restrictions}
%\label{table:ChFP-MappingSemanticsComparison}
%\begin{tabular}{|c|p{7.0cm}| p{0.8cm}|p{0.8cm}| p{0.8cm}| p{1.1cm}|}
%\hline
% No. & Restrictions relaxed & paper \cite{Jane2}
%  & paper \cite{WormAttacks}
% & paper \cite{Jane4}
% &  this paper \\\hline
%%1 & Not allow cooperation within groups of components of the same type. & $\surd$ & $\surd$
%% & $\surd$ & $\surd$\\\hline
% 1 & The cooperation set between interacting groups of components is not restricted to be
% the set of common action labels between these groups of components.  &  &$\surd$
% & &$\surd$
% \\\hline
% 2 & Shared activities may have different local rates.  &
% &
% & & $\surd$
% \\\hline
% 3 & Allow passive rate  &  &$\surd$ & &$\surd$
% \\\hline
% 4 & Each action name may appear in different local derivatives within the
%  definition of a sequential component, and may occur multiple times
%  with that derivative definition.
%   &  & & &$\surd$
% \\\hline
% 5 & Action hiding is  considered.
%   &  & & &
% $\surd$
%% \footnote{Action hiding is not discussed in this thesis, but can be employed based on our scheme.
%%In our scheme each unknown action $\tau$ can be distinguished since
%%they have distinct attached labels.
%%  }
%  \\\hline
%\end{tabular}
%\end{center}
%%\end{minipage}
%\end{center}
%\end{table}

For an arbitrary CTMC,
% there are backward and forward equations
%describing the evolution of the transition probabilities. From these
%equations
the evolution of probabilities distributed on each state can be
described by a set of linear ODEs~(\cite{QueueingNetworks}, page
52). For example, for the (aggregated) CTMC underlying a PEPA model,
the corresponding differential equations describing the evolution of
the probability distributions are
\begin{equation}\label{eq:ChFP-KolmogorovProbDistr}
    \frac{\mathrm{d}\mathbf{\pi}}{\mathrm{d}t}=Q^T\mathbf{\pi},
\end{equation}
where each entry of $\mathbf{\pi}(t)$ represents the probability of
the system being in  each state at time $t$, and $Q$ is an
infinitesimal generator matrix corresponding to the CTMC\@. Clearly,
the dimension of the coefficient matrix $Q$ is the square of the
size of the state space, which increases  as the number of
components increases.

\par  The derived ODEs (\ref{eq:ChFP-DingDerivedODEs}) describe the
evolution of the population of the components in \emph{each local
derivative}, while (\ref{eq:ChFP-KolmogorovProbDistr}) reflects the
the probability evolution  at \emph{each state}. Since the scale of
(\ref{eq:ChFP-DingDerivedODEs}), i.e. the number of the ODEs, is
only determined by the number of local derivatives and is unaffected
by the size of the state space, so it avoids the state-space
explosion problem. But the scale of
(\ref{eq:ChFP-KolmogorovProbDistr}) depends on the size of the state
space, so it suffers from the explosion problem. The price paid is
that the ODEs (\ref{eq:ChFP-DingDerivedODEs}) are generally
nonlinear due to synchronisations, while
(\ref{eq:ChFP-KolmogorovProbDistr}) is linear. However, if there is
no synchronisation contained then (\ref{eq:ChFP-DingDerivedODEs})
becomes linear, and there is some correspondence and consistency
between these two different types of ODEs, which will be
demonstrated in Section~\ref{section:ConvergWithoutSync}.

\subsection{Existence and uniqueness of  ODEs' solution}
\par For any set of ODEs, it is important to consider if the
equations have a solution, and if so whether that solution is
unique.

\begin{theorem}
For a given PEPA model without passive rates, the derived ODEs from
this model have a unique solution in the time interval
$[0,\infinity)$.
\end{theorem}

\begin{proof}
Notice that each entry of $F(\mathbf{x})=\sum_llf(\mathbf{x},l)$ is
a linear combination of the transition rate functions
$f(\mathbf{x},l)$, so $F(\mathbf{x})$ is globally Lipschitz
continuous since each $f(\mathbf{x},l)$ is globally Lipschitz
continuous. That is, there exits $M>0$ such that $\forall
\mathbf{x},\mathbf{y}$,
\begin{equation}\label{eq:F(x)Lip}
||F(\mathbf{x})-F(\mathbf{y})||\leq M||\mathbf{x}-\mathbf{y}||.
\end{equation}
By the classical theory in ODEs (e.g.\ Theorem~6.2.3
in~\cite{DifferentialEquations2}, page 14), the derived ODEs have a
unique solution in $[0,\infinity)$.
\end{proof}

As we have mentioned, in the formula
(\ref{eq:ChFP-ODEsDerivativeCentric}), the term $\sum_{\{l\mid
U\in\mathrm{pre}(l)\}}\hspace{-0mm}f(\mathbf{x},l)$ represents the
exit rates in the local derivative $U$, while the term \linebreak
$\sum_{\{l\mid U\in\mathrm{post}(l)\}}\hspace{-0mm}f(\mathbf{x},l)$
reflects the entry rates in $U$. For each type of component at any
time, the sum of all exit activity rates must be equal to the sum of
all entry activity rates, since the system is closed and there is no
exchange with the environment. This leads to the following
proposition.
\begin{proposition}\label{pro:ChFP-ODEsConservationLaw} Let $C_{i_j}$ be a local
derivative of  component type $C_i$. Then for any $i$ and $t$,
$\displaystyle\sum_j\frac{\mathrm{d}\mathbf{x}\left(C_{i_j},t\right)}{\mathrm{d}t}=0$,
and
$\sum_j\mathbf{x}\left(C_{i_j},t\right)=\sum_j\mathbf{x}\left(C_{i_j},0\right).$
\end{proposition}

\begin{proof} In the definition of activity matrix,
the numbers of $-1$ and $1$ appearing in the entries of any
transition vector $l$ (i.e. a column of the activity matrix), which
correspond to the component type $C_i$, are the
same~\cite{JieThesis}, i.e.
\begin{equation}\label{eq:ChFP-local-1}
\#\{j:l[C_{i_j}]=-1\}=\#\{j:l[C_{i_j}]=1\}.
\end{equation}
Let $\mathbf{y}$ be an indicator vector with the same dimension as
$l$ satisfying:
$$
    \mathbf{y}[{C}_{i_j}]=\left\{\begin{array}{cl}
                                   1, & \mathrm{if}\; l[{C}_{i_j}]=\pm1, \\
                                   0, & \mathrm{otherwise.}
                                 \end{array}\right.
$$
So $\mathbf{y}^Tl=0$ by (\ref{eq:ChFP-local-1}). Thus
\begin{equation*}
\mathbf{y}^T\frac{\mathrm{d}\mathbf{x}}{\mathrm{d}t}
=\mathbf{y}^T\sum_{l}lf(\mathbf{x},l)=\sum_{l}\mathbf{y}^Tlf(\mathbf{x},l)=0.
\end{equation*}
That is, $\displaystyle
\sum_j\frac{\mathrm{d}\mathbf{x}\left(C_{i_j},t\right)}{\mathrm{d}t}
=\mathbf{y}^T\frac{\mathrm{d}\mathbf{x}}{\mathrm{d}t}=0.$ So
$\sum_j\mathbf{x}\left(C_{i_j},t\right)$ is a constant and equal to
$\sum_j\mathbf{x}\left(C_{i_j},0\right)$, i.e. the number of the
copies of component type $C_i$ in the system initially.
\end{proof}

\par Proposition \ref{pro:ChFP-ODEsConservationLaw} means that the ODEs satisfy a
\emph{Conservation Law}, i.e. the number of each kind of component
remains constant at all times.

\subsection{Convergence and consistence of ODEs' solution: nonsynchronised models}
\label{section:ConvergWithoutSync}

\par Now we consider PEPA models without synchronisation. For this
special class of PEPA models, we will show that the solutions of the
derived ODEs have finite limits. Moreover, the limits coincide with
the steady-state probability distributions of the underlying CTMCs.

\subsubsection{Features of ODEs without synchronisations}

\par Suppose the PEPA model has no synchronisations.
Without loss of generality, we suppose that there is only one kind
of component $C$ in the system. In fact, if there are several types
of component in the system, the ODEs related to the different types
of component can be separated and  treated independently since there
are no interactions between them. Thus, we assume there is only one
kind of component $C$  and that $C$ has $k$ local derivatives: $C_1,
C_2, \cdots, C_k$. Then (\ref{eq:ChFP-DingDerivedODEs}) is
\begin{equation}\label{eq:ChFP-NonSyn1}
\begin{split}
  \frac{\mathrm{d}\left(\mathbf{x}(C_1,t),\cdots,\mathbf{x}(C_k,t)\right)^T}{\mathrm{d}t}=\sum_{l}lf(\mathbf{x},l).
\end{split}
\end{equation}
Since (\ref{eq:ChFP-NonSyn1}) are linear ODEs, we may rewrite
(\ref{eq:ChFP-NonSyn1}) as the following matrix form:
\begin{equation}\label{eq:ChFP-NonSynODEsMatrix}
 \frac{\mathrm{d}\left(\mathbf{x}(C_1,t),\cdots,\mathbf{x}(C_k,t)\right)^T}{\mathrm{d}t}
 =Q^T\left(\mathbf{x}(C_1,t),\cdots,\mathbf{x}(C_k,t)\right)^T,
\end{equation}
where $Q=\left(q_{ij}\right)$ is a $k\times k$ matrix.

\par $Q$ has many good properties.
\begin{proposition}\label{pro:ChFP-NonSynQmatrix}
$Q=\left(q_{ij}\right)_{k\times k}$ in
(\ref{eq:ChFP-NonSynODEsMatrix}) is an infinitesimal generator
matrix, that is, $\left(q_{ij}\right)_{k\times k}$ satisfies
\begin{enumerate}
  \item $0\leq -q_{ii}<\infinity$ for all $i$;
  \item $q_{ij}\geq 0$ for all $i\neq j$;
  \item $\sum_{j=1}^kq_{ij}=0$ for all $i$.
\end{enumerate}
\end{proposition}

\begin{proof} According to
(\ref{eq:ChFP-NonSynODEsMatrix}), we have
\begin{equation}\label{eq:ChFP-NonSyn-combine}
\frac{\mathrm{d}\mathbf{x}(C_i,t)}{\mathrm{d}t}=\sum_{j=1}^k\mathbf{x}(C_j,t)q_{ji}.
\end{equation}
Notice by (\ref{eq:ChFP-ODEsDerivativeCentric}),
$$
\frac{\mathrm{d}\mathbf{x}(C_i,t)}{\mathrm{d}t}=-\sum_{\{l\mid
C_i\in\mathrm{pre}(l)\}}\!\!\!\!f(\mathbf{x},l)+\sum_{\{l\mid
C_i\in\mathrm{post}(l)\}}\hspace{-3mm}f(\mathbf{x},l).
$$
So
\begin{equation}\label{eq:ChFP-NonSyn-ProProof-1}
\sum_{j=1}^k\mathbf{x}(C_j,t)q_{ji} =-\sum_{\{l\mid
C_i\in\mathrm{pre}(l)\}}\!\!\!\!f(\mathbf{x},l)+\sum_{\{l\mid
C_i\in\mathrm{post}(l)\}}\hspace{-3mm}f(\mathbf{x},l).
\end{equation}

Since there is no synchronisation in the system, the transition
function $f(\mathbf{x},l)$ is linear with  respect to $\mathbf{x}$
and therefore there is no nonlinear term,``$\min$'', in it. In
particular, if $C_i\in\mathrm{pre}(l) $, then
$f(\mathbf{x},l)=r_{l}(C_i)\mathbf{x}[C_i]$,  which is the apparent
rate of $l$ in $C_i$ in state $\mathbf{x}$ defined in
Definition~\ref{def:DingApparentRate}. We should point out that
according to our semantics of mapping PEPA models to ODEs, the fluid
approximation-version of $f(\mathbf{x},l)$ also holds, i.e.\
$f(\mathbf{x}(t),l)=r_l(C_{i})\mathbf{x}(C_{i},t)$. So
(\ref{eq:ChFP-NonSyn-ProProof-1}) becomes
\begin{equation}\label{eq:ChFP-NonSyn-ProProof-2}
\mathbf{x}(C_i,t)q_{ii}+\sum_{j\neq i}\mathbf{x}(C_j,t)q_{ji}
=\mathbf{x}(C_i,t)\sum_{\{l\mid
C_i\in\mathrm{pre}(l)\}}\hspace{-3mm}\left(-r_l(C_i)\right)+\sum_{\{l\mid
C_i\in\mathrm{post}(l)\}}\hspace{-3mm}f(\mathbf{x},l).
\end{equation}
Moreover, as long as $f(\mathbf{x},l)=r_l(C_i)\mathbf{x}(C_i)$ for
some $l$ and some positive constants $r_l(C_i)$, which implies that
$l$ can be fired at $C_i$, we must have $C_i\in\mathrm{pre}(l)$.
That is to say, if $C_i\in\mathrm{post}(l)$ then $f(\mathbf{x},l)$
cannot be of the  form of $r\mathbf{x}(C_i,t)$ for any constant
$r>0$. Otherwise, we have $C_i\in \mathrm{pre}(l)$, which results a
contradiction\footnote{In this paper we do not allow a self-loop in
the considered model. That is, any PEPA definition like
``$C\rmdef(\alpha,r).C$'' which  results in
$C\in\mathrm{pre}(\alpha)$ and $C\in\mathrm{post}(\alpha)$
simultaneously, is not allowed.} to $C_i\in\mathrm{post}(l)$. So
according to (\ref{eq:ChFP-NonSyn-ProProof-2}), we have
\begin{equation}\label{eq:ChFP-NonSyn-ProProof-3}
\mathbf{x}(C_i,t)q_{ii}=\mathbf{x}(C_i,t)\sum_{\{l\mid
C_i\in\mathrm{pre}(l)\}}\hspace{-3mm}\left(-r_l(C_i)\right),
\end{equation}
\begin{equation}\label{eq:ChFP-NonSyn-ProProof-4}
\sum_{j\neq i}\mathbf{x}(C_j,t)q_{ji} =\sum_{\{l\mid
C_i\in\mathrm{post}(l)\}}\hspace{-3mm}f(\mathbf{x},l).
\end{equation}
Thus by (\ref{eq:ChFP-NonSyn-ProProof-3}), $q_{ii}=\sum_{\{l\mid
C_i\in\mathrm{pre}(l)\}}\left(-r_l(C_i)\right)$, and $0\leq
-q_{ii}<\infinity$ for all $i$. Item 1 is proved.

\par Similarly, for any $C_j,\; j\neq i$, if
$f(\mathbf{x},l)=r\mathbf{x}(C_j,t)$ for some $l$ and positive
constant $r$, then clearly $C_j$ is in the pre set of $l$. That is
$C_j\in\mathrm{pre}(l)$. So by (\ref{eq:ChFP-NonSyn-ProProof-4}),
\begin{equation}
\mathbf{x}(C_j,t)q_{ji}=\sum_{\{l\mid C_j\in\mathrm{pre}(l),
C_i\in\mathrm{post}(l)\}}\hspace{-3mm}
f(\mathbf{x},l)=\mathbf{x}(C_j,t)\sum_lr_l^{C_j\rightarrow C_i},
\end{equation}
which implies $q_{ji}=\sum_lr_l^{C_j\rightarrow C_i}\geq 0$ for all
$i\neq j$, i.e. item 2 holds.

\par  We now prove item 3. By
Proposition~\ref{pro:ChFP-ODEsConservationLaw},
\begin{equation}\label{eq:ChFP-SumDxDt=0}
\frac{\mathbf{x}(C_1,t)}{dt}+\frac{\mathbf{x}(C_2,t)}{dt}+\cdots+\frac{\mathbf{x}(C_k,t)}{dt}=0.
\end{equation}
Then by (\ref{eq:ChFP-NonSyn-combine}) and
(\ref{eq:ChFP-SumDxDt=0}), for all $t$,
\begin{equation*}
\begin{split}
   &\mathbf{x}(C_1,t)\sum_{j=1}^kq_{1j}+
   \mathbf{x}(C_2,t)\sum_{j=1}^kq_{2j}+\cdots+\mathbf{x}(C_k,t)\sum_{j=1}^kq_{kj}\\
   =&\sum_{i=1}^k\mathbf{x}(C_i,t)\sum_{j=1}^kq_{ij}=\sum_{j=1}^k\sum_{i=1}^k\mathbf{x}(C_i,t)q_{ij}
   =\sum_{j=1}^k\frac{\mathrm{d}\mathbf{x}(C_j,t)}{\mathrm{d}t}=0.
\end{split}
\end{equation*}
This implies $\sum_{j=1}^kq_{ij}=0$ for all $i$.
\end{proof}

\par In the proof of Proposition~\ref{pro:ChFP-NonSynQmatrix}, we
have shown the relationship between the coefficient matrix $Q$ and
the activity rates:
$$
   q_{ii}=-\sum_{\{l\mid C_i\in\mathrm{pre}(l)\}}r_l(C_i),\quad q_{ij}=\sum_{l}r_l^{C_i\rightarrow
   C_j}\;(i\neq j).
$$
We point out that this infinitesimal generator matrix $Q_{k\times
k}$ may not be the infinitesimal generator matrix of the CTMC
derived via the usual semantics of PEPA (we call it the ``original"
CTMC for convenience). In fact, the original CTMC has a state space
with $k^N$ states and the dimension of its infinitesimal generator
matrix is $k^N\times k^N$, where $N$ is the total number of
components in the system. However, this $Q_{k\times k}$ is the
infinitesimal generator matrix of a CTMC underlying the PEPA model
in which there is only one copy of the component, i.e.\ $N=1$. To
distinguish this from the original one, we refer to this CTMC as the
\emph{``singleton'' CTMC}.

\subsubsection{Convergence and consistency for the ODEs}

%\begin{remark}
 Proposition~\ref{pro:ChFP-NonSynQmatrix} illustrates that the
 coefficient matrix of the derived ODEs is an
infinitesimal generator.  If
 there is only one component in the system, then equation
 (\ref{eq:ChFP-NonSynODEsMatrix}) captures the probability distribution
 evolution equations of the original CTMC\@. Based on this proposition, we can furthermore determine the
convergence of the solutions.
\begin{theorem}\label{thm:ChFP-NonSynConvergence}
Suppose $\mathbf{x}\left(C_{j},t\right)\;(j=1,2,\cdots,k)$ satisfy
(\ref{eq:ChFP-NonSyn1}), then for any given initial values
$\mathbf{x}\left(C_{j},0\right)\geq0\;(j=1,2,\cdots,k)$, there exist
constants $\mathbf{x}(C_j,\infinity)$, such that
\begin{equation}\label{}
  \lim_{t\rightarrow\infinity}\mathbf{x}(C_j,t)=\mathbf{x}(C_j,\infinity),
  \quad j=1,2,\cdots,k.
\end{equation}
\end{theorem}

\begin{proof}
By Proposition~\ref{pro:ChFP-NonSynQmatrix}, the matrix $Q$ in
(\ref{eq:ChFP-NonSynODEsMatrix}) is an infinitesimal generator
matrix. Consider a ``singleton'' CTMC which has the state space
$S=\{C_1, C_2, \cdots, C_k\}$, the infinitesimal generator matrix
$Q$ in (\ref{eq:ChFP-NonSynODEsMatrix}) and the initial probability
distribution
$\mathbf{\pi}(C_j,0)=\frac{\mathbf{x}(C_j,0)}{N}\;(j=1,2,\cdots,k)$.
Then according to Markov theory (\cite{QueueingNetworks}, page 52),
$\mathbf{\pi}(C_j,t)\;(j=1,2,\cdots,k)$, the probability
distribution of this new CTMC at time $t$, satisfies
\begin{equation}\label{NewCTMCDistrODES}
 \frac{\mathrm{d}\left(\mathbf{\pi}(C_1,t),\cdots,\mathbf{\pi}(C_k,t)\right)}{\mathrm{d}t}
 =\left(\mathbf{\pi}(C_1,t),\cdots,\mathbf{\pi}(C_k,t)\right)Q
\end{equation}
Since the singleton CTMC is assumed irreducible and
positive-recurrent, it has a steady-state probability distribution
$\{\mathbf{\pi}(C_j,\infinity)\}_{j=1}^k$, and
\begin{equation}\label{equ:sinlegton}
  \lim_{t\rightarrow\infinity}\mathbf{\pi}(C_j,t)=\mathbf{\pi}(C_j,\infinity),
  \quad j=1,2,\cdots,k.
\end{equation}
Note that $\frac{\mathbf{x}(C_{j},t)}{N}$ also satisfies
(\ref{NewCTMCDistrODES}) with the initial values
$\frac{\mathbf{x}(C_{j},0)}{N}$ equal to $\mathbf{\pi}(C_j,0)$,
where $N$ is the population of the components. By the uniqueness of
the solutions of (\ref{NewCTMCDistrODES}), we have
\begin{equation}\label{}
  \frac{\mathbf{x}(C_j,t)}{N}=\mathbf{\pi}(C_j,t),
  \quad j=1,2,\cdots,k,
\end{equation}
and hence by (\ref{equ:sinlegton}),
\begin{equation*}
  \lim_{t\rightarrow\infinity}\mathbf{x}(C_j,t)=\lim_{t\rightarrow\infinity}N\mathbf{\pi}(C_j,t)=N\mathbf{\pi}(C_j,\infinity),
  \quad j=1,2,\cdots,k.
\end{equation*}
\end{proof}

Clearly, if there are multiple types of component in the system,
then Theorem~\ref{thm:ChFP-NonSynConvergence} holds for each each
component type, since there is no cooperation between different
component types and they can be treated independently.

%$m$ types of components in the system: $C_1,C_2,\cdots,C_m$,
% each with $k_1,k_2,\cdots,k_m$ local derivatives respectively. Since
% there is no cooperation between different types of component, we
% can deal with each type independently.

%
%\begin{remark}Suppose there are $m$ types of components in
%the system: $C_1,C_2,\cdots,C_m$,
% each with $k_1,k_2,\cdots,k_m$ local derivatives respectively. Since
% there is no cooperation between different types of component, we
% can deal with each type independently.
%Thus, by Theorem~\ref{thm:ChFP-NonSynConvergence}, for each
%component type $C_i$,
%\begin{equation}\label{}
%  \lim_{t\rightarrow\infinity}\mathbf{x}(C_{i_j},t)=N_i\mathbf{\pi}(C_{i_j},\infinity),
%  \quad j=1,2,\cdots,k_i,
%\end{equation}
%where $\{\mathbf{\pi}(C_{i_j})\}_{j=1,2,\cdots,k_i}$ are the
%corresponding steady state distributions, and $N_i$ is the
%population of $C_i$, $i=1,2,\cdots,m$.
%\end{remark}

It is shown in~\cite{Gilmore2005} that for some special examples the
equilibrium solutions of the ODEs coincide with the steady state
probability distributions of the underlying original CTMC\@. This
theorem states that this holds for all for PEPA models without
synchronisations. Moreover, it is also exposed by this theorem that
the fluid approximation is consistent with the CTMC underlying the
same nonsynchronised PEPA model.

\section{Relating to density dependent CTMCs}

%For a PEPA model without synchronisation, the solution of the
%derived ODEs through the fluid approximation has a finite limit that
%is consistent with the steady-state distribution of the
%corresponding singleton CTMC, as
%Theorem~\ref{thm:ChFP-NonSynConvergence} exposes. However,

A general PEPA model may have synchronisations, which result in the
nonlinearity of the derived ODEs. However, it is difficult to rely
on pure analytical methods to explore the asymptotic behaviour of
the solution of the derived ODEs from an arbitrary PEPA model
(except for some special classes of models, see the next two
sections).

\par Fortunately,  Kurtz's theorem~\cite{Kurtz,KurtzBook} establishes the relationship
between a sequence of Markov chains and a corresponding set of ODEs:
the complete solution of some ODEs is the limit of a sequence of
Markov chains. In the context of PEPA, the derived ODEs can be
considered as the limit of pure jump Markov processes, as first
exposed in~\cite{Jane4} for a special case. Thus we  may investigate
the convergence of the ODEs' solutions by alternatively studying the
corresponding property of the Markov chains through this consistency
relationship. This approach leads  to  the result presented in the
next section: under a particular condition the solution will
converge and the limit is consistent with the limit steady-state
probability distribution of a family of CTMCs underlying the given
PEPA model. Let us first introduce the concept of density dependent
Markov chains underlying PEPA models.

\subsection{Density dependent Markov chains underlying PEPA
models}\label{section:ChFP-density-dependent-Markov}

\par In the numerical state vector representation scheme, each vector is a
single state and the rates of the transitions between states are
specified by the rate functions. For example, the transition from
state $\mathbf{x}$ to $\mathbf{x}+l$ can be written as
$$
   \mathbf{x}\mathop{\longrightarrow}\limits^{(l,f(\mathbf{x},l))}\mathbf{x}+l.
$$
Since all the transitions are only determined by the current state
rather than the previous ones, given any starting state a CTMC can
be obtained. More specifically, the state space of the CTMC is the
set of all reachable numerical state vectors $\mathbf{x}$. The
infinitesimal generator is determined by the transition rate
function,
\begin{equation}\label{equ:TransitionFunction}
           q_{\mathbf{x},\mathbf{x}+l}=f(\mathbf{x},l).
\end{equation}

\par Because the transition rate function is defined according to the
semantics of PEPA, the CTMC mentioned above is in fact the
aggregated CTMC underlying the given PEPA model. In other words, the
transition rate of the aggregated CTMC is specified by the
transition rate function in
Definition~\ref{def:TransitionRateFunction}.

\par It is obvious that the aggregated CTMC depends on the starting state
of the given PEPA model. By altering the population of components
presented in the model, which can be done by varying the initial
states, we may get a sequence of aggregated CTMCs. Moreover,
Proposition~\ref{proposition:Ch3-HomogeousAndLipschitz} indicates
that the transition rate function has the homogenous property:
$Hf(\mathbf{x}/H,l)=f(\mathbf{x},l),\linebreak \forall H>0$. This
property identifies the aggregated CTMC to be \emph{density
dependent}.

\begin{definition}~\cite{Kurtz}.
A family of CTMCs $\{X_n\}_n$ is called \emph{density dependent} if
and only if there exists a continuous function
$f(\mathbf{x},l),\;\mathbf{x}\in \mathbb{R}^d,\;l\in \mathbb{Z}^d$,
such that the infinitesimal generators of $X_n$ are given by:
$$
           q^{(n)}_{\mathbf{x},\mathbf{x}+l}=nf(\mathbf{x}/n,l),\quad l\neq 0,
$$
where $q^{(n)}_{\mathbf{x},\mathbf{x}+l}$ denotes an entry of the
infinitesimal generator of $X_n$, $\mathbf{x}$ a numerical state
vector and $l$ a transition vector.
\end{definition}

This allows us to conclude the following proposition.
\begin{proposition}\label{pro:DensityDependentMCfromPEPA}
Let $\{X_n\}$ be a sequence of aggregated CTMCs generated from a
given PEPA model (by scaling the initial state), then $\{X_n\}$ is
density dependent.
\end{proposition}

\begin{proof}
For any $n$, the transition between states is determined by
$$
        q^{(n)}_{\mathbf{x},\mathbf{x}+l}=f(\mathbf{x},l),
$$
where $\mathbf{x}, \mathbf{x}+l$ are state vectors, $l$ corresponds
to an activity, $f(\mathbf{x},l)$ is the rate of the transition from
state $\mathbf{x}$ to $\mathbf{x}+l$. By
Proposition~\ref{proposition:Ch3-HomogeousAndLipschitz},
$$nf(\mathbf{x}/n,l)=f(\mathbf{x},l).$$
So the infinitesimal generator of $X_n$ is given by:
$$
           q^{(n)}_{\mathbf{x},\mathbf{x}+l}=f(\mathbf{x},l)=nf(\mathbf{x}/n,l),\quad l\neq 0.
$$
Therefore, $\{X_n\}$ is a sequence of density dependent CTMCs.
\end{proof}

\par In particular, the family of density dependent CTMCs,
$\{X_n(t)\}$, derived from a given PEPA model
 with the starting condition $X_n(0)=n\mathbf{x}_0 \;(\forall n)$, is called the
\emph{density dependent CTMCs associated with $\mathbf{x}_0$}. The
CTMCs $\frac{X_n(t)}{n}$ are called the \emph{concentrated density
dependent} CTMCs. Here $n$ is called the \emph{concentration level},
indicating that the entries within the numerical vector states (of
$\frac{X_n(t)}{n}$) are incremented and decremented in steps of
$\frac1n$.

For example, Consider the following PEPA model, which is
Model~\ref{model:Uer-Provider} presented previously:
\begin{equation*}
\begin{split}
User_1 \rmdef &(task_1, a).User_2\\
User_2 \rmdef &(task_2, b).User_1\\
Provider_1 \rmdef &(task_1, a).Provider_2\\
Provider_2 \rmdef &(reset, d).Provider_1\\
 (User_1[M]) &
\sync{\{task1\}} (Provider_1[N]).
\end{split}
\end{equation*}
The activity matrix and transition rate functions have been
specified in Table~\ref{table:ChFP-ActMatrix-ModelUserProvider}. In
this table, $U_i,P_i\;(i=1,2)$ are the local derivatives
representing $User_i$ and $Provider_i$ respectively. For
convenience, the labelled activities or transition vectors
${task_1}^{(U_1\rightarrow U_2,P_1\rightarrow P_2)}$,
$task_2^{U_2\rightarrow U_1}$, ${reset}^{P_2\rightarrow P_1}$ will
subsequently be denoted by $l^{task_1},l^{task_2}, l^{reset}$
respectively.
\begin{table}[htbp]
\begin{center}
\caption{Activity matrix and transition rate function of
Model~\ref{model:Uer-Provider}}\label{table:ChFP-ActMatrix-ModelUserProvider}
\begin{tabular}{|c| c |c| c| }
  \hline
  $l$  &   ${task_1}^{(U_1\rightarrow U_2,P_1\rightarrow P_2)}$ &   $task_2^{U_2\rightarrow U_1}$
       &   ${reset}^{P_2\rightarrow P_1}$ \\\hline
  $U_1$& $-1$ & 1  & 0  \\%$x_1$             \\
  $U_2$& 1  & $-1$  & 0    \\%$x_2$             \\
  $P_1$ & $-1$  & 0  & 1    \\%$x_3$             \\
  $P_2$ & 1  & 0  & $-1$    \\\hline\hline%$x_4$             \\
  $f(\mathbf{x},l)$  & $a\min(\mathbf{x}[U_1],\mathbf{x}[P_1])$ &
  $b\mathbf{x}[U_2]$ & $d\mathbf{x}[P_2]$\\ \hline
  \end{tabular}
\end{center}
\end{table}

Suppose $\mathbf{x}_1=(M,0,N,0)^T=(1,0,1,0)^T$. Let $X_1(t)$ be the
aggregated CTMC underlying Model~\ref{model:Uer-Provider} with
initial state $\mathbf{x}_1$. Then the state space of $X_1(t)$,
denoted by $S_1$, is composed of
\begin{equation}\label{eq:ExampleS_1}
\begin{split}
    \begin{array}{ll}
       \mathbf{x}_1=(1,0,1,0)^T, & \mathbf{x}_2=(0,1,0,1)^T,\\
       \mathbf{x}_3=(1,0,0,1)^T, & \mathbf{x}_4=(0,1,1,0)^T. \\
     \end{array}
\end{split}
\end{equation}
According to the transition rate functions presented in
Table~\ref{table:ChFP-ActMatrix-ModelUserProvider}, we have, for
instance,
$$
    q^{(1)}_{\mathbf{x}_1,\mathbf{x}_2}=q_{\mathbf{x}_1,\mathbf{x}_1+l^{task_1}}
    =f(\mathbf{x}_1,l^{task_1})=a\min(\mathbf{x}_1[U_1],
    \mathbf{x}_1[P_1])=a.
$$
Varying the initial states we may get other aggregated CTMCs. For
example, let $X_2(t)$ be the aggregated CTMC corresponding to the
initial state $X_2(0)=2\mathbf{x}_0=(2,0,2,0)^T$. Then the state
space $S_2$ of $X_2(t)$ has the states
\begin{equation}\label{eq:ExampleS_2}
\begin{split}
    \begin{array}{lll}
       \mathbf{x}_1=(2,0,2,0)^T, & \mathbf{x}_2=(1,1,1,1)^T, & \mathbf{x}_3=(1,1,2,0)^T,\\
       \mathbf{x}_4=(1,1,0,2)^T, & \mathbf{x}_5=(0,2,1,1)^T, & \mathbf{x}_6=(2,0,1,1)^T,\\
       \mathbf{x}_7=(0,2,0,2)^T, & \mathbf{x}_8=(0,2,2,0)^T, & \mathbf{x}_9=(2,0,0,2)^T.
     \end{array}
\end{split}
\end{equation}
The rate of transition from $\mathbf{x}_1$ to $\mathbf{x}_2$ is
determined by
$$
    q^{(2)}_{\mathbf{x}_1,\mathbf{x}_2}=q_{\mathbf{x}_1,\mathbf{x}_1+l^{task_1}}
    =f(\mathbf{x}_1,l^{task_1})=2a=2f(\mathbf{x}_1/2,l^{task_1}).
$$
Similarly, let $X_n(t)$ be the aggregated CTMC corresponding to the
initial state $X_n(0)=n\mathbf{x}_0$. Then the transition from
$\mathbf{x}$ to $\mathbf{x}+l$ is determined by
\begin{equation*}
           q^{(n)}_{\mathbf{x},\mathbf{x}+l}=f(\mathbf{x},l)=nf(\mathbf{x}/n,l).
\end{equation*}
Thus a family of aggregated CTMCs, i.e. $\{X_n(t)\}$, has been
obtained from Model~\ref{model:Uer-Provider}. These derived
$\{X_n(t)\}$ are density dependent CTMCs associated with
$\mathbf{x}_0$. As illustrated by this example, the density
dependent CTMCs are obtained by scaling the starting state
$\mathbf{x}_0$. So the starting state of each CTMC is different,
because $X_n(0)=n\mathbf{x}_0$, i.e.\ $X_n(0)=n(M,0,N,0)^T$.

\subsection{Fluid approximation as the limit of the CTMCs}
\par As discussed above, a set of ODEs and a sequence of density
dependent Markov chains can be derived from the same PEPA model. The
former one is deterministic while the latter is stochastic. However,
both of them are determined by the same activity matrix and the same
rate functions that are uniquely generated from the given PEPA
model. Therefore, it is natural to believe that there is some kind
of consistency between them.

\par As we have mentioned, the complete solution of some ODEs can be the limit of a
sequence of Markov chains according to Kurtz's
theorem~\cite{Kurtz,KurtzBook}. Such consistency in the context of
PEPA has been previously illustrated for a particular PEPA
model~\cite{Jane4}. Here we give a modified version of this result
for general PEPA models, in which the convergence is in the sense of
almost surely rather than probabilistically as in~\cite{Jane4}.
%A
%sequence converges to a limit almost surely means that events for
%which this sequence does not converge to this limit have probability
%zero. Convergence in this sense can imply the convergence in
%probability, so it is stronger.
\begin{theorem}\label{thm:ChFP-ODEmeetsAggCTMC}%(\textbf{Modified Result of~\cite{Jane4}})
Let $X(t)$ be the solution of the ODEs
(\ref{eq:ChFP-DingDerivedODEs}) derived from a given PEPA model with
initial condition $\mathbf{x}_0$, and let $\{X_n(t)\}$ be the
density dependent CTMCs associated with $\mathbf{x}_0$ underlying
the same PEPA model. Let $\hat{X}_n(t)=\frac{X_n(t)}{n}$,
 then for any $t>0$,
\begin{equation}
        \lim_{n\rightarrow\infinity}\sup_{u\leq t}\|\hat{X}_n(u)-X(u)\|=0\quad \quad a.s.
\end{equation}
\end{theorem}

\begin{proof}
According to Kurtz's theorem, which is listed
in~\ref{section:Appendix-Some-Theorems}, it is sufficient to prove:
for any compact set $K\subset \mathbb{R}^{N_{d}}$,
\begin{enumerate}
\item $\exists M_K>0$ such that $\|F(\mathbf{x})-F(\mathbf{y})\|\leq M_K\|\mathbf{x}-\mathbf{y}\|$;
\item $\sum_l\|l\|\sup_{\mathbf{x}\in K}f(\mathbf{x},l)<\infinity$.
\end{enumerate}
Clearly, above term 1 is satisfied. Since $f(\mathbf{x},l)$ is
continuous by
Proposition~\ref{proposition:Ch3-HomogeousAndLipschitz}, it is
bounded on any compact $K$. Notice that any entry of $l$ takes
values in $\{1,-1,0\}$, so $\|l\|$ is bounded. Thus term 2 is
satisfied, which completes the proof.
\end{proof}

\par Theorem~\ref{thm:ChFP-ODEmeetsAggCTMC} allows us to investigate the
properties of $X(t)$ through studying the characteristics of the
family of CTMCs $\hat{X}_n(t)=\frac{X_n(t)}{n}$. Notice that
$X_n(t)$ takes values in the state space which corresponds to the
starting state $n\mathbf{x}_0$. Clearly, each state of
$\hat{X}_n(t)$ is bounded and nonnegative (more precisely, each
entry in any numerical state vector is nonnegative, and bounded by
$\|\mathbf{x}_0\|$ according to the conservative law).  So the ODEs'
solution $X(t)$ inherits these characteristics since $X(t)$ is the
limit of $\hat{X}_n(t)$ as $n$ goes to infinity. That is, $X(t)$ is
bounded and nonnegative. The proof is trivial and omitted here.
Instead, a purely analytic proof of these properties will be given
in Section~\ref{section:Analytic-Proof-Boundedness}.

\subsection{Consistency between the fluid approximation and the CTMCs}

As shown in Theorem~\ref{thm:ChFP-ODEmeetsAggCTMC}, for a given PEPA
model with synchronisations, the derived ODEs can be taken as the
underlying density dependent CTMC with the concentration level
infinity. If the given model has no synchronisations, then by
Proposition~\ref{pro:ChFP-NonSynQmatrix} and the proof of
Theorem~\ref{thm:ChFP-NonSynConvergence}, the ODEs coincide the
probability distribution evolution equations of the CTMC with the
concentration level one, except for a scaling factor. These two
conclusions embody the consistency between the fluid approximation
and the CTMCs for PEPA models. Moreover, as we will see in the next
section, the limit of the ODEs' solution as time tends to infinity
(if it exists) is consistent with the limit of the expectations of
the corresponding CTMCs.

\par However, it is natural to ask such a question:  for a PEPA model with
synchronisations, since the ODEs corresponds the CTMC with the
concentration level infinity and our performance evaluation is based
on the usual underlying CTMC, i.e. the CTMC with the concentration
level one, what is the loss brought by using the fluid approximation
approach? Although in the sense of concentration level there is a
gap between one and infinity, but in practice for large scale
synchronised models (i.e. models with a large number of repetitive
components), the relative errors of performance measure such as
throughput and average response time between the concentration level
one and infinity are usually small enough to be ignored (i.e. less
than $5\%$)~\cite{JieThesis}. Therefore, it is safe to employ the
fluid approximation for large scale models.

\par This paper more likely concentrates on the long-time behaviour
the ODEs' solution. In the following sections, we focus on the
problem of whether the ODEs' solution converges as time goes to
problem will be presented in the next section, which is based on the
consistency relationship between the derived ODEs and the CTMCs
revealed in Theorem~\ref{thm:ChFP-ODEmeetsAggCTMC}.

\section{Convergence of ODEs' solution: a probabilistic approach}

\par Analogous to the steady-state probability distributions of the Markov
chains underlying PEPA models, upon which performance measures such
as throughput and utilisation can be derived, we expect the solution
of the generated ODEs to have similar equilibrium conditions. In
particular, if the solution has a limit as time goes to infinity we
will be able to similarly obtain the performance from the steady
state, i.e. the limit. Therefore, whether the solution of the
derived ODEs converges becomes an important problem.

\par We should point out that Kurtz's theorem cannot directly apply to
the problem of whether or not the solution the derived ODEs
converges. This is because  Kurtz's theorem only deals with the
approximation between the ODEs and Markov chains during any finite
time, rather than considering the asymptotic behaviour of the ODEs
as time goes to infinity. This section will present our
investigation and results about this problem.

\subsection{Convergence under a particular condition}

\par  We follow the assumptions in Theorem~\ref{thm:ChFP-ODEmeetsAggCTMC}.
 Denote the expectation of
$\hat{X}_n(t)$ as $\hat{M}_n(t)$, i.e.
$\hat{M}_n(t)=E[\hat{X}_n(t)]$. For any $t$, the stochastic
processes $\{\hat{X}_n(t)\}_n$ converge to the deterministic $X(t)$
when $n$ tends to infinity, as
Theorem~\ref{thm:ChFP-ODEmeetsAggCTMC} shows. It is not surprising
to see that $\{\hat{M}_n(t)\}_n$, the expectations of
$\{\hat{X}_n(t)\}_n$, also converge to $X(t)$ as $n\rightarrow
\infinity$:
\begin{lemma}\label{lem:Mn(t)=X(t)} For any $t$,
$$
          \lim_{n\rightarrow\infinity}\hat{M}_n(t)=X(t).
$$
\end{lemma}
\begin{proof} Since $X(t)$ is deterministic, then $E[X(t)]=X(t)$.
By Theorem~\ref{thm:ChFP-ODEmeetsAggCTMC}, for all $t$,
$\hat{X}_n(t)$ converges to $X(t)$ almost surely as $n$ goes to
infinity. Notice that $\hat{X}_n(t)$ is bounded (see the discussion
in the previous section), then by Lebesgue's dominant convergence
theorem, we have
$$\lim_{n\rightarrow\infinity}E\|\hat{X}_n(t)-X(t)\|=0.$$
Since a norm $\|\cdot\|$ can be considered as a convex function, by
Jensen's inequality (Theorem~2.2
in~\cite{BasicStochasticProcesses}), we have $\|(E[\cdot])\|\leq
E[\|\cdot\|]$. Therefore,
\begin{equation*}
\begin{split}
        \lim_{n\rightarrow\infinity}\|\hat{M}_n(t)-X(t)\|
        &=\lim_{n\rightarrow\infinity}\|E[\hat{X}_n(t)]-E[X(t)]\|\\
        &\leq
        \lim_{n\rightarrow\infinity}E\|\hat{X}_n(t)-X(t)\|\\
        %&=E\lim_{n\rightarrow\infinity}|X_n(t)-X(t)|\; (\mbox{Lebesgue Thm})\\
        &=0.
\end{split}
\end{equation*}
\end{proof}

%\begin{remark}
%By Remark~\ref{remark:X_n/nBounded}, we know $\{\hat{X}_n(t)\}_n$ is
%bounded. So $\{\hat{M}_n(t)\}$ is bounded, which implies that $X(t)$
%is bounded. The nonnegativeness of $X(t)$ can also be obtained by a
%similar argument.
%\end{remark}

\par
Lemma~\ref{lem:Mn(t)=X(t)} states that the ODEs' solution $X(t)$ is
just the limit function of the sequence of the expectation functions
of the corresponding density dependent Markov chains. This provides
some clues: the characteristics of the limit $X(t)$ depend on the
properties of $\{\hat{M}_{n}(t)\}_n$. Therefore, we expect to be
 able to investigate $X(t)$ by studying $\{\hat{M}_{n}(t)\}_n$.

\par  Since $\hat{M}_{n}(t)$ is the expectation of the
Markov chain $\hat{X}_n(t)$, $\hat{M}_{n}(t)$ can be expressed by a
formula in which the transient probability distribution is involved.
That is,
$$
\hat{M}_n(t)=E[\hat{X}_n(t)]=\sum_{\mathbf{x}\in
\hat{S}_n}\mathbf{x}\hat{\mathbf{\pi}}^n_t(\mathbf{x}),
$$
where $\hat{S}_n$ is the state space,
$\hat{\mathbf{\pi}}^n_t(\cdot)$ is the probability distribution of
$\hat{X}_n$ at time $t$. Let ${S}_n$ and $\mathbf{\pi}^n_t(\cdot)$
be the state space and the probability distribution of $X_n(t)$
respectively\footnote{We should point out that the probability
distributions of $X_n(t)$ and $\hat{X}_n(t)$ are the same, i.e.\
$\mathbf{\pi}^n_t(\mathbf{x})=\hat{\mathbf{\pi}}^n_t(\mathbf{x}/n)$.}.
Then
\begin{equation*}
    \hat{M}_n(t)=E[\hat{X}_n(t)]=E\left[\frac{X_n(t)}{n}\right]
    =\sum_{\mathbf{x}\in S^n}\frac{\mathbf{x}}{n}\mathbf{\pi}^n_t(\mathbf{x}).
\end{equation*}
We have assumed the Markov chains underlying PEPA models to be
irreducible and positive-recurrent. Then the transient probability
distributions of these Markov chains will converge to the
corresponding steady-state probability distributions. We denote the
steady-state probability distributions of $X_n(t)$ and
$\hat{X}_n(t)$ as $\mathbf{\pi}^n_{\infinity}(\cdot)$ and
$\hat{\mathbf{\pi}}^n_{\infinity}(\cdot)$ respectively. Then, we
have a lemma.
\begin{lemma}\label{lem:ChFP-Mn(t)Converges-Tconverges}
For any $n$, there exists a $\hat{M}_n(\infinity)$, such that
$$\lim_{t\rightarrow \infinity}\hat{M}_n(t)=\hat{M}_n(\infinity).$$
\end{lemma}

\begin{proof}
\begin{equation*}
\begin{split}
        &\lim_{t\rightarrow\infinity}\hat{M}_n(t)\\
        =&\lim_{t\rightarrow\infinity}\sum_{\mathbf{x}\in {S}_n}\frac{\mathbf{x}}{n}\mathbf{\pi}^n_t(\mathbf{x})
         =\sum_{\mathbf{x}\in {S}_n}\lim_{t\rightarrow\infinity}\frac{\mathbf{x}}{n}\mathbf{\pi}^n_t(\mathbf{x})
        =\sum_{\mathbf{x}\in {S}_n}\frac{\mathbf{x}}{n}\mathbf{\pi}^n_{\infinity}(\mathbf{x})
        \equiv \hat{M}_n(\infinity).
\end{split}
\end{equation*}
\end{proof}
Clearly, we also have $\hat{M}_n(\infinity)=\sum_{\mathbf{x}\in
\hat{S}_n}\mathbf{x}\hat{\mathbf{\pi}}^n_{\infinity}(\mathbf{x})$.

\begin{remark}\label{remark:M_{nk}limits}
Currently, we do not know whether the sequence
$\{\hat{M}_n(\infinity)\}_n$ converges as $n\rightarrow \infinity$.
But since $\{\hat{M}_n(\infinity)\}_n$ is bounded which is due to
the conservation law that PEPA models satisfy, there exists
$\{n'\}\subset\{n\}$ such that $\{\hat{M}_{n'}(\infinity)\}$
converges to a limit, namely $\hat{M}_{\infinity}(\infinity)$. That
is
$$
\lim_{n'\rightarrow
\infinity}\hat{M}_{n'}(\infinity)=\hat{M}_{\infinity}(\infinity).
$$
Thus,
\begin{equation}\label{eq:lim_nlim_t}
 \lim_{n'\rightarrow
\infinity}\lim_{t\rightarrow
\infinity}\hat{M}_{n'}(t)=\lim_{n'\rightarrow
\infinity}\hat{M}_{n'}(\infinity)=\hat{M}_{\infinity}(\infinity).
\end{equation}
\end{remark}

At the moment, there are two questions:
\begin{enumerate}
  \item  Whether $\lim_{t\rightarrow
\infinity}\lim_{n'\rightarrow \infinity}\hat{M}_{n'}(t)$ exists?
  \item  If $\lim_{t\rightarrow
\infinity}\lim_{n'\rightarrow \infinity}\hat{M}_{n'}(t)$ exists,
whether
$$ \lim_{t\rightarrow \infinity}\lim_{n'\rightarrow
\infinity}\hat{M}_{n'}(t)=\lim_{n'\rightarrow
\infinity}\lim_{t\rightarrow \infinity}\hat{M}_{n'}(t)?$$
\end{enumerate}
If the answer to the first question is yes, then the solution of the
ODEs converges, since by
Lemma~\ref{lem:ChFP-Mn(t)Converges-Tconverges},
$
\lim_{t\rightarrow \infinity}X(t)=\lim_{t\rightarrow
\infinity}\lim_{n'\rightarrow \infinity}\hat{M}_{n'}(t).
$
If the answer to the second question is yes, then the limit of
$X(t)$ is consistent with the stationary distributions of the Markov
chains since
$$
\lim_{t\rightarrow \infinity}X(t)=\lim_{t\rightarrow
\infinity}\lim_{n'\rightarrow
\infinity}\hat{M}_{n'}(t)=\lim_{n'\rightarrow
\infinity}\lim_{t\rightarrow
\infinity}\hat{M}_{n'}(t)=\hat{M}_{\infinity}(\infinity).
$$
In short, the positive answers to these two questions determine the
convergence and consistency for the ODEs's solution, see
Figure~\ref{fig:ChFP-InterchangeLimits}. Fortunately, the two
answers can be guaranteed by the condition
(\ref{eq:ParticularCondition2}) in the following
Proposition~\ref{lem:ParticularCondition}.
{\begin{figure}[htbp]\caption{Convergence and consistency diagram
for derived ODEs}
 \label{fig:ChFP-InterchangeLimits} \large\begin{equation*}
   \begin{CD}
            \hat{M}_{n'}(t) @> n'\rightarrow \infinity >\text{ \quad Lemma~\ref{lem:Mn(t)=X(t)} \quad } >
            X(t)(=\hat{M}_{\infinity}(t))\\
               @V t\rightarrow\infinity  V \text{Lemma~\ref{lem:ChFP-Mn(t)Converges-Tconverges}} V
              @V {\Huge\color{red}???} V t\rightarrow \infinity V \\
         \hat{M}_{n'}(\infinity)
         @>\text{\quad Remark~\ref{remark:M_{nk}limits} \quad }> n'\rightarrow\infinity >
         \hat{M}_{\infinity}(\infinity)
   \end{CD}
\end{equation*}
\end{figure}
}

\begin{proposition}\label{lem:ParticularCondition}(\textbf{A particular condition})
If there exist $A, B>0$, such that
 \begin{equation}\label{eq:ParticularCondition2}
            \sup_{n'}\left\|\hat{M}_{n'}(t)
             -\hat{M}_{n'}(\infinity)\right\|<Be^{-At},
\end{equation}
then $\lim_{t\rightarrow
\infinity}X(t)=\hat{M}_{\infinity}(\infinity)$.
\end{proposition}

\begin{proof}
 \begin{eqnarray*}
  %\begin{split}
           \left\|X(t)-\hat{M}_{\infinity}(\infinity)\right\|
          &=&\left\|\lim_{n'\rightarrow\infinity}[\hat{M}_{n'}(t)-\hat{M}_{n'}(\infinity)]\right\|\\
          &\leq&\limsup_{n'\rightarrow\infinity}\left\|\hat{M}_{n'}(t)-\hat{M}_{n'}(\infinity)\right\|\\
          &\leq&\limsup_{n'\rightarrow\infinity}\left[\sup_{n'}\left\|\hat{M}_{n'}(t)-\hat{M}_{n'}(\infinity)\right\|\right]\\
          &\leq&\limsup_{n'\rightarrow\infinity}Be^{-At}\\
          &=& Be^{-At}\longrightarrow 0, \mbox{as }
          t\longrightarrow \infinity.
  %\end{split}
  \end{eqnarray*}
  So $\lim_{t\rightarrow \infinity}X(t)=\hat{M}_{\infinity}(\infinity)$.
\end{proof}

Notice that $\displaystyle \hat{M}_n(t)=\sum_{\mathbf{x}\in
S^n}\frac{\mathbf{x}}{n}\mathbf{\pi}^n_t(\mathbf{x})$ and
$\displaystyle \hat{M}_n(\infinity)=\sum_{\mathbf{x}\in
S^n}\frac{\mathbf{x}}{n}\mathbf{\pi}^n_{\infinity}(\mathbf{x})$, so
in order to estimate
$\left\|\hat{M}_{n}(t)-\hat{M}_{n}(\infinity)\right\|$ in
(\ref{eq:ParticularCondition2}), we need first to estimate the
difference between $\mathbf{\pi}^n_t$ and
$\mathbf{\pi}^n_{\infinity}$.

\begin{lemma}\label{lem:ChFP-PartCond-HoldCond}
If there exists $A>0$ and $B_1>0$, such that for any $n'$ and all
$\mathbf{x}\in S^{n'}$,
\begin{equation}\label{eq:ParticularCondition3}
           |\mathbf{\pi}^{n'}_t(\mathbf{x})-\mathbf{\pi}^{n'}_{\infinity}(\mathbf{x})|\leq
          \mathbf{\pi}^{n'}_{\infinity}(\mathbf{x}) B_1e^{-At},
\end{equation}
then there exists $B>0$ such that $\sup_{n'}\left\|\hat{M}_{n'}(t)
             -\hat{M}_{n'}(\infinity)\right\|<Be^{-At}$ holds.
\end{lemma}
\begin{proof} We know that $\displaystyle\hat{X}_n(0)=\frac{{X}_n(0)}{n}=\mathbf{x}_0$ for
any $n$. By the conservation law, the population of each entity in
any state is determined by the starting state. So for any $n'$ and
all $\mathbf{x}\in S^{n'}$, $||\mathbf{x}/{n'}||\leq C_1\sum_{P\in
\mathcal{D}}\mathbf{x}_0[P]<\infinity$, where $\mathcal{D}$ is the
set of all local derivatives and $C_1$ is a constant independent of
$n'$. Let $C=\sup_{n'}\max_{\mathbf{x}\in
S^{n'}}||\mathbf{x}/{n'}||$, then $C<\infinity$.
 \begin{eqnarray*}
  %\begin{split}
          \left\|\hat{M}_{n'}(t)-\hat{M}_{n'}(\infinity)\right\|
          &=&\left\|\sum_{\mathbf{x}\in S^{n'}}\frac{\mathbf{x}}{n'}\mathbf{\pi}^{n'}_t(\mathbf{x})
          -\sum_{\mathbf{x}\in S^{n'}}\frac{\mathbf{x}}{n'}\mathbf{\pi}^{n'}_{\infinity}(\mathbf{x})\right\|\\
          &=&\left\|\sum_{\mathbf{x}\in S^{n'}}\frac{\mathbf{x}}{n'}(\mathbf{\pi}^{n'}_t(\mathbf{x})
          -\mathbf{\pi}^{n'}_{\infinity}(\mathbf{x}))\right\|\\
          &\leq&\sup_{n'}\max_{\mathbf{x}\in S^{n'}}\left\|\frac{\mathbf{x}}{n'}\right\|
          \sum_{\mathbf{x}\in S^{n'}}|\mathbf{\pi}^{n'}_t(\mathbf{x})-\mathbf{\pi}^{n'}_{\infinity}(\mathbf{x})|\\
          &=&C\sum_{\mathbf{x}\in S^{n'}}|\mathbf{\pi}^{n'}_t(\mathbf{x})-\pi^{n'}_{\infinity}(\mathbf{x})|\\
          &\leq& C\sum_{\mathbf{x}\in S^{n'}}\mathbf{\pi}^{n'}_{\infinity}(\mathbf{x})B_1e^{-tA}\\
          &=& CB_1e^{-tA}.
  %\end{split}
  \end{eqnarray*}
Let $B=CB_1$. Then $\sup_{n'}\left\|\hat{M}_{n'}(t)
             -\hat{M}_{n'}(\infinity)\right\|<Be^{-At}$.
\end{proof}
%In the following, we will investigate the particular condition
%(\ref{eq:ParticularCondition2}). An important related estimation is
%given first in the next section.

\subsection{Investigation of the particular condition} This section
will present the study of the particular condition
(\ref{eq:ParticularCondition2}). We will expose that the condition
is related to well-known constants of Markov chains such as the
spectral gap and the Log-Sobolev constant. The methods and results
developed in the field of functional analysis of Markov chains are
utilised to investigate the condition.

\subsubsection{Important estimation for Markov
kernel} We first give an estimation for the Markov kernel which is
defined below. Let $Q$ be the infinitesimal generator of a Markov
chain $X$ on a finite state $S$. Let
$$
   K_{ij}=\left\{
                  \begin{array}{cc}
                    \frac{Q_{ij}}{m}, & i\neq j \\
                    1+ \frac{Q_{ii}}{m}, & i=j \\
                  \end{array}
                \right.\quad \mbox{where } m=\sup_i(-Q_{ii}).
$$
$K$ is a transition probability matrix, satisfying
$$ K(\mathbf{x},\mathbf{y})\geq 0,\quad\quad \sum_\mathbf{y}K(\mathbf{x},\mathbf{y})=1.
$$
$K$ is called an Markov \emph{kernel}  ($K$ is also called the
\emph{uniformisation} of the CTMC in some literature). A Markov
chain on a finite state space $S$ can be described through its
\emph{kernel} $K$. The continuous time semigroup associated with $K$
is defined by
$$
H_t=\exp(-t(I-K)).
$$
Let $\mathbf{\pi}$ be the unique stationary measure of the Markov
chain. Then $H_t(\mathbf{x},\mathbf{y})\rightarrow
\mathbf{\pi}(\mathbf{y})$ as $t$ tends to infinity. Following the
convention in the literature we will also use $(K,\mathbf{\pi})$ to
represent a Markov chain.

\par Notice
$$Q=m(K-I),\quad \quad K=\frac{Q}{m}+I.$$
Clearly,
$$
  P_t=e^{tQ}=e^{mt(K-I)}=e^{-mt(I-K)}=H_{mt},%\left(e^{-t(I-K)}\right)^m=H_{mt},
$$
and thus $
  H_{t}=P_{\frac tm}
$, where $P_t$ is  called the semigroup associated with the
infinitesimal generator $Q$. An estimation of $H_t$ is given below.
\begin{lemma}(Corollary 2.2.6, \cite{SCoste-FiniteMarkovChains}) Let $(K,\mathbf{\pi})$ be a
finite Markov chain, and
\linebreak$\displaystyle\mathbf{\pi}(*)=\min_{\mathbf{x}\in
S}{\mathbf{\pi}(\mathbf{x})}$. Then
\begin{equation}\label{eq:Corollary226}
\sup_{\mathbf{x},\mathbf{y}}\left|\frac{H_t(\mathbf{x},\mathbf{y})}{\mathbf{\pi}(\mathbf{y})}-1\right|\leq
e^{2-c}\;\quad \mbox{for
}\;t=\frac{1}{\alpha}\log\log\frac{1}{\mathbf{\pi}(*)}+\frac{c}{\lambda},
\end{equation}
where $\lambda>0,\alpha>0$ are the spectral gap and the Log-Sobolev
constant respectively, which are defined  and interpreted
in~\ref{section:Appendix-spectral-gap-log-sobolev}.
\end{lemma}

It can be implied from (\ref{eq:Corollary226}) that $\forall
\mathbf{x},\mathbf{y}\in S$,
\begin{equation}\label{eq:Corollary226-1}
\left|{H_t(\mathbf{x},\mathbf{y})}-{\mathbf{\pi}(\mathbf{y})}\right|\leq
\mathbf{\pi}(\mathbf{y})\left[e^{2+\frac{\lambda}{\alpha}\log\log\frac{1}{\mathbf{\pi}(*)}}\right]e^{-\lambda
t}.
\end{equation}
Since $H_{t}=P_{\frac tm}$, so
\begin{equation}\label{eq:Corollary226-2}
\left|{P_{\frac{t}{m}}(\mathbf{x},\mathbf{y})}-{\mathbf{\pi}(\mathbf{y})}\right|\leq
\mathbf{\pi}(\mathbf{y})\left[e^{2+\frac{\lambda}{\alpha}\log\log\frac{1}{\mathbf{\pi}(*)}}\right]e^{-\lambda
t},
\end{equation}
and thus (replacing ``$t$'' by ``$mt$'' on the both sides of
(\ref{eq:Corollary226-2})),
\begin{equation}\label{eq:Corollary226-3}
\left|{P_{t}(\mathbf{x},\mathbf{y})}-{\mathbf{\pi}(\mathbf{y})}\right|\leq
\mathbf{\pi}(\mathbf{y})\left[e^{2+\frac{\lambda}{\alpha}\log\log\frac{1}{\mathbf{\pi}(*)}}\right]e^{-m\lambda
t}.
\end{equation}
The investigation and utilisation of the formulae
(\ref{eq:Corollary226-3}) will be presented in the next subsection.

\subsubsection{Study of the particular
condition}\label{subsection:ChFP-OpenProblem} For each $n$, let
$Q^n$ be the infinitesimal generator of the density dependent Markov
chain $X_n(t)$ underlying a given PEPA model and thus the transition
probability matrix is $P^n_t=e^{tQ^n}$. For each $X_n(t)$, the
initial state of the corresponding system is $n\mathbf{x}_0$, so the
initial probability distribution of $X_n(t)$ is
$\mathbf{\pi}^n_0(n\mathbf{x}_0)=1$ and
$\mathbf{\pi}^n_0(\mathbf{x})=0,\;\forall \mathbf{x}\in {S^n},
\mathbf{x}\neq n\mathbf{x}_0$. So $\forall \mathbf{y}\in {S^n}$,
\begin{equation}\label{eq:Pt=pi}
\begin{split}
        \mathbf{\pi}^n_t(\mathbf{y})=\sum_{\mathbf{x}\in
        S^n}\mathbf{\pi}^n_0(\mathbf{x})P^n_t(\mathbf{x},\mathbf{y})
        =\mathbf{\pi}^n_0(n\mathbf{x}_0)P^n_t(n\mathbf{x}_0,\mathbf{y})=P^n_t(n\mathbf{x}_0,\mathbf{y}).
\end{split}
\end{equation}

\par The formula (\ref{eq:Corollary226-3}) in the context of
$X_n(t)$ is
\begin{equation}\label{eq:1}
\left|P^n_{t}(\mathbf{x},\mathbf{y})-\mathbf{\pi}^n_{\infinity}(\mathbf{y})\right|\leq
\mathbf{\pi}^n_{\infinity}(\mathbf{y})\left[e^{2+
\frac{\lambda_n}{\alpha_n}\log\log\frac{1}{\mathbf{\pi}^n_{\infinity}(*)}}\right]e^{-m_n\lambda_n
t},
\end{equation}
where $\lambda_n, \lambda_n, m_n, \mathbf{\pi}^n_{\infinity}(*)$ are
the respective parameters associated with $X_n(t)$.

\par Let $\mathbf{x}=n\mathbf{x}_0$ in (\ref{eq:1}), and notice
$P^n_t(n\mathbf{x}_0,\mathbf{y})=\mathbf{\pi}^n_t(\mathbf{y})$ by
(\ref{eq:Pt=pi}), then we have
\begin{equation}\label{eq:Corollary226-4}
\left|{\mathbf{\pi}^n_{t}(\mathbf{y})}-{\mathbf{\pi}^n_{\infinity}(\mathbf{y})}\right|\leq
\mathbf{\pi}^n_{\infinity}(\mathbf{y})\left[e^{2
+\frac{\lambda_n}{\alpha_n}\log\log\frac{1}{\mathbf{\pi}^n_{\infinity}(*)}}\right]e^{-m_n\lambda_n
t},\quad\quad \forall \mathbf{y}\in S^n.
\end{equation}

From the comparison of (\ref{eq:Corollary226-4}) and
(\ref{eq:ParticularCondition3}), we know that if there are some
conditions imposed on
$\frac{\lambda_n}{\alpha_n}\log\log\frac{1}{\mathbf{\pi}^n_{\infinity}(*)}$
and $-m_n\lambda_n t$, then (\ref{eq:ParticularCondition3}) can be
induced from (\ref{eq:Corollary226-4}). See the following Lemma.
\begin{lemma}\label{lem:TwoConditions}
If there exists $T>0, B_2>0, A>0$ such that
\begin{equation}\label{eq:Condition1}
\sup_{n}\left\{-m_n\lambda_nT+\frac{\lambda_n}{\alpha_n}\log\log\frac{1}{\mathbf{\pi}^{n}_{\infinity}(*)}\right\}\leq
B_2,
\end{equation}
and
\begin{equation}\label{eq:Condition2}
\inf_{n}\{m_n\lambda_n\}\geq A,
\end{equation}
then
\begin{equation}\label{eq:ParticularCondition3-1}
           |\mathbf{\pi}^{n}_t(\mathbf{x})-\mathbf{\pi}^{n}_{\infinity}(\mathbf{x})|\leq
          \mathbf{\pi}^{n}_{\infinity}(\mathbf{x}) B_1e^{-At},\quad\quad \forall
          \mathbf{x}\in S^{n},
\end{equation}
where $B_1=e^{A_1T+B_2+2}$, and the ``particular condition''
(\ref{eq:ParticularCondition2}) holds.
\end{lemma}

\begin{proof}
\begin{equation*}
\begin{split}
\left[e^{2+\frac{\lambda_n}{\alpha_n}\log\log\frac{1}{\mathbf{\pi}^n_{\infinity}(*)}}\right]e^{-m_n\lambda_nt
}&=\left[e^{-m_n\alpha_nT+2+\frac{\lambda_n}{\alpha_n}\log\log\frac{1}{\mathbf{\pi}^n_{\infinity}(*)}}\right]e^{-m_n\lambda_n(t-T)
}\\
&\leq e^{B_2+2}e^{-A(t-T)}\\
&=B_1e^{-At}.
\end{split}
\end{equation*}
Then by (\ref{eq:Corollary226-4}), (\ref{eq:ParticularCondition3-1})
holds. Thus by Lemma~\ref{lem:ChFP-PartCond-HoldCond},
(\ref{eq:ParticularCondition2}) holds.
\end{proof}

\begin{remark}\label{remark:Tmlambda>loglog}
When $n$ tends to infinity, the discrete space is approximating a
continuous space and thus the probability of any single point is
tending to $0$, that is, $\mathbf{\pi}^n_{\infinity}(*)$ tends to
$0$. So $\log\log\frac{1}{\mathbf{\pi}^n_{\infinity}(*)}\rightarrow
\infinity$ as
$n$ goes to infinity. Notice that by Lemma~\ref{Lemma C.2.1} %\ref{lem:RelationGammaAlpha}
in~\ref{section:Appendix-spectral-gap-log-sobolev},
$$
  \frac{\lambda_n}{\alpha_n}\leq \log(1/\mathbf{\pi}^n_{\infinity}(*)).
 $$
Therefore, in order to have
$$-m_n\lambda_nT+\frac{\lambda_n}{\alpha_n}\log\log\frac{1}{\mathbf{\pi}^{n}_{\infinity}(*)}\leq
B_2,$$ it is sufficient to let
\begin{equation}\label{eq:Condition4}
 Tm_n\lambda_n\geq O([\log(1/\mathbf{\pi}^n_{\infinity}(*))]^2).
\end{equation}
Moreover, (\ref{eq:Condition4}) can imply  both
(\ref{eq:Condition1}) and (\ref{eq:Condition2}).
\end{remark}

\par According to the above analysis, our problem is simplified to
checking that whether (\ref{eq:Condition4}) is satisfied by the
density dependent Markov chains $\{X_n(t)\}$.

\par By Remark~\ref{Remark C.2.2} in~\ref{section:Appendix-spectral-gap-log-sobolev},  %Remark~\ref{remark:EigenvalueI-K} in Appendix,
$\lambda_n$ is the smallest non-zero eigenvalue of
$I-\frac{K^n+K^{*,n}}{2}$, where
$$
   K^n=\frac{Q^n}{m_n}+I
$$ and $K^{*,n}$ is adjoint to $K^n$. A matrix $Q^{*,n}$ is said to
be adjoint to the generator $Q^{n}$, if
$Q^{n}(\mathbf{x},\mathbf{y})\mathbf{\pi}^n_{\infinity}(\mathbf{x})$
equals
$Q^{*,n}(\mathbf{y},\mathbf{x})\mathbf{\pi}^n_{\infinity}(\mathbf{y})$.
Clearly, $Q^{*,n}=m_n(K^{*,n}-I)$, or equivalently,
$$
  K^{*,n}=\frac{Q^{*,n}}{m_n}+I.
$$
So
\begin{equation}\label{eq:T2}
   m_n\left(I-\frac{K^n+K^{*,n}}{2}\right)=-\frac{Q^n+Q^{*,n}}{2}.
\end{equation}
Denote the smallest non-zero eigenvalue of $-\frac{Q^n+Q^{*,n}}{2}$
by $\sigma_n$. Then by (\ref{eq:T2}),
\begin{equation}\label{eq:mlambda=sigma}
   m_n\lambda_n=\sigma_n.
\end{equation}

Now we state our main result in this section.
\begin{theorem}\label{thm:ChFP-ConvergenceMildCondition}
Let $\{X_n(t)\}$ be the density dependent Markov chain derived from
a given PEPA model. For each $n\in \mathbb{N}$, let $S^n$ and
$\mathbf{\pi}^{n}_{\infinity}$ be the state space and steady-state
probability distribution of $X_n(t)$ respectively. $Q^n$ is the
infinitesimal generator of $X_n(t)$ and $\sigma_n$ is the smallest
non-zero eigenvalue of $-\frac{Q^n+Q^{*,n}}{2}$, where $Q^{*,n}$ is
adjoint to $Q^n$ in terms of $\mathbf{\pi}^n_{\infinity}$. If
\begin{equation}\label{eq:ChFP-ThmMildCondition}
     \mathbf{\pi}^{n}_{\infinity}(*)\rmdef\min_{\mathbf{x}\in S^n}\mathbf{\pi}^{n}_{\infinity}(\mathbf{x})
     \geq \frac{1}{\exp({O(\sqrt{\sigma_n}))}}
\end{equation}
for sufficiently large $n$, then $X(t)$ has a finite limit as time
tends to infinity, where $X(t)$ is the solution of the corresponding
derived ODEs from the same PEPA model.
\end{theorem}

\begin{proof} By the given condition of $
     \mathbf{\pi}^{n}_{\infinity}(*)\geq \frac{1}{\exp(O(\sqrt{\sigma_n}))},
$
$$
    \log\left[\frac{1}{\mathbf{\pi}^{n}_{\infinity}(*)}\right]\leq
    \log[\exp(O(\sqrt{\sigma_n}))]=O(\sigma_n^{1/2}).
$$
Thus
$$\left(\log\left[\frac{1}{\mathbf{\pi}^{n}_{\infinity}(*)}\right]\right)^2\leq O(\sigma_n).$$
Choose a large $T$ such that
$$
  Tm_n\lambda_n=T\sigma_n\geq O(\sigma_n)\geq\left(\log\left[\frac{1}{\mathbf{\pi}^{n}_{\infinity}(*)}\right]\right)^2.
$$
By Remark~\ref{remark:Tmlambda>loglog} and
Lemma~\ref{lem:TwoConditions}, the particular condition holds.
Therefore
$$\lim_{t\rightarrow\infinity}X(t)=\hat{M}_{\infinity}(\infinity).$$
\end{proof}

According to the above theorem, our problem is simplified to
checking that whether (\ref{eq:ChFP-ThmMildCondition}) is satisfied
by the density dependent Markov chains $\{X_n(t)\}$. In
(\ref{eq:ChFP-ThmMildCondition}) both the spectral gap
$\lambda_n\;(=\frac{\sigma_n}{m_n})$  and
$\mathbf{\pi}^{n}_{\infinity}(*)$ are unknown. In fact, due to the
state space explosion problem, $\mathbf{\pi}^n_{\infinity}$ cannot
be easily solved from
 $\mathbf{\pi}^n_{\infinity}Q^n=0$ or equivalently
$\mathbf{\pi}^n_{\infinity}K^n=\mathbf{\pi}^n_{\infinity}$.
Moreover, the estimation of the spectral gap for a given Markov
chain in current literature (e.g.\ \cite{SCoste-FiniteMarkovChains})
is heavily based on the known stationary distribution. Thus, the
current results cannot provide a practical check for
(\ref{eq:ChFP-ThmMildCondition}).

\par The convergence and consistency are supported by many numerical
experiments, so we believe that (\ref{eq:ChFP-ThmMildCondition}) is
unnecessary. In other words, we believe that
(\ref{eq:ChFP-ThmMildCondition}) always holds in the context of
PEPA, although at this moment we cannot prove it. Before concluding
this section, we leave an open problem:
\begin{conjecture}
The formula (\ref{eq:ChFP-ThmMildCondition}) in
Theorem~\ref{thm:ChFP-ConvergenceMildCondition} holds for any PEPA
model.
\end{conjecture}

\section{Fluid analysis: an analytic approach (I)}

The previous sections have demonstrated the fluid approximation and
relevant analysis for PEPA.  Some fundamental results about the
derived ODEs such as the boundedness, nonnegativeness and
convergence of the solutions, have been established through a
probabilistic approach. In this section we will discuss the
boundedness and nonnegativeness again, and prove them by a purely
analytical argument. The convergence presented in the previous
section is proved under a particular condition that cannot currently
be easily checked. This section will present alternative approaches
to deal with the convergence problem. In particular, for an
interesting model with two synchronisations, its structural
invariance as revealed in~\cite{JieThesis}, will be shown to play an
important role in the proof of the convergence. Moreover, for a
class of PEPA models which have two component types and one
synchronisation, an analytical proof of the convergence under some
mild conditions on the populations will be presented. These
discussions and investigations will provide a new insight into the
fluid approximation of PEPA.

\subsection{Analytical proof of boundedness and
nonnegativeness}\label{section:Analytic-Proof-Boundedness}

Recall that the set of derived ODEs from a general PEPA model is
\begin{equation}\label{eq:ChFA-ODEsDerivativeCentric}
\begin{split}
\frac{\mbox{d}\mathbf{x}(U,t)}{\mbox{d}t} =-\sum_{\{l\mid
U\in\mbox{pre}(l)\}}\!\!\!\!f(\mathbf{x},l)+\sum_{\{l\mid
U\in\mbox{post}(l)\}}\hspace{-3mm}f(\mathbf{x},l).
\end{split}
\end{equation}
As mentioned, in this formula the term $\sum_{\{l\mid
U\in\mbox{pre}(l)\}}\hspace{-0mm}f(\mathbf{x},l)$ represents the
exit rates in the local derivative $U$. An important fact to note
is: the exit rates in a local derivative $C_{i_j}$ in state
$\mathbf{x}$ are bounded by all the apparent rates in this local
derivative. In fact, according to
Proposition~\ref{proposition:Ch3-RateFunctionComparison} in
Section~\ref{subsec:Background-NumericalRepresentation}, if
$C_{i_j}\in \mbox{pre}(l)$ where $l$ is a labelled activity, then
the transition rate function $f(\mathbf{x},l)$ is bounded by
$r_l(\mathbf{x},C_{i_j})$, the apparent rates of $l$ in $C_{i_j}$ in
state $\mathbf{x}$. That is,
\begin{equation}\label{eq:ChFA-local-section1-1}
f(\mathbf{x},l)\leq r_l(\mathbf{x},C_{i_j})=
\mathbf{x}[C_{i_j}]r_l(C_{i_j}).
\end{equation}
We should point out that (\ref{eq:ChFA-local-section1-1}) is based
on the discrete state space underlying the given model. According to
our semantics of mapping PEPA model to ODEs, the fluid
approximation-version of (\ref{eq:ChFA-local-section1-1}) also
holds, i.e. $f(\mathbf{x}(t),l)\leq
\mathbf{x}(C_{i_j},t)r_l(C_{i_j})$. Hereafter the notation
$\mathbf{x}[\cdot]$ indicates a discrete state $\mathbf{x}$, while
$\mathbf{x}(\cdot)$ or $\mathbf{x}(\cdot,t)$ reflects a continuous
state $\mathbf{x}$ at time $t$. Therefore, we have the following

\begin{proposition}\label{pro:ChFA-f(x,l)-Bounded} For any local derivative $C_{i_j}$,
\begin{equation}
    \sum_{\{l\mid C_{i_j}\in\mbox{pre}(l)\}}\hspace{-3mm}f(\mathbf{x}(t),l)\leq
    \mathbf{x}(C_{i_j},t)\sum_{\{l\mid C_{i_j}\in\mbox{pre}(l)\}}\hspace{-3mm}
    r_l(C_{i_j}),
\end{equation}
where $r_l(C_{i_j})$ is the apparent rate of $l$ in $C_{i_j}$ for a
single instance of $C_{i_j}$ defined in
Definition~\ref{def:DingApparentRate}.
\end{proposition}

\par Proposition~\ref{pro:ChFA-f(x,l)-Bounded} and
Propositions~\ref{pro:ChFP-ODEsConservationLaw} can guarantee the
boundedness and nonnegativeness of the solutions.
 In the following,
we will present an analytical proof of these properties, based on
the two propositions. Suppose the initial values
$\mathbf{x}\left(C_{i_j},0\right)$ are given, and we denote
$\sum_j\mathbf{x}\left(C_{i_j},0\right)$ by $N_{C_i}$. We have a
theorem:
\begin{theorem}\label{thm:ChFA-AnalyticalProof-BoundedSolution}
 If $\mathbf{x}\left(C_{i_j},t\right)$ satisfies
(\ref{eq:ChFA-ODEsDerivativeCentric}) with nonnegative initial
values, then
\begin{equation}
         0\leq \mathbf{x}\left(C_{i_j},t\right)\leq N_{C_i}, \quad \mbox{for any }t\geq 0.
\end{equation}
Moreover, if the initial values are positive, then the solutions are
always positive, i.e.,
\begin{equation}
         0< \mathbf{x}\left(C_{i_j},t\right)< N_{C_i}, \quad  \mbox{for any }t\geq 0.
\end{equation}

\end{theorem}

\begin{proof}
By Proposition~\ref{pro:ChFP-ODEsConservationLaw},
$\sum_j\mathbf{x}\left(C_{i_j},t\right)=N_{C_i}$ for all $t$. All
that is left to do is to prove that $\mathbf{x}(C_{i_j},t)$ is
positive or nonnegative. The proof is divided into two cases.

\par Case 1: Suppose all the initial values are positive,
 i.e. $\mathop{\min}_{i_j}\left\{\mathbf{x}(C_{i_j},0)\right\}>0$.
 We will show that $\mathop{\min}_{i_j}\left\{\mathbf{x}(C_{i_j},t)\right\}>0$ for all $t\geq 0$.
Otherwise, if there exists a $t>0$ such that
$\mathop{\min}_{i_j}\left\{\mathbf{x}(C_{i_j},t)\right\}\leq0$, then
there  exists a point $t'>0$ such that
 $\mathop{\min}_{i_j}\left\{\mathbf{x}(C_{i_j},t')\right\}=0$.
Let $t^{*}$ be the first such point, i.e.
$$
   t^{*}=\inf\left\{t>0\mid
         \mathop{\min}_{i_j}\left\{\mathbf{x}(C_{i_j},t)\right\}=0
   \right\},
$$
then $0<t^{*}<\infinity$. Without loss of generality, we assume
$\mathbf{x}(C_{1_1},t)$ reaches zero at $t^{*}$, i.e.,
$$\mathbf{x}(C_{1_1},t^{*})=0,\quad
\mathbf{x}(C_{i_j},t^{*})\geq0 \;\;(i\neq 1 \vee j\neq 1)$$ and
$$
\mathbf{x}(C_{i_j},t)>0,\quad t\in [0,t^{*}),\quad \forall i,j.
$$
Thus, for $t\in [0,t^{*}]$, by
Proposition~\ref{pro:ChFA-f(x,l)-Bounded},
\begin{eqnarray*}
 \frac{\mathrm{d}\mathbf{x}(C_{1_1},t)}{\mathrm{d}t}&=&
 -\sum_{\{l\mid C_{1_1}\in\mbox{pre}(l)\}}\hspace{-3mm}f(\mathbf{x},l)
 +\sum_{\{l\mid C_{1_1}\in\mbox{post}(l)\}}\hspace{-3mm}f(\mathbf{x},l)\\
  &\geq&-\sum_{\{l\mid C_{1_1}\in\mbox{pre}(l)\}}\hspace{-3mm}f(\mathbf{x},l)\\
  &\geq&-\mathbf{x}(C_{1_1},t)\sum_{\{l\mid C_{1_1}\in\mbox{pre}(l)\}}\hspace{-3mm}
    r_l(C_{1_1}).
\end{eqnarray*}
Set $R=\sum_{\{l\mid C_{1_1}\in\mbox{pre}(l)\}}r_l(C_{1_1})$, then
 \begin{equation}\label{equ:local_1}
     \frac{\mathrm{d}\mathbf{x}(C_{1_1},t)}{\mathrm{d}t}\geq-R\mathbf{x}(C_{1_1},t).
 \end{equation}
By Lemma~\ref{Lemma:Differential-Inequality}
in~\ref{section:Appendix-Some-Theorems}, (\ref{equ:local_1}) implies
$$
     \mathbf{x}(C_{1_1},t^{*})\geq \mathbf{x}(C_{1_1},0)e^{-Rt^{*}}>0.
$$
This is a contradiction to $\mathbf{x}(C_{1_1},t^{*})=0$. Therefore
$0<\mathbf{x}\left(C_{i_j},t\right)$, and thus by
Proposition~\ref{pro:ChFP-ODEsConservationLaw},
$$0< \mathbf{x}\left(C_{i_j},t\right)< N_{C_0}, \quad \forall t.$$

\par Case 2: Suppose
$\mathop{\min}_{i_j}\left\{\mathbf{x}(C_{i_j},0)\right\}=0$. Let
$u_{\delta}(i_j,0)=\mathbf{x}(C_{i_j},0)+\delta$ where $\delta>0$.
Let $u_{\delta}(i_j,t)$ be the solution of
(\ref{eq:ChFA-ODEsDerivativeCentric}), given the initial value
$u_{\delta}(i_j,0)$. By the proof of Case 1,
$u_{\delta}(i_j,t)>0\;(\forall t\geq0)$. Noticing $\min(\cdot)$ is a
Lipschitz function, by the Fundamental Inequality
in~\ref{section:Appendix-Some-Theorems}, we have
\begin{equation}
     |u_{\delta}(i_j,t)-\mathbf{x}(C_{i_j},t)|\leq \delta e^{Kt},
\end{equation}
where $K$ is a constant. Thus, for any given $t\geq0$,
\begin{equation}\label{LocalEquation_1}
\mathbf{x}(C_{i_j},t)\geq u_{\delta}(i_j,t)-\delta e^{Kt}>-\delta
e^{Kt}.
\end{equation}
Let $\delta\downarrow0$ in (\ref{LocalEquation_1}), then we have
$\mathbf{x}(C_{i_j},t)\geq 0$. The proof is completed.
\end{proof}

\subsection{A case study on convergence with two synchronisations}

If a model has synchronisations, then the derived ODEs are
nonlinear. The nonlinearity results in the complexity of the dynamic
behaviour of fluid approximations. However, for some special models,
we can still determine the convergence of the solutions. What
follows is a case study for an interesting PEPA model, in which the
structural property of invariance will be shown to play an important
role in the proof of the convergence.

\subsubsection{Theoretical study of convergence}
\label{subsec:ChFA-Analytica-proof-InterestingModel}

\par The model considered here is given below,
\begin{model}\label{model:X-Y}
\begin{displaymath}
    \begin{array}{rcl}
    \mathit{X_1} & \rmdef & (\mathit{action1},\mathit{a_1}).\mathit{X_2}\\
        \mathit{X_2} & \rmdef & (\mathit{action2},\mathit{a_2}).\mathit{X_1}\\
        \mathit{Y_1} & \rmdef & (\mathit{action1},\mathit{a_1}).\mathit{Y_3}+(\mathit{job1},\mathit{c_1}).\mathit{Y_2}\\
        \mathit{Y_2} & \rmdef & (\mathit{job2},\mathit{c_2}).\mathit{Y_1}\\
        \mathit{Y_3} & \rmdef & (\mathit{job3},\mathit{c_3}).\mathit{Y_4}\\
        \mathit{Y_4} & \rmdef & (\mathit{action2},\mathit{a_2}).\mathit{Y_2}+(\mathit{job4},\mathit{c_4}).\mathit{Y_3}\\
%[0.0ex]
%\multicolumn{3}{l}
\mathit{\left(X_1[M_1]||X_2[M_2]\right)}&\sync{action1,action2}&
\mathit{\left(Y_1[N_1]||Y_2[N_2]||Y_3[N_3]||Y_4[N_4]\right).}
%[0.0ex]
\end{array}
\end{displaymath}
\end{model}

\begin{figure}[htbp]
 \begin{center} \includegraphics[width=10cm]{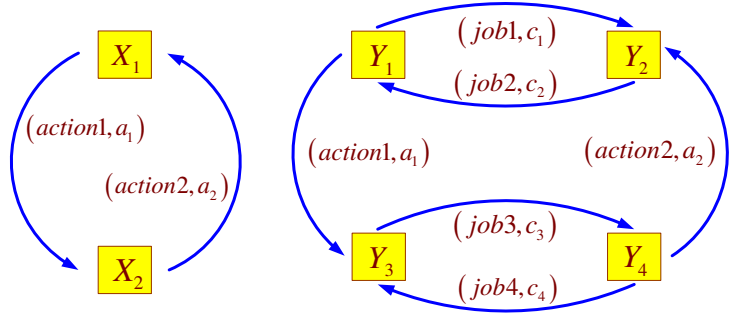}\\
 \end{center} \caption{Transition systems of the components of Model~\ref{model:X-Y}}\label{figure:ChFA-FlowChar-X-Y}
\end{figure}

The operations of $X$ and $Y$ are illustrated in
Figure~\ref{figure:ChFA-FlowChar-X-Y}. According to the mapping
semantics, the derived ODEs from this model are

\begin{equation}\label{eq:ChFA-Model-XY-ODEs}
\left\{
\begin{split}
\frac{\mathrm{d}x_1}{\mathrm{d}t}&=-a_1\min\{x_1,y_1\}+a_2\min\{x_2,y_4\}\\
\frac{\mathrm{d}x_2}{\mathrm{d}t}&=a_1\min\{x_1,y_1\}-a_2\min\{x_2,y_4\}\\
\frac{\mathrm{d}y_1}{\mathrm{d}t}&=-a_1\min\{x_1,y_1\}-c_1y_1+c_2y_2\\
\frac{\mathrm{d}y_2}{\mathrm{d}t}&=-c_2y_2+a_2\min\{x_2,y_4\}+c_1y_1\\
\frac{\mathrm{d}y_3}{\mathrm{d}t}&=-c_3y_3+c_4y_4+a_1\min\{x_1,y_1\}\\
\frac{\mathrm{d}y_4}{\mathrm{d}t}&=-a_2\min\{x_2,y_4\}-c_4y_4+c_3y_3
\end{split}\right.,
\end{equation}
where $x_i,y_j\;(i=1,2,\;j=1,2,\cdots,4)$ denote the populations of
$X$ and $Y$ in the local derivatives $X_i, Y_j$ respectively. Now we
state an interesting \textbf{assertion} for the specific PEPA model:
the difference between the number of $Y$ in their local derivatives
$Y_3$ and $Y_4$, and the number of $X$ in the local derivative
$X_2$, i.e. $y_3+y_4-x_2$, is a constant in any state. This fact can
be explained as follows. Notice that there is only one way to
increase $y_3+y_4$, i.e. enabling the activity $action1$. As long as
$action1$ is activated, then there is a copy of $Y$ entering the
$Y_3$ from $Y_1$. Meanwhile, since $action1$ is shared by $X$, a
corresponding copy of $X$ will go to $X_2$ from $X_1$. In other
words, $y_3+y_4$ and $x_2$ increase equally and simultaneously. On
the other hand, there is also only one way to decrease $y_3+y_4$ and
$x_2$, i.e. enabling the cooperated activity $action2$. This also
allows $y_3+y_4$ and $x_2$ to decrease both equally and
simultaneously. So, the difference $y_3+y_4-x_2$ will remain
constant in any state and thus at any time. The assertion indicates
that each state and therefore the whole state space of underlying
CTMC may have some interesting structure properties, such as
invariants. The techniques and applications of structural analysis
of PEPA models have been developed in~\cite{JieThesis}.

Throughout this section, let $M$ and $N$ be the total populations of
the $X$ and $Y$ respectively, i.e. $M=M_1+M_2$ and
$N=N_1+N_2+N_3+N_4$. Notice $y_1+y_2+y_3+y_4=N$ and $x_1+x_2=M$ by
the conservation law, so $y_1+y_2-x_1=N-M-(y_3+y_4-x_2)$ is another
invariant because $y_3+y_4-x_2$ is a constant.
 The fluid-approximation
version of these two invariants also holds, which is illustrated by
the following Lemma~\ref{lem:ChFA-ModelXY-Invariant}.
\begin{lemma}\label{lem:ChFA-ModelXY-Invariant} For any $t\geq0$,
\begin{equation*}
\begin{split}
     y_1(t)+y_2(t)-x_1(t)&=y_1(0)+y_2(0)-x_1(0),\\
   y_3(t)+y_4(t)-x_2(t)&=y_3(0)+y_4(0)-x_2(0).
\end{split}
\end{equation*}
\end{lemma}

\begin{proof} According to (\ref{eq:ChFA-Model-XY-ODEs}), for any $t\geq0$,
   $$\frac{\mathrm{d}(y_1(t)+y_2(t)-x_1(t))}{\mathrm{d}t}
   =\frac{\mathrm{d}y_1}{\mathrm{d}t}+\frac{\mathrm{d}y_2}{\mathrm{d}t}-\frac{\mathrm{d}x_1}{\mathrm{d}t}=0.$$
So $y_1(t)+y_2(t)-x_1(t)=y_1(0)+y_2(0)-x_1(0)$, $\forall t\geq0$. By
a similar argument, we also have
$y_3(t)+y_4(t)-x_2(t)=y_3(0)+y_4(0)-x_2(0)$, $\forall t\geq0$.
\end{proof}

\par In the following we show  how to use this kind of invariance to
prove the convergence of the solution of
(\ref{eq:ChFA-Model-XY-ODEs}) as time goes to infinity. Before
presenting the results and the proof, we first rewrite
(\ref{eq:ChFA-Model-XY-ODEs}) as follows:

\begin{equation}\label{eq:ChFA-Model-XY-ODEQ1Q2Q3Q4}
\begin{split}
%\begin{eqnarray*}
\left(
          \begin{array}{c}
            \frac{\mathrm{d}x_1}{\mathrm{d}t} \\
            \frac{\mathrm{d}x_2}{\mathrm{d}t} \\
            \frac{\mathrm{d}y_1}{\mathrm{d}t} \\
            \frac{\mathrm{d}y_2}{\mathrm{d}dt} \\
            \frac{\mathrm{d}y_3}{\mathrm{d}t} \\
            \frac{\mathrm{d}y_4}{\mathrm{d}t}
          \end{array}
 \right)=&I_{\{ x_1<y_1,
                x_2< y_4
           \}}Q_1
 \left(\begin{array}{c}
            x_1 \\
            x_2 \\
            y_1\\
            y_2\\
            y_3\\
            y_4
          \end{array}\right)
 +I_{\{x_1<y_1, x_2\geq y_4\}}Q_2 \left(\begin{array}{c}
            x_1 \\
            x_2 \\
            y_1\\
            y_2\\
            y_3\\
            y_4
          \end{array}\right)\\
&+I_{\{ x_1\geq y_1,
                x_2<y_4
           \}}Q_3
 \left(\begin{array}{c}
            x_1 \\
            x_2 \\
            y_1\\
            y_2\\
            y_3\\
            y_4
          \end{array}\right)
 +I_{\{x_1\geq y_1, x_2\geq y_4\}}Q_4 \left(\begin{array}{c}
            x_1 \\
            x_2 \\
            y_1\\
            y_2\\
            y_3\\
            y_4
          \end{array}\right),
%\end{eqnarray*}
\end{split}
\end{equation}
where the matrices $Q_i\;(i=1,2,3,4)$ are given as below:
\begin{eqnarray*}
Q_1=\left(
  \begin{array}{cc|cccc}
   -a_1 & a_2 &  &  &  &  \\
    a_1 &  -a_2  &  &  &  \\\hline
    -a_1 &  & -c_1 & c_2 &  &  \\
     & a_2 &  c_1 & -c_2 &  &  \\
    a_1 &  &  &  & -c_3 & c_4 \\
     & -a_2 &  &  & c_3 & -c_4
  \end{array}
\right),
\end{eqnarray*}
\begin{eqnarray*}
Q_2=\left(
  \begin{array}{cc|cccc}
   -a_1 & 0 &  &  &  & a_2 \\
    a_1 & 0 &  &  &  & -a_2 \\\hline
    -a_1 & 0 & -c_1 & c_2 &  &  \\
     & 0 &  c_1 & -c_2 &  & a_2 \\
    a_1 & 0 &  &  & -c_3 & c_4 \\
     & 0 &  &  & c_3 & -(c_4+a_2)
  \end{array}
\right),
\end{eqnarray*}
\begin{eqnarray*}
Q_3=\left(
  \begin{array}{cc|cccc}
   0& a_2 &  -a_1 &  &  &  \\
   0 &  -a_2 & a_1  &  &  &  \\\hline
   0  &  & -a_1-c_1 & c_2 &  &  \\
   0  & a_2 &  c_1 & -c_2 &  &  \\
   0  &  & a_1 &  & -c_3 & c_4 \\
   0  & -a_2 &  &  & c_3 & -c_4
  \end{array}
\right),
\end{eqnarray*}
\begin{eqnarray*}
Q_4=\left(
  \begin{array}{cc|cccc}
   0 & 0 & -a_1 &  &  & a_2 \\
   0 &  0 & a_1 &  &  &-a_2  \\\hline
    0 & 0 & -a_1-c_1 & c_2 &  &  \\
    0 & 0 &  c_1 & -c_2 &  & a_2 \\
    0 & 0 &  a_1  && -c_3 & c_4 \\
    0 &0 &  &  & c_3 & -a_2-c_4
  \end{array}
\right).
\end{eqnarray*}

\par As (\ref{eq:ChFA-Model-XY-ODEQ1Q2Q3Q4}) illustrates, the derived ODEs are
piecewise linear and they may be dominated by $Q_i\;(i=1,2,3,4)$
alternately. If the system is always dominated by only one specific
matrix after a time, then the ODEs become linear after this time.
For linear ODEs, as long as the eigenvalues of their coefficient
matrices are either zeros or have negative real parts, then bounded
solutions will converge as time tends to infinity, see
Corollary~\ref{Corollary D.2.1.}
in~\ref{section:Appendix-Jordan-Form}.  Fortunately, here the
eigenvalues of the matrices $Q_i\;(i=1,2,3,4)$ in
Model~\ref{model:X-Y} satisfy this property, the proof of which is
shown in~\ref{section:Appendix-Eigenvalues-Matrices}. In addition,
the solution of the derived ODEs from any PEPA model is bounded, as
Theorem~\ref{thm:ChFA-AnalyticalProof-BoundedSolution} illustrated.
Therefore, if we can guarantee that after a time the ODEs
(\ref{eq:ChFA-Model-XY-ODEQ1Q2Q3Q4}) become linear, which means that
one of the four matrices $Q_i\;(i=1,2,3,4)$ will be the coefficient
matrix of the linear ODEs, then by Corollary~\ref{Corollary D.2.1.}
the solution will converge. So the convergence problem is reduced to
determining whether the linearity can be finally guaranteed.

\par It is easy to see that the comparisons between $x_1$ and $y_1$,
$x_2$ and $y_4$ determine the linearity. For instance, if after a
time $T$, we always have $x_1>y_1$ and $x_2>y_4$, then the matrix
$Q_4$ will dominate the system. Fortunately, the invariance in the
model, as Lemma~\ref{lem:ChFA-ModelXY-Invariant} reveals, can
determine the comparisons in some circumstances. This is because
 this invariance reflects the relationship between different
component types that are connected through synchronisations. This
leads to several conclusions as follows.

\begin{proposition}\label{pro:ChFA-Model-XY-Convergence-theorem1}
If $y_1(0)+y_2(0)\leq x_1(0)$ and $y_3(0)+y_4(0)\leq x_2(0)$, then
the solution of (\ref{eq:ChFA-Model-XY-ODEQ1Q2Q3Q4}) converges.
\end{proposition}

\begin{proof} By Lemma~\ref{lem:ChFA-ModelXY-Invariant}, $y_1(t)+y_2(t)\leq x_1(t)$
and $y_3(t)+y_4(t)\leq x_2(t)$ for all time $t$. Since both $y_2(t)$
and $y_4(t)$ are nonnegative by
Theorem~\ref{thm:ChFA-AnalyticalProof-BoundedSolution}, we have
$y_1(t)\leq x_1(t)$ and $y_4(t)\leq x_2(t)$ for any $t$. Thus,
(\ref{eq:ChFA-Model-XY-ODEQ1Q2Q3Q4}) becomes
%\begin{equation}\label{eq:ChFA-Model-XY-Q4dominates}
%\begin{split}
%\left(
%          \begin{array}{c}
%            \frac{\mathrm{d}x_1}{\mathrm{d}t} \\
%            \frac{\mathrm{d}x_2}{\mathrm{d}t} \\
%            \frac{\mathrm{d}y_1}{\mathrm{d}t} \\
%            \frac{\mathrm{d}y_2}{\mathrm{d}t} \\
%            \frac{\mathrm{d}y_3}{\mathrm{d}t} \\
%            \frac{\mathrm{d}y_4}{\mathrm{d}t}
%          \end{array}
% \right)=&Q_4
% \left(\begin{array}{c}
%            x_1 \\
%            x_2 \\
%            y_1\\
%            y_2\\
%            y_3\\
%            y_4
%          \end{array}\right).
%          \end{split}
%\end{equation}
\begin{equation}\label{eq:ChFA-Model-XY-Q4dominates}
\left(   \frac{\mathrm{d}x_1}{\mathrm{d}t},
            \frac{\mathrm{d}x_2}{\mathrm{d}t},
            \frac{\mathrm{d}y_1}{\mathrm{d}t},
            \frac{\mathrm{d}y_2}{\mathrm{d}t},
            \frac{\mathrm{d}y_3}{\mathrm{d}t},
            \frac{\mathrm{d}y_4}{\mathrm{d}t}
          \right)^T=Q_4
 \left(     x_1,
            x_2,
            y_1,
            y_2,
            y_3,
            y_4\right)^T.
\end{equation}
Notice that (\ref{eq:ChFA-Model-XY-Q4dominates}) is linear, and all
eigenvalues of $Q_4$ other than zeros have negative real parts, then
according to Corollary~\ref{Corollary D.2.1.}, the solution of
(\ref{eq:ChFA-Model-XY-Q4dominates}) converges as time goes to
infinity.
\end{proof}

\begin{proposition}\label{pro:ChFA-Model-XY-Convergence-theorem2}
 Suppose $y_1(0)+y_2(0)> x_1(0)$ and $y_3(0)+y_4(0)\leq
x_2(0)$. If either (I). $\displaystyle
N>\left(2+\frac{a_1+c_1}{c_2}\right)M$, or (II). $\displaystyle N>
\frac{2(c_1+c_2)+a_2}{c_2}M$, where $N>M>0$ are the populations of
$Y$ and $X$ respectively, then the solution of
(\ref{eq:ChFA-Model-XY-ODEQ1Q2Q3Q4}) converges.
\end{proposition}

\begin{proof} Suppose (I) holds. According to the conservation law,
$\sum_{i=1}^4y_i(t)=N$. By the boundedness of the solution, we have
$x_2(t)\leq M$. Then by Lemma~\ref{lem:ChFA-ModelXY-Invariant},
$y_3(t)+y_4(t)\leq x_2(t)\leq M$. Therefore,
\begin{equation}\label{eq:ChFA-local-300}
y_2(t)=N-(y_3(t)+y_4(t))-y_1(t)\geq N-M-y_1(t).
\end{equation}
Since $\min\{x_1,y_1\}\leq y_1$, so $-a\min\{x_1,y_1\}\geq -a_1y_1$.
 Thus
\begin{equation}
\begin{split}
     \frac{\mathrm{d}y_1}{\mathrm{d}t}&=-a_1\min\{x_1,y_1\}-c_1y_1+c_2y_2\\
                    &\geq-a_1y_1-c_1y_1+c_2y_2\\
                    &\geq -(a_1+c_1)y_1+c_2(N-M-y_1)\\
                    &=-(a_1+c_1+c_2)y_1+c_2(N-M).
\end{split}
\end{equation}
That is
$$
   \frac{\mathrm{d}y_1}{\mathrm{d}t}\geq -(a_1+c_1+c_2)y_1+c_2(N-M).
$$
Applying Lemma~\ref{Lemma:Differential-Inequality}
in~\ref{section:Appendix-Some-Theorems} to this formula, we have
\begin{equation}\label{equ: local22}
    y_1(t)\geq
    \left(y_1(0)-\frac{c_2(N-M)}{a_1+c_1+c_2}\right)e^{-(a_1+c_1+c_2)t}+\frac{c_2(N-M)}{a_1+c_1+c_2}.
\end{equation}
Since the first term of the right side of (\ref{equ: local22})
converges to zero as time goes to infinity, i.e.
$\displaystyle\lim_{t\rightarrow\infinity}\left(y_1(0)-\frac{c_2(N-M)}{a_1+c_1+c_2}\right)e^{-(a_1+c_1+c_2)t}=0$,
and the second term $\displaystyle\frac{c_2(N-M)}{a_1+c_1+c_2}>M$
which results from the condition $\displaystyle
N>\left(2+\frac{a_1+c_1}{c_2}\right)M$, then there exists $T>0$ such
that for any $t>T$, $y_1(t)>M\geq x_1(t)$. Then after time $T$,
(\ref{eq:ChFA-Model-XY-ODEQ1Q2Q3Q4}) becomes linear, and is
dominated by $Q_2$. Because all eigenvalues of $Q_2$ are either
zeros or have negative real parts, the solution converges as time
goes to infinity.

\par Now we assume $(II)$ holds. Similarly, since $\min\{x_2,y_4\}\leq x_2\leq
M$, and $y_1\leq N-y_2$ which is due to $y_1+y_2\leq N$, we have
\begin{equation}
\begin{split}
     \frac{\mathrm{d}y_2}{\mathrm{d}t} &=-c_2y_2+a_2\min\{x_2,y_4\}+c_1y_1\\
                     &\leq-c_2y_2+a_2M+c_1y_1\\
                     &\leq-c_2y_2+a_2M+c_1(N-y_2)\\
                     &=-(c_1+c_2)y_2+a_2M+c_1N.
             \end{split}
\end{equation}
By Lemma~\ref{Lemma:Differential-Inequality}
in~\ref{section:Appendix-Some-Theorems},
\begin{equation}
    y_2\leq
    \left(y_2(0)-\frac{a_2M+c_1N}{c_1+c_2}\right)e^{-(c_1+c_2)t}
    +\frac{a_2M+c_1N}{c_1+c_2}.
\end{equation}
Therefore, since $e^{-(c_1+c_2)t}$ in above formula converges to
zero as time tends to infinity, then for any $\epsilon>0$, there
exists $T>0$ such that for any time $t>T$,
\begin{equation}
    y_2\leq\frac{a_2M+c_1N}{c_1+c_2}+\epsilon.
\end{equation}
Notice that the condition
 $\displaystyle N>\frac{2(c_1+c_2)+a_2}{c_2}M$ implies
$$\displaystyle\frac{c_2N-a_2M-(c_1+c_2)M}{c_1+c_2}>M,$$ and let
$\epsilon$  be small enough that
 $\displaystyle
          \frac{c_2N-a_2M-(c_1+c_2)M}{c_1+c_2}-\epsilon>M.
 $
Then by (\ref{eq:ChFA-local-300}),  $y_1\geq (N-M)-y_2$. Therefore,
\begin{equation}
\begin{split}
     y_1&\geq
   (N-M)-y_2\\
   &\geq(N-M)-\frac{a_2M+c_1N}{c_1+c_2}-\epsilon\\
   &=\frac{c_2N-a_2M-(c_1+c_2)M}{c_1+c_2}-\epsilon\\
   &>M\geq x_1.
\end{split}
\end{equation}
So $y_1(t)>x_1(t)$, $y_4(t)\leq x_2(t)$, for any $t>T$, then by a
similar argument the solution of
(\ref{eq:ChFA-Model-XY-ODEQ1Q2Q3Q4}) converges.
\end{proof}

Both condition (I) and (II) in
Proposition~\ref{pro:ChFA-Model-XY-Convergence-theorem2} require $N$
to be larger enough than $M$, to guarantee that $y_1$ is larger than
$x_1$. Since our PEPA model is symmetric,
Proposition~\ref{pro:ChFA-Model-XY-Convergence-theorem2} has a
corresponding symmetric version.

\begin{proposition}\label{pro:ChFA-Model-XY-Convergence-theorem3}
 Suppose $y_1(0)+y_2(0)\leq x_1(0)$ and $y_3(0)+y_4(0)>x_2(0)$.
 If either (I). $\displaystyle N>\left(2+\frac{a_2+c_3}{c_4}\right)M$, or
(II). $\displaystyle N> \frac{2(c_3+c_4)+a_1}{c_1}M$, where $N>M>0$
are the populations of $Y$ and $X$ respectively, then the solution
of (\ref{eq:ChFA-Model-XY-ODEQ1Q2Q3Q4}) converges.
\end{proposition}

\par
The proof of
Proposition~\ref{pro:ChFA-Model-XY-Convergence-theorem3} is omitted
here. We should point out that in our model the shared activity
$action1$ (respectively, $action2$) has the same local rate $a_1$
(respectively, $a_2$). We have taken the advantage of this in the
above proofs. If the local rates of shared activities are not set to
be the same, analogous  conclusions can still hold but the
discussion will be more complicated. However, the structural
property invariance can still play an important role.

\par The above three propositions have illustrated the convergence for
all situations in terms of the starting state, except for the case
of $y_1(0)+y_2(0)> x_1(0), y_3(0)+y_4(0)> x_2(0).$ See a summary in
Table~\ref{table:ChFA-SummaryTable-Convergence-Model-X-Y}. If
$y_1(0)+y_2(0)> x_1(0)$, $y_3(0)+y_4(0)> x_2(0)$, then the dynamic
behaviour of the system is rather complex. A numerical study for
this case will be presented in the next subsection.

%\begin{table}[htbp]\caption{A summary for the convergence of Model~\ref{model:X-Y}}
%\label{table:ChFA-SummaryTable-Convergence-Model-X-Y}
%\begin{center}
%\begin{tabular}{|p{4.5cm}|p{3.5cm}|p{4cm}|}
%   \hline\hline
%Starting state condition & Additional condition & Conclusion
%\\\hline $y_1(0)+y_2(0)\leq x_1(0)$, $y_3(0)+y_4(0)\leq x_2(0)$ & No
%& Proposition~\ref{pro:ChFA-Model-XY-Convergence-theorem1}\\\hline
%
% $y_1(0)+y_2(0)> x_1(0)$,
%$y_3(0)+y_4(0)\leq x_2(0)$ & $N>k_1M$ &
%Proposition~\ref{pro:ChFA-Model-XY-Convergence-theorem2}\\\hline
%
%$y_1(0)+y_2(0)\leq x_1(0)$, $y_3(0)+y_4(0)> x_2(0)$ & $N>k_2M$ &
%Proposition~\ref{pro:ChFA-Model-XY-Convergence-theorem3}\\\hline
%
%$y_1(0)+y_2(0)> x_1(0)$, $y_3(0)+y_4(0)> x_2(0)$ & None identified &
%Explored numerically\\\hline
%\end{tabular}
%\end{center}
%\end{table}

\begin{table}[htbp]\caption{A summary for the convergence of Model~\ref{model:X-Y}}
\label{table:ChFA-SummaryTable-Convergence-Model-X-Y}
\begin{center}
\begin{tabular}{|c|c|c|}
   \hline\hline
Starting state condition & Additional condition & Conclusion
\\\hline
$y_1(0)+y_2(0)\leq x_1(0)$, & & \\
 $y_3(0)+y_4(0)\leq x_2(0)$ & No
& Proposition~\ref{pro:ChFA-Model-XY-Convergence-theorem1}\\\hline

$y_1(0)+y_2(0)>x_1(0)$, & & \\
$y_3(0)+y_4(0)\leq x_2(0)$ & $N>k_1M$ &
Proposition~\ref{pro:ChFA-Model-XY-Convergence-theorem2}\\\hline

$y_1(0)+y_2(0)\leq x_1(0)$, & & \\
$y_3(0)+y_4(0)> x_2(0)$ & $N>k_2M$ &
Proposition~\ref{pro:ChFA-Model-XY-Convergence-theorem3}\\\hline

$y_1(0)+y_2(0)> x_1(0)$, & & \\
 $y_3(0)+y_4(0)> x_2(0)$ & None identified &
Explored numerically\\\hline
\end{tabular}
\end{center}
\end{table}

\begin{figure}[htbp]
 \begin{center}
 \includegraphics[width=9cm]{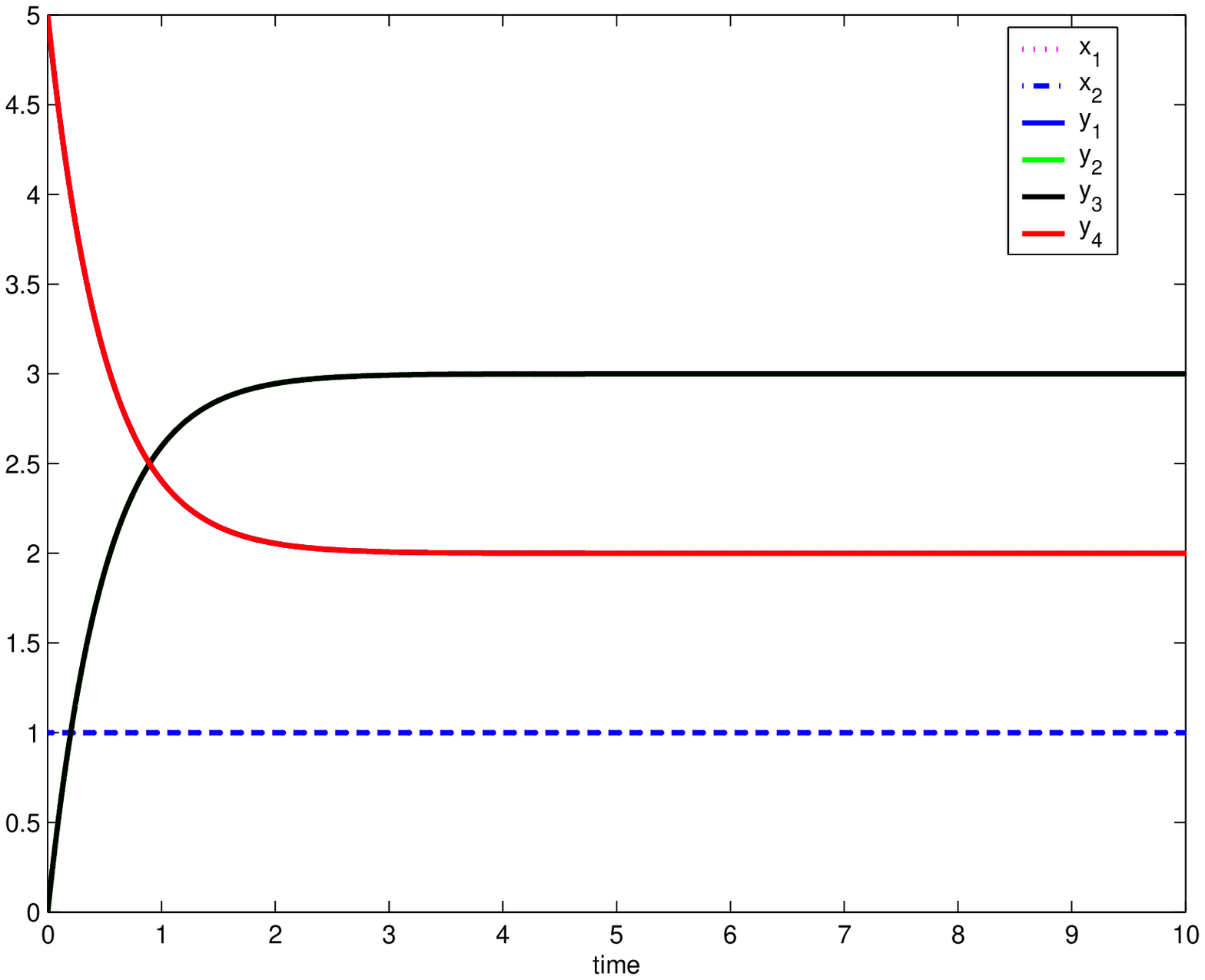}\\
 \end{center} \caption[Numerical study for Model~\ref{model:X-Y}: rates $(1,1,1,1,1,1)$]{
 Numerical study for Model~\ref{model:X-Y}:
 rates $(1,1,1,1,1,1)$; equilibrium point $(1,1,2,3,3,2)$.
  (Note: the curves of $x_1$ and $x_2$, the curves of $y_1$ and $y_4$, as
well as those of $y_2$ and $y_3$ respectively, completely overlap.)
 }\label{fig:Q1-115005}
\end{figure}

\begin{figure}[htbp]
 \begin{center}
 \includegraphics[width=9cm]{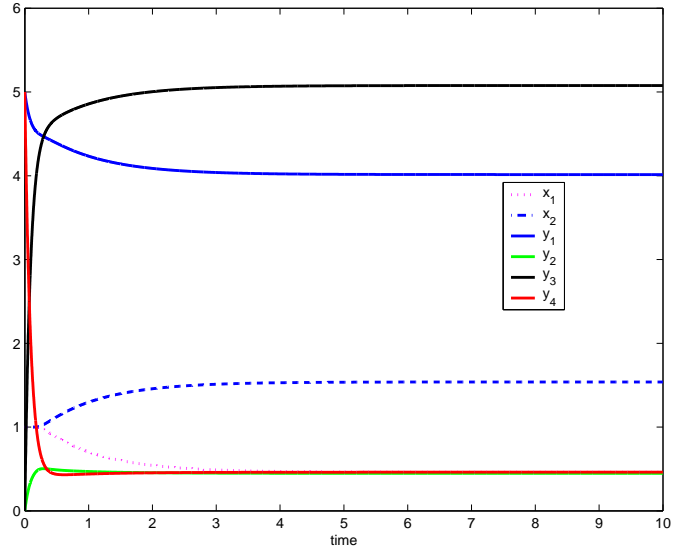}\\
 \end{center}
 \caption[Numerical study for Model~\ref{model:X-Y}: rates $(1,1,1,10,1,10)$]{
 Numerical study for Model~\ref{model:X-Y}:
 rates $(1,1,1,10,1,10)$; equilibrium point $(0.4616,1.5384,4.0140,0.4476,5.0769, 0.4615)$
 }\label{fig:Q2-RX1_1Y1_10_1_10}
\end{figure}

\begin{figure}[htbp]
 \begin{center}
 \includegraphics[width=9cm]{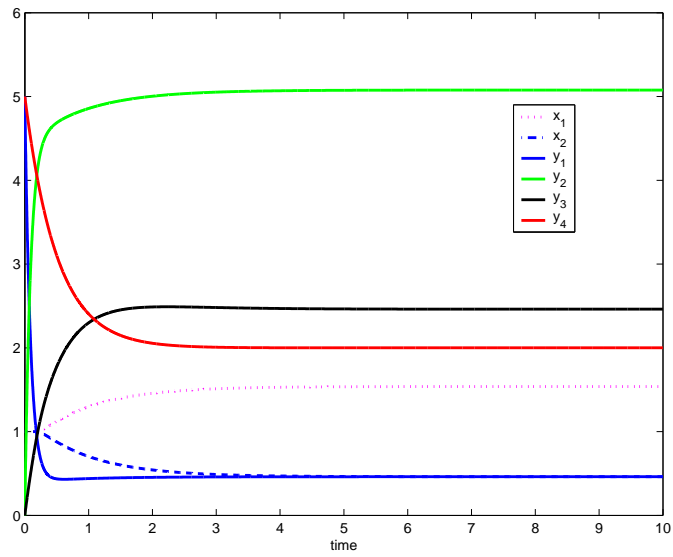}\\
 \end{center} \caption[Numerical study for Model~\ref{model:X-Y}: rates $(1,1,10,1,1,1)$]{
 Numerical study for Model~\ref{model:X-Y}:
 rates $(1,1,10,1,1,1)$;
 equilibrium point $(1.5384,0.4616,0.4615,5.0769,2.4616,2.0000)$
 }\label{fig:Q3-RX1_1Y10_1_1_1}
\end{figure}

\begin{figure}[htbp]
 \begin{center}
 \includegraphics[width=9cm]{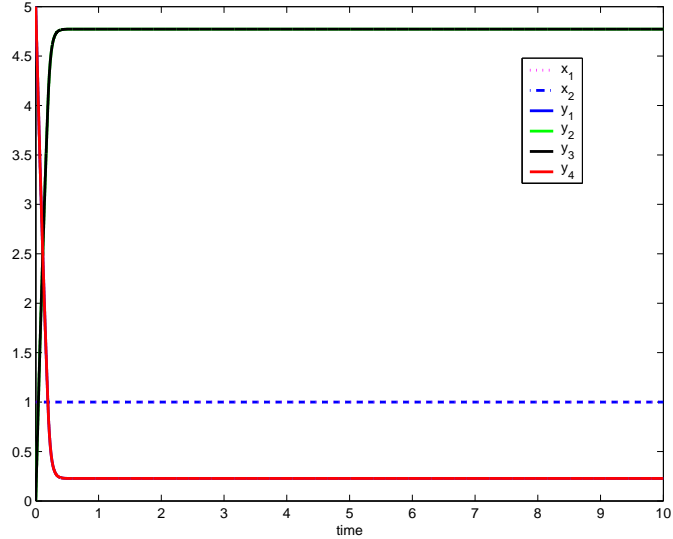}\\
 \end{center} \caption[Numerical study for Model~\ref{model:X-Y}: rates $(20,20,1,1,1,1)$]
 {Numerical study for Model~\ref{model:X-Y}:
 rates $(20,20,1,1,1,1)$;
  equilibrium point $(1,1,0.2273,4.7727,4.7727,0.2273)$.
  (Note: the curves of $x_1$ and $x_2$, the curves of $y_1$ and $y_4$, as
well as those of $y_2$ and $y_3$ respectively, completely overlap.)
 }\label{fig:Q4-RX20_20Y1_1_1_1}
\end{figure}

\subsubsection{Experimental study of convergence}

This subsection presents a numerical study at different action rate
conditions. The  starting state in this subsection is always assumed
as $(1,1,5,0,0,5)$, which satisfies the condition of
$y_1(0)+y_2(0)>x_1(0)$ and $y_3(0)+y_4(0)>x_2(0)$.

\par  If all the action rates in the model are set to one, i.e.
$(a_1,a_2,c_1,c_2,c_3,c_4)=(1,1,1,1,1,1)$, then the equilibrium
point of the ODEs is
$(x_1^{*},x_2^{*},y_1^{*},y_2^{*},y_3^{*},y_4^{*})=(1,1,2,3,3,2)$,
as the numerical solution of the ODEs illustrates. In this case, the
matrix $Q_1$ finally dominates the system. See
Figure~\ref{fig:Q1-115005}. Notice that in this figure, the curves
of $x_1$ and $x_2$ completely overlap, as well as the curves of
$y_1$ and $y_4$, and those of $y_2$ and $y_3$.

\par In other
situations, for example, if
$(a_1,a_2,c_1,c_2,c_3,c_4)=(1,1,1,10,1,10)$ then the equilibrium
point is $(0.4616, 1.5384, 4.0140,    0.4476, 5.0769, 0.4615)$, and
the matrix $Q_2$ eventually dominates the system. See
Figure~\ref{fig:Q2-RX1_1Y1_10_1_10}. Moreover, the matrices $Q_3$
and $Q_4$ can also finally dominate the system as long as the action
rates are appropriately specified. See
Figure~\ref{fig:Q3-RX1_1Y10_1_1_1} and
Figure~\ref{fig:Q4-RX20_20Y1_1_1_1}. We should point out that
similarly to Figure~\ref{fig:Q1-115005}, in
Figure~\ref{fig:Q4-RX20_20Y1_1_1_1} the curves of $x_1$ and $x_2$,
the curves of $y_1$ and $y_4$, as well as those of $y_2$ and $y_3$
respectively, completely overlap.

\par In short, the system dynamics is rather complex in the situation of
$y_1(0)+y_2(0)>x_1(0)$ and $ y_3(0)+y_4(0)>x_2(0)$. A summary of
these numerical studies is organised in
Table~\ref{table:ChFA-Complex-Q1Q2Q3Q4-XYmodel}.

\begin{table}[htbp]\caption{Complex dynamical behaviour of
Model~\ref{model:X-Y}: starting state
(1,1,5,0,0,5)}\label{table:ChFA-Complex-Q1Q2Q3Q4-XYmodel}
\begin{minipage}{13.8cm}
\begin{center}
\small
\begin{tabular}{|c|c|c|c|}
   \hline\hline
Rates: & Equilibrium points: & Dominator & Figure\\
$(a_1,a_2,c_1,c_2,c_3,c_4)$ &
$(x_1^{*},x_2^{*},y_1^{*},y_2^{*},y_3^{*},y_4^{*})$&  &
\\\hline

$(1,1,1,1,1,1)$ & $(1,1,2,3,3,2)$ & $Q_1$ &
Figure~\ref{fig:Q1-115005}\\\hline

$(1,1,1,10,1,10)$ & $v_1$\footnote{$v_1=(0.4616,    1.5384, 4.0140,
0.4476, 5.0769, 0.4615)$} & $Q_2$ &
Figure~\ref{fig:Q2-RX1_1Y1_10_1_10}\\\hline

$(1,1,10,1,1,1)$ &
$v_2$\footnote{$v_2=(1.5384,0.4616,0.4615,5.0769,2.4616,2.0000)$} &
$Q_3$ & Figure~\ref{fig:Q3-RX1_1Y10_1_1_1}\\\hline

 $(20,20,1,1,1,1)$ &
$(1,1,0.2273,4.7727,4.7727,0.2273)$ & $Q_4$ &
Figure~\ref{fig:Q4-RX20_20Y1_1_1_1}\\\hline

\end{tabular}
\end{center}
\end{minipage}
\end{table}

%\begin{table}[htbp]\caption{Complex dynamical behaviour of
%Model~\ref{model:X-Y}: starting state
%(1,1,5,0,0,5)}\label{table:ChFA-Complex-Q1Q2Q3Q4-XYmodel}
%\begin{center}
%\small
%\begin{tabular}{|c|c|c|c|}
%   \hline\hline
%Rates: & Equilibrium points: & Dominator & Figure\\
%$(a_1,a_2,c_1,c_2,c_3,c_4)$ &
%$(x_1^{*},x_2^{*},y_1^{*},y_2^{*},y_3^{*},y_4^{*})$&  &
%\\\hline
%
%$(1,1,1,1,1,1)$ & $(1,1,2,3,3,2)$ & $Q_1$ &
%Figure~\ref{fig:Q1-115005}\\\hline
%
%$(1,1,1,10,1,10)$ & $(0.4616,    1.5384,    4.0140,    0.4476,
%5.0769, 0.4615)$ & $Q_2$ &
%Figure~\ref{fig:Q2-RX1_1Y1_10_1_10}\\\hline
%
%$(1,1,10,1,1,1)$ & $(1.5384,0.4616,0.4615,5.0769,2.4616,2.0000)$ &
%$Q_3$ & Figure~\ref{fig:Q3-RX1_1Y10_1_1_1}\\\hline
%
% $(20,20,1,1,1,1)$ &
%$(1,1,0.2273,4.7727,4.7727,0.2273)$ & $Q_4$ &
%Figure~\ref{fig:Q4-RX20_20Y1_1_1_1}\\\hline
%
%\end{tabular}
%\end{center}
%\end{table}

\section{Fluid analysis: an analytic approach (II)}

\subsection{Convergence for two component types and one
synchronisation (I): a special case}

The problem of convergence for more general models without strict
conditions, is rather complex and has not been completely solved.
But for a particular class of PEPA model --- a model  composed of
two types of component with one synchronisation between them, we can
determine the convergence of the solutions of the derived ODEs.

\par As discussed in the previous section, the  ODEs derived from
PEPA are piecewise linear and may be dominated by different
coefficient matrices alternately. For any PEPA model which has two
component types and one synchronisation, the two corresponding
coefficient matrices can be proved to have a good property: their
eigenvalues are either zeros or have negative real parts. The
remaining issue for convergence is to ascertain that these two
matrices will not always alternately dominate the system. In fact,
we will prove that under some mild conditions, there exists a time
after which there is only one coefficient matrix dominating the
system. This means the ODEs become linear after that time. Since the
coefficient matrix of the linear ODEs satisfies the good eigenvalue
property, then by Corollary~\ref{Corollary D.2.1.}, the bounded
solution will converge as time goes to infinity.

\par We first utilise  an example in this section to show our approach to dealing with the
convergence problem for this class of PEPA models. The proof for a
general case in this class is presented in the next subsection.

\subsubsection{A previous model and the fluid approximation}
\label{Subsection_Features of Matrix}

Let us look at the following PEPA model, which is
Model~\ref{model:Uer-Provider} presented previously:
\begin{equation*}
\begin{split}
User_1 \rmdef &(task_1, a).User_2\\
User_2 \rmdef &(task_2, b).User_1\\
Provider_1 \rmdef &(task_1, a).Provider_2\\
Provider_2 \rmdef &(reset, d).Provider_1\\
 (User_1[M]) &
\sync{\{task1\}} (Provider_1[N]).
\end{split}
\end{equation*}
The derived ODEs are as follows:
\begin{equation}\label{eq:ChFA-ODEs-UserProvider1}
\left\{
\begin{split}
\frac{\mathrm{d}x_1}{\mathrm{d}t}&=-a\min\{x_1,y_1\}+bx_2\\
\frac{\mathrm{d}x_2}{\mathrm{d}t}&=a\min\{x_1,y_1\}-bx_2\\
\frac{\mathrm{d}y_1}{\mathrm{d}t}&=-a\min\{x_1,y_1\}+dy_2\\
\frac{\mathrm{d}y_2}{\mathrm{d}t}&=a\min\{x_1,y_1\}-dy_2\\
\end{split}\right.
\end{equation}
where $x_i$ and $y_i$ represent the populations of $User_i$ and
$Provider_i$ respectively, $i=1,2$. Clearly,
(\ref{eq:ChFA-ODEs-UserProvider1}) is equivalent to
\begin{equation}\label{eq:ChFA-ODEs-UserProvider2}
\left(
       \begin{array}{c}
       \frac{\mathrm{d}x_1}{\mathrm{d}t}\\
       \frac{\mathrm{d}x_2}{\mathrm{d}t}\\
       \frac{\mathrm{d}y_1}{\mathrm{d}t}\\
       \frac{\mathrm{d}y_2}{\mathrm{d}t}\\
       \end{array}
\right) =I_{\{x_1\leq y_1\}}Q_1\left(
       \begin{array}{c}
       x_1\\
       x_2\\
       y_1\\
       y_2\\
       \end{array}
     \right)+I_{\{x_1>y_1\}}Q_2\left(
       \begin{array}{c}
       x_1\\
       x_2\\
       y_1\\
       y_2\\
       \end{array}
     \right),
\end{equation}
where
\begin{equation}
Q_1=\left(
       \begin{array}{cc|cc}
        -a & b & 0 & 0 \\
         a & -b & 0 & 0 \\\hline
         -a & 0 & 0 & d \\
         a & 0 & 0 & -d
       \end{array}
     \right),\quad
Q_2=\left(
       \begin{array}{cc|cc}
         0 & b & -a & 0 \\
         0 & -b & a & 0 \\\hline
         0 & 0 & -a & d \\
         0 & 0 & a & -d
       \end{array}
     \right).
\end{equation}

\par Our interest is to
see if the solution of (\ref{eq:ChFA-ODEs-UserProvider2}) will
converge as time goes to infinity. As we mentioned, this convergence
problem can be divided into two subproblems, i.e. whether the
nonlinear equations can finally become linear and whether the
eigenvalues of the coefficient matrix are either zeros or have
negative real parts. If answers to these two subproblems are both
positive, then the convergence will hold.

\par The second subproblem can be easily dealt with.
By calculations, the matrix $Q_1$ has eigenvalues $0$ (two folds),
$-d$, and $-(a+b)$. Similarly, $Q_2$ has eigenvalues $0$ (two
folds), $-b, -(a+d)$. Therefore, the eigenvalues of $Q_1$ and $Q_2$
other than zeros are negative. Moreover, for a general PEPA model
which has two component types and one synchronisation,
Theorem~\ref{thm:ChFA-Q1Q2-Eigen-NegativeRealParts} in the next
section shows that the corresponding coefficient matrices always
have this property.

\par The remaining work to determine the convergence of  the ODE solution,
is to solve the first subproblem, i.e. to ascertain that after a
time it is always the case that $x_1>y_1$ or $x_1\leq y_1$. In this
model, there is no invariance relating the two different component
types, so we cannot rely on invariants to investigate this
subproblem. However, we have a new way to cope with this problem.

%\subsubsection{Outline of proof (II): Dealing with nonlinearity in PEPA}\label{subsec:ChFA-RoutePEPA-Introduction}
\subsubsection{Proof outline of convergence}\label{subsec:ChFA-RoutePEPA-Introduction}
%\par As we have mentioned in the previous subsection,
%all remaining work to determine the convergence for our ODEs
%(\ref{eq:ChFA-ODEs-UserProvider1}) or
%(\ref{eq:ChFA-ODEs-UserProvider2}) is to ensure that the ODEs will
%finally become linear.

%For convenience, the ODEs
%(\ref{eq:ChFA-ODEs-UserProvider1}) are presented again, see the
%following
%\begin{equation}\label{eq:ChFA-ODEs-UserProvider1-1}
%\left\{
%\begin{split}
%\frac{\mathrm{d}x_1}{\mathrm{d}t}&=-a\min\{x_1,y_1\}+bx_2\\
%\frac{\mathrm{d}x_2}{\mathrm{d}t}&=a\min\{x_1,y_1\}-bx_2\\
%\frac{\mathrm{d}y_1}{\mathrm{d}t}&=-a\min\{x_1,y_1\}+dy_2\\
%\frac{\mathrm{d}y_2}{\mathrm{d}t}&=a\min\{x_1,y_1\}-dy_2\\
%\end{split}\right.
%\end{equation}
 Notice that $y_1(t)\leq N$ by the boundedness of
solutions.  If we can prove that after time $T$, $x_1(t)\geq cM$,
where $c>0$ is independent of $M$, we will get, provided $cM> N$,
$$
x_1(t)\geq cM> N\geq y_1(t),\; t\geq T.
$$
Therefore, the ODEs (\ref{eq:ChFA-ODEs-UserProvider1}) will become
linear after time $T$. In the following, we identify that
$x_1(t)\geq cM$.

\par Let
\begin{equation*}
\alpha(t)=\left\{\begin{array}{cc}
            \frac{\min\{x_1(t),y_1(t)\}}{x_1(t)}, & x_1(t)\neq0, \\
            1, & x_1(t)=0,
          \end{array}\right.
\end{equation*}
then $0\leq \alpha(t)\leq 1$ by the nonnegativity of $x_1(t)$ and
$y_1(t)$. The ODEs associated with component type $X$ can be
rewritten as
\begin{equation}\label{eq:ChFA-ODEs-X-UserProvider}
\left\{
\begin{split}
\frac{\mathrm{d}x_1}{\mathrm{d}t}&=-a\alpha(t)x_1+bx_2,\\
\frac{\mathrm{d}x_2}{\mathrm{d}t}&=a\alpha(t)x_1-bx_2.
\end{split}\right.
\end{equation}
Let
 \begin{equation*}
    A(t)=\left(
           \begin{array}{cc}
             -a\alpha(t) & b \\
             a\alpha(t) & -b \\
           \end{array}
         \right),
\quad X(t)=\left(\begin{array}{c}
         x_1(t) \\
         x_2(t)
       \end{array}\right).
\end{equation*}
Then (\ref{eq:ChFA-ODEs-X-UserProvider}) can be written as
\begin{equation}\label{eq:dX/dt=AX}
  \frac{dX(t)}{dt}=A(t)X(t).
\end{equation}
The solution of (\ref{eq:dX/dt=AX}) is
\begin{equation}
 \displaystyle{X(t)=e^{\int_0^tA(s)ds}X(0)}.
\end{equation}

%\par Notice, if $A(t)$ is a constant matrix $A$, then the solution
%$\displaystyle\ e^{\int_0^tA(s)ds}X(0)=e^{At}X(0)$ coincides with
%(\ref{eq:ChFA-JaneSection-LinearCase-1}), which is in the linear
%situation.

\par Let $\displaystyle B(t)=\frac1t\int_0^tA(s)ds$, then
\begin{equation}
X(t)=e^{\int_0^tA(s)ds}X(0)=e^{tB(t)}X(0),
\end{equation}
and
\begin{equation}\label{eq:ChFA-locla-B(t)}
\begin{split}
  B(t)&=\frac1t\left(
         \begin{array}{cc}
           -a\int_0^t\alpha(s)ds & bt \\
           a\int_0^t\alpha(s)ds & -bt \\
         \end{array}
       \right)=\left(
         \begin{array}{cc}
           -a\beta(t) & b \\
           a\beta(t) & -b \\
         \end{array}
       \right),
\end{split}
\end{equation}
where $\displaystyle \beta(t)=\frac{\int_0^t\alpha(s)ds}{t}$.
Obviously, $0\leq\beta(t)\leq 1$ because $0\leq \alpha(s)\leq 1$.
%\par Similar to the diagonalisation of $A$ in Section~\ref{subsection:LinearODEsMarkov},
 Notice that the matrix $B(t)$ can be diagonalised as
\begin{equation}\label{eq:ChFA-Model-UserProvider-Diagonal-B(t)}
\begin{split}
   B(t)&=\left(
         \begin{array}{cc}
           -a\beta(t) & b \\
           a\beta(t) & -b \\
         \end{array}
       \right)\\
   &=\left(
     \begin{array}{cc}
       \frac{b}{a\beta(t)+b} & 1 \\
       \frac{a\beta(t)}{a\beta(t)+b} & -1 \\
     \end{array}
   \right)\left(
              \begin{array}{cc}
                0 & 0 \\
                0 & -(a\beta(t)+b) \\
              \end{array}
            \right)
   \left(
     \begin{array}{cc}
       1 & 1 \\
       \frac{a\beta(t)}{a\beta(t)+b} & -\frac{b}{a\beta(t)+b} \\
     \end{array}
   \right)\\
 &=U(t)\left(
              \begin{array}{cc}
                0 & 0 \\
                0 & -(a\beta(t)+b) \\
              \end{array}
            \right)U^{-1}(t),
\end{split}
\end{equation}
where
\begin{equation}\label{eq:ChFA-Model-UserProvider-Diagonal-B(t)-U}
U(t)=\left(\begin{array}{cc}
       \frac{b}{a\beta(t)+b} & 1 \\
       \frac{a\beta(t)}{a\beta(t)+b} & -1 \\
\end{array}\right),\quad
%\end{equation}
%\begin{equation}\label{eq:ChFA-Model-UserProvider-Diagonal-B(t)-inv(U)}
U^{-1}(t)= \left(
     \begin{array}{cc}
       1 & 1 \\
       \frac{a\beta(t)}{a\beta(t)+b} & -\frac{b}{a\beta(t)+b} \\
     \end{array}
   \right),
\end{equation}
so
\begin{equation}
X(t)=e^{tB(t)}X(0)=U(t)\left(
              \begin{array}{cc}
                1 & 0 \\
                0 & e^{-t(a\beta(t)+b)} \\
              \end{array}
            \right)U^{-1}(t)X(0).
\end{equation}

\par  In the  formula
(\ref{eq:ChFA-Model-UserProvider-Diagonal-B(t)}), both $0$ and
$-(a\beta(t)+b)$ are $B(t)$'s eigenvalues, while the columns of
$U(t)$ are the corresponding eigenvectors. In particular,
$\displaystyle\left(
\frac{b}{a\beta(t)+b},\frac{a\beta(t)}{a\beta(t)+b} \right)^T$, i.e.
the first column of $U(t)$, is the eigenvector corresponding to the
eigenvalue zero.

%\par Analogously to (\ref{eq:ChFA-JaneSection-LinearCase-2-x^{*}}),
We define a function $\hat{X}(t)$ by
\begin{equation}\label{eq:ChFA-Model-UserProvider-hat(X)(t)}
\hat{X}(t)=U(t)\left(
\begin{array}{cc}
                1 & 0 \\
                0 & 0 \\
              \end{array}
            \right)U^{-1}(t)X(0).
\end{equation}
By simple calculation,
\begin{equation}\label{eq:ChFA-Model-UserProvider-hat(X)(t)-1}
\displaystyle\hat{X}(t)=(\hat{x}_1(t),\hat{x}_2(t))^T=\left(\frac{bM}{a\beta(t)+b},
\frac{a\beta(t)M}{a\beta(t)+b}\right)^T.
\end{equation}
Clearly, $\hat{X}(t)$ is the normalised eigenvector corresponding
the zero eigenvalue (at the time $t$). Here the normalisaton is in
terms of the total population $M$ of the component type $X$.
Moreover, $\hat{X}(t)$ embodies some equilibrium meaning. In fact,
we have a conclusion:
\begin{equation}\label{eq:ChFA-Model-UserProvider-X(t)Approximate-hat(X)(t)}
\lim_{t\rightarrow\infinity}\|X(t)-\hat{X}(t)\|=0.
\end{equation}
Since the explicit expression of $U(t)$ is available, the proof of
(\ref{eq:ChFA-Model-UserProvider-X(t)Approximate-hat(X)(t)}) is easy
and thus omitted. This conclusion is also included in
Proposition~\ref{proposition:ChFA-X(t)-hatX(t)-UserProvider-NotExpl},
the proof of which does not rely on the explicit information of
$U(t)$ and will be introduced later.

\par Now we discuss the benefit brought by this formula. By
(\ref{eq:ChFA-Model-UserProvider-X(t)Approximate-hat(X)(t)}) the
first entry of $X(t)$ approximates the first entry of $\hat{X}(t)$,
i.e. $x_1(t)$ approximates
$\displaystyle\hat{x}_1(t)=\frac{bM}{a\beta(t)+b}$. Thus, for any
$\epsilon>0$, there exists $T>0$ such that for any $t\geq T$,
$$
   x_1(t)>\hat{x}_1(t)-\epsilon=\frac{bM}{a\beta(t)+b}-\epsilon.
$$
Since $\displaystyle\frac{bM}{a+b}\leq\frac{bM}{a\beta(t)+b}\leq M$
because $0\leq \beta(t)\leq 1$, so $\displaystyle x_1(t)>
\frac{bM}{a+b}-\epsilon$. Therefore, if $\displaystyle
\frac{bM}{a+b}>N$, then by the boundedness of $y_1(t)$, i.e.
$y_1(t)\leq N$, we have
$$x_1(t)> \frac{bM}{a+b}-\epsilon>N\geq y_1(t),
$$
as long as $\epsilon$ is small enough. This  means that $Q_2$ will
dominate the system after time $T$. So we have

\begin{proposition}\label{proposition:ChFA-ModelXY-Convergence}
If $\displaystyle\frac{bM}{a+b}>N$, then the solution of the ODEs
(\ref{eq:ChFA-ODEs-UserProvider2}) converges as time tends to
infinity.
\end{proposition}
Since the model is symmetric, this proposition has a symmetric
version: if $\displaystyle\frac{dN}{a+d}>M$, then the solution of
the ODEs also converges.

\par  As we discussed, there are two key steps in the proof
of Proposition~\ref{proposition:ChFA-ModelXY-Convergence}. The first
step is to establish the approximation of $x_1(t)$ to
$\hat{x}_1(t)$, i.e. $x_1(t)\approx\hat{x}_1(t)$. The second one is
to give an estimation $\hat{x}_1(t)\geq cM$. According to these two
conclusions, we have $x_1(t)\geq c'M$ where $c'<c$, and therefore
can conclude that $x_1(t)\geq c'M>N>y_1(t)$ provided  the condition
$c'M>N$. This is the main philosophy behind our proof for
$x_1(t)>y_1(t)$.

\par For the sake of generality, the proofs of these two conclusions
should not rely on the explicit expressions of the eigenvalues and
eigenvectors of $B(t)$. This is because  for general PEPA models
with two component types and one synchronisation, these explicit
expressions are not always available. The following subsection will
present our discussions about these steps, and the proofs for the
conclusions which do not rely on these explicit expressions.

\subsubsection{Proof not relying on explicit
expressions}\label{section:Proof-Not-Rely-On-Explicit-Expressions}

%\color{black}

This subsection will divide into two parts. In the first part, we
will give a lower bound for the eigenvalues of the coefficient
matrix $B(t)$, based on which a proof of the approximation of $X(t)$
to $\hat{X}(t)$ is given. The second part will establish the
estimation $\hat{x}_1(t)\geq cM$. All proofs in this subsection do
not require knowledge of the explicit expressions of the eigenvalues
and eigenvectors of $B(t)$.

\par For convenience, in this subsection we define
\begin{equation}\label{eq:ChFA-local-f(beta)}
 f(\beta)=\left(
         \begin{array}{cc}
           -a\beta & b \\
           a\beta & -b \\
         \end{array}
       \right),
\end{equation}
where $f$ is a matrix-valued function defined on $\mathbb{R}$. Then
the matrix
\begin{equation*}\label{eq:ChFA-locla-B(t)-0}
\begin{split}
  B(t)=\left(
         \begin{array}{cc}
           -a\beta(t) & b \\
           a\beta(t) & -b \\
         \end{array}
       \right)
\end{split}
\end{equation*}
can be written as
%is in fact a composition of the functions of
%$f(\beta)$ and $\beta=\beta(t)$. That is,
 $B(t)=f(\beta(t))$. The diagonalisation of $f(\beta)$ is
\begin{equation*}\label{eq:ChFA-locla-B(t)-0}
\begin{split}
  f(\beta)=g(\beta)\left(
         \begin{array}{cc}
           0 & 0 \\
           0 &\lambda(\beta) \\
         \end{array}
       \right)g^{-1}(\beta),
\end{split}
\end{equation*}
where $\lambda(\beta)$ is $f(\beta)$'s nonzero eigenvalue, and
$g(\beta)$ is a matrix whose columns are the eigenvectors of
$f(\beta)$. Here $g^{-1}(\beta)$ is the inverse of the matrix
$g(\beta)$.  Notice that $\lambda(\beta)$ is real, because if
$\lambda(\beta)$ is complex then its conjugation must be an
eigenvalue, which is contradicted by the fact that $f(\beta)$ only
has two eigenvalues, $0$ and $\lambda(\beta)$. The following
discussions in this subsection do not rely on the explicit
expressions of $\lambda(\beta)$, $g(\beta)$ and $g^{-1}(\beta)$,
although it is easy to see that $\lambda(\beta)=-a\beta+b$ and
\begin{equation*}
g(\beta)=\left(\begin{array}{cc}
       \frac{b}{a\beta+b} & 1 \\
       \frac{a\beta}{a\beta+b} & -1 \\
\end{array}\right), \quad g^{-1}(\beta)= \left(
     \begin{array}{cc}
       1 & 1 \\
       \frac{a\beta}{a\beta+b} & -\frac{b}{a\beta+b} \\
     \end{array}
   \right).
\end{equation*}

{\textbf{1. $X(t)$ approximates $\hat{X}(t)$}}\\

\par In the following, we will give a lower bound for the nonzero
eigenvalue of $f(\beta)$, i.e. $\lambda(\beta)$, and based on this
prove the approximation of $A(t)$ to $\hat{X}(t)$ as time tends to
infinity.

\par If $\beta>0$, then the transpose of $f(\beta)$, i.e. $f(\beta)^T$,
is an infinitesimal generator, and thus the nonzero eigenvalue
$\lambda(\beta)$ has negative real part, i.e.
$\Re(\lambda(\beta))<0$. If $\beta=0$, then $f(\beta)$ is
independent of $\beta$ and becomes a nonnegative matrix, i.e. each
entry of it is nonnegative. Based on the Perron-Frobenious theorem
which is presented in the next subsection, we can still have
$\Re(\lambda(0))<0$. Therefore, for any $\beta$, $f(\beta)$'s
eigenvalue other than zero has negative real part. This conclusion
is stated in the following lemma.

\begin{lemma}\label{lemma:ChFA-lambda(beta)geq0}
For any $\beta\in [0,1]$, $\Re(\lambda(\beta))<0$, where
$\lambda(\beta)$ is a nonzero eigenvalue of $f(\beta)$.
\end{lemma}
%\begin{proof}
%After a shift $\max\{a\beta,b\}I$, $f(\beta)$ becomes
%$\tilde{f}(\beta)=f(\beta)+\max\{a\beta,b\}I$, which is a
%nonnegative matrix. Then similarly to the proof of
%Theorem~\ref{thm:ChFA-Q1Q2-Eigen-NegativeRealParts} in the next
%section, which is based on the Perron-Frobenious theorem in the same
%section, we can conclude that the eigenvalue other than zero has
%negative real part.
%\end{proof}
The proof of Lemma~\ref{lemma:ChFA-lambda(beta)geq0} is presented
in~\ref{section:Appendix-Proof-Lemma}.
Lemma~\ref{lemma:ChFA-lambda(beta)geq0} can further lead to the
following
\begin{lemma}\label{lemma:ChFA-lambda(beta)-UniformLowerBound}
Let
\begin{equation}\label{eq:ChFA-Local-202}
\Lambda_1=\inf_{
\beta\in[0,1]}\{-\Re(\lambda(\beta))\mid\lambda(\beta) \;\mbox{is
$f(\beta)$'s non-zero
   eigenvalue}\},
\end{equation}
then $\Lambda_1>0$.
\end{lemma}

\begin{proof}
By Lemma~\ref{lemma:ChFA-lambda(beta)geq0}, $-\Re(\lambda(\beta))>0$
for any $\beta\in[0,1]$, so $\Lambda_1\geq 0$. Suppose
$\Lambda_1=0$. Because the eigenvalue $\lambda(\beta)$ is a
continuous function of the matrix $f(\beta)$, where $f(\beta)$ is
also continuous on $[0,1]$ with respect to $\beta$, so
$\lambda(\beta)$ is a continuous function of $\beta$ on $[0,1]$.
This is due to the fact that a composition of continuous functions
is still continuous. Noticing $\Re(\cdot)$ is also a continuous
function, so $-\Re(\lambda(\beta))$ is continuous with respect to
$\lambda(\beta)$, and thus with respect to $\beta$ on $[0,1]$. Since
a continuous function on a closed interval can achieve its minimum,
there exists $\beta_0\in[0,1]$ such that $-\Re(\lambda(\beta_0))$
achieves the minimum $\Lambda_1$, i.e.
$-\Re(\lambda(\beta_0))=\Lambda_1=0$. This is contradicted to
Lemma~\ref{lemma:ChFA-lambda(beta)geq0}. Therefore, $\Lambda_1>0$.
\end{proof}

\par For any $t\in[0,\infinity)$, $\beta(t)\in [0,1]$.
% Thus $\{\beta(t)\mid t\in[0,\infinity)\}\subseteq [0,1]$.
 Noticing
$B(t)=f(\beta(t))$, therefore
\begin{equation}\label{eq:ChFA-Local-201}
\begin{split}
&\{\lambda\mid\lambda \;\mbox{is $B(t)$'s non-zero
   eigenvalue},t>0\}\\
=& \{\lambda\mid\lambda \;\mbox{is $f(\beta(t))$'s non-zero
   eigenvalue},t>0\}\\
 \subseteq &\{\lambda\mid\lambda \;\mbox{is $f(\beta)$'s non-zero
   eigenvalue}; \beta\in[0,1]\}.
 \end{split}
\end{equation}
Thus,
\begin{equation}\label{eq:ChFA-Local-201}
\begin{split}
\Lambda &\triangleq\inf\{-\Re(\lambda)\mid\lambda \;\mbox{is
$B(t)$'s
nonzero   eigenvalue},t>0\}\\
&\geq \inf \{-\Re(\lambda)\mid\lambda \;\mbox{is $f(\beta)$'s
non-zero eigenvalue}; \beta\in[0,1]\}=\Lambda_1.
 \end{split}
\end{equation}
Because $\Lambda_1>0$ by
Lemma~\ref{lemma:ChFA-lambda(beta)-UniformLowerBound}, so
$\Lambda>0$. That is,
\begin{corollary}\label{corollary:ChFA-Lambda(beta(t))-LowerBound}
Let
   \begin{equation*}\Lambda=\inf_{t\geq0}\{-\Re(\lambda(t))\mid\lambda(t) \;\mbox{is $B(t)$'s non-zero
   eigenvalue}\},
   \end{equation*}
   then $\Lambda>0$.
\end{corollary}
Based on this corollary, we can prove the approximation of $X(t)$ to
$\hat{X}(t)$.
\begin{proposition}\label{proposition:ChFA-X(t)-hatX(t)-UserProvider-NotExpl}
 Let $X(t)=(x_1(t),x_2(t))^T=e^{tB(t)}X(0)$, i.e. the solution of
 (\ref{eq:ChFA-ODEs-X-UserProvider}). Let $\hat{X}(t)$ be defined by
 (\ref{eq:ChFA-Model-UserProvider-hat(X)(t)}), i.e.
\begin{equation}\label{eq:ChFA-hatX(t)}
   \hat{X}(t)=\left( \begin{array}{c}
                \hat{x}_1(t)   \\
                \hat{x}_2(t)\\
              \end{array}
            \right)
            =U(t)\left(
              \begin{array}{cc}
                1 &  \\
                 & 0\\
              \end{array}
            \right)U(t)^{-1}\left(
              \begin{array}{c}
                x_1(0)   \\
                x_2(0)\\
              \end{array}
            \right).
\end{equation}
Then
$\displaystyle\lim_{t\rightarrow\infinity}\|X(t)-\hat{X}(t)\|=0$.
\end{proposition}

\begin{proof}
Notice that eigenvectors of a matrix are continuous functions of the
matrix. Since $g(\beta)$ is composed of the eigenvectors of the
matrix $f(\beta)$ and $f(\beta)$ is continuous on $[0,1]$ with
respect to $\beta$, therefore $g(\beta)$ is continuous on $[0,1]$
with respect to $\beta$. Because the inverse of a matrix is a
continuous mapping, so $g^{-1}(\beta)$, i.e. the inverse of
$g(\beta)$, is continuous with respect to $g(\beta)$, and therefore
is continuous on $[0,1]$ with respect to $\beta$ since $g(\beta)$ is
continuous on $[0,1]$. Since any continuous function is bounded on a
compact set $[0,1]$, both $g(\beta)$ and $g^{-1}(\beta)$ are bounded
on $[0,1]$. That is, there exists $K>0$ such that $\|g(\beta)\|\leq
K$ and $\|g^{-1}(\beta)\|\leq K$ for all $\beta\in[0,1]$. Because
$$
   \{U(t)\mid t\in[0,\infinity)\}=\{g(\beta(t))\mid
   t\in[0,\infinity)\}\subseteq \{g(\beta)\mid \beta\in[0,1]\},
$$
we have $$\displaystyle\sup_{t\geq 0}\|U(t)\|\leq
\sup_{\beta\in[0,1]}\|g(\beta)\|\leq K.$$ Similarly,
$\displaystyle\sup_{t\geq 0}\left\|U^{-1}(t)\right\|\leq K$.
 Notice
\begin{equation*}
\begin{split}
    X(t)-\hat{X}(t)=& e^{tB(t)}X(0)-U(t)\left(
              \begin{array}{cc}
                1 & 0 \\
                0 & 0 \\
              \end{array}
            \right)U(t)^{-1}X(0)\\
    =&\left[U(t)\left(
              \begin{array}{cc}
                1 & 0 \\
                0 & e^{t\lambda(t)}\\
              \end{array}
            \right)U(t)^{-1}-U(t)\left(
              \begin{array}{cc}
                1 & 0 \\
                0 & 0 \\
              \end{array}
            \right)U(t)^{-1}\right]X(0)  \\
 =&U(t)\left(
              \begin{array}{cc}
                0 & 0 \\
                0 & e^{t\lambda(t)}\\
              \end{array}
            \right)
            U(t)^{-1}X(0),
   \end{split}
\end{equation*}
where $\lambda(t)$ is $B(t)$'s nonzero eigenvalue. By a similar
argument to $\lambda(\beta)$, $\lambda(t)$ is also real.
%, because if
%$\lambda(t)$ is complex then its conjugation must be an eigenvalue,
%which is contradicted to the fact of that $B(t)$ only has two
%eigenvalues, $0$ and $\lambda(t)$.
 Therefore,
$$-\lambda(t)=\Re(-\lambda(t))=-\Re(\lambda(t))\geq \Lambda>0$$
or $\lambda(t)\leq -\Lambda<0$, where $\Lambda$ is defined in
Corollary~\ref{corollary:ChFA-Lambda(beta(t))-LowerBound}. Then
\begin{equation*}
\begin{split}
  \|X(t)-\hat{X}(t)\|
 =&\left\|U(t)\left(
              \begin{array}{cc}
                0 & 0 \\
                0 & e^{t\lambda(t)}\\
              \end{array}
            \right)U(t)^{-1}X(0) \right\|\\
\leq & \|U(t)\|
         \left\|\left(
              \begin{array}{cc}
                0 & 0 \\
                0 & e^{t\lambda(t)}\\
              \end{array}
            \right)\right\|\|U(t)^{-1}\|\|X(0)\|\\
\leq & K^2\|X(0)\|e^{t\lambda(t)}\\
\leq & K^2\|X(0)\|e^{-t\Lambda}.
   \end{split}
\end{equation*}
Here we have used the norm property: $\|AB\|\leq \|A\|\|B\|$. Since
$\Lambda>0$, we have
$\displaystyle\lim_{t\rightarrow\infinity}\|X(t)-\hat{X}(t)\|=0$.
\end{proof}

\textbf{2: An lower-bound estimation on population in local
derivatives}\\

In the following, we will prove that there exists $T$, such that
$\hat{x}_1(t)\geq cM$ for any $t>T$.  We first define a function
$$
     h(\beta)=(h_1(\beta),h_2(\beta))^T=g(\beta)\left(
                               \begin{array}{cc}
                                 1 & 0 \\
                                 0 & 0 \\
                               \end{array}
                             \right)g^{-1}(\beta)\left(
                               \begin{array}{c}
                                 x_1(0)\\
                                 x_2(0)\\
                               \end{array}
                             \right).
$$
Clearly, we have $\hat{X}(t)=h(\beta(t))$ and
$\hat{x}_1(t)=h_1(\beta(t))$. Since $\beta(t)\in[0,1]$ for all $t$,
the following proposition can imply $\hat{x}_1(t)\geq cM$ .

\begin{proposition}\label{proposition:ChFA-Fact2-x1-LowerBound}
There exists $c>0$ such that
\begin{equation*}
      \inf_{\beta\in[0,1]}h_1(\beta)\geq cM.
\end{equation*}
where $M=x_1(0)+x_2(0)$, $c$ is independent of $M$.
\end{proposition}

\begin{proof} Without loss of generality, we assume $M=1$.
We will show $\displaystyle \inf_{\beta\in[0,1]}h_1(\beta)=c>0$.
Since $h_1(\beta)$ is a continuous function of $\beta$ which is due
to the continuity of $g(\beta)$ and $g^{-1}(\beta)$, $h_1(\beta)$
can achieve its minimum  on $[0,1]$. That is, there exists  $
\beta_0\in[0,1]$, such that
 $$ h_1(\beta_0)=\inf_{\beta\in[0,1]}h_1(\beta)=c.$$  Consider the matrix
$$f(\beta_0)=\left(
         \begin{array}{cc}
           -a\beta_0 & b \\
           a\beta_0 & -b \\
         \end{array}
       \right)$$
and a set of linear ODEs
\begin{equation}\label{eq:ChFA-local-linear-new}
  \left(
    \begin{array}{c}
      \frac{\mathrm{d}z_1}{\mathrm{d}t} \\
      \frac{\mathrm{d}z_2}{\mathrm{d}t} \\
    \end{array}
  \right)=f(\beta_0)
        \left(
    \begin{array}{c}
      z_1 \\
      z_2 \\
    \end{array}
  \right).
\end{equation}
The solution of (\ref{eq:ChFA-local-linear-new}), given an initial
value  $Z(0)=X(0)=(x_1(0),x_2(0))^T$, is \linebreak
$Z(t)=e^{tf(\beta_0)}X(0)$.

According to (\ref{eq:ChFA-Model-UserProvider-Diagonal-B(t)}),
$f(\beta_0)$ can be diagonalised as
$$
      f(\beta_0)=g(\beta_0)\left(
         \begin{array}{cc}
           0 & 0 \\
           0 &\lambda(\beta_0) \\
         \end{array}
       \right)g^{-1}(\beta_0).
$$
where $\lambda(\beta_0)$ is the nonzero and real eigenvalue of
$f(\beta_0)$. Thus
%Similarly to (\ref{eq:ChFA-JaneSection-LinearCase-1})
%and the discussions in Section~\ref{subsection:LinearODEsMarkov},
\begin{equation}\label{eq:ChFA-local-3}
\begin{split}
 Z(t)=e^{tf(\beta_0)}X(0)&=g(\beta_0)\left(
              \begin{array}{cc}
                1 & 0 \\
                0 & e^{t\lambda(\beta_0)} \\
              \end{array}
            \right)g(\beta_0)^{-1}\left(\begin{array}{c}
                x_1(0)   \\
                x_2(0) \\
              \end{array}\right).
\end{split}
 \end{equation}
Because $\lambda(\beta_0)<0$ by
Lemma~\ref{lemma:ChFA-lambda(beta)geq0}, so as time goes to
infinity,
\begin{equation}\label{eq:ChFA-local-4}
\begin{split}
 Z(t)&=g(\beta_0)\left(
              \begin{array}{cc}
                1 & 0 \\
                0 & e^{t\lambda(\beta_0)} \\
              \end{array}
            \right)g(\beta_0)^{-1}\left(\begin{array}{c}
                x_1(0)   \\
                x_2(0) \\
              \end{array}\right)\\
              &\longrightarrow
              g(\beta_0)\left(
              \begin{array}{cc}
                1 & 0 \\
                0 & 0 \\
              \end{array}
            \right)g(\beta_0)^{-1}\left(\begin{array}{c}
                x_1(0)   \\
                x_2(0) \\
              \end{array}\right)=h(\beta_0).
\end{split}
 \end{equation}
That is, $\displaystyle\lim_{t\rightarrow\infinity}Z(t)=h(\beta_0)$.
In the following, we discuss two possible cases: $\beta_0>0$ and
$\beta_0=0$.

\par  If $\beta_0>0$, then
the transpose of the matrix $f(\beta_0)$, i.e.\;$f(\beta_0)^T$, is
 an infinitesimal generator of an irreducible CTMC, which has two states and the
transition rates between these two states are $a\beta_0$ and $b$
respectively. Moreover, the transient distribution of this CTMC,
denoted by $Z(t)=(z_1(t),z_2(t))^T$, satisfies the ODEs
(\ref{eq:ChFA-local-linear-new}). As time goes to infinity, the
transient distribution $Z(t)$ converges to the unique steady-state
probability distribution. Since
$\displaystyle\lim_{t\rightarrow\infinity}Z(t)=h(\beta_0)$,
therefore $h(\beta_0)=(h_1(\beta_0),h_2(\beta_0))^T$ is the
steady-state probability distribution and thus $h_1(\beta_0)>0$. So
$\inf_{\beta\in[0,1]}h_1(\beta)=h_1(\beta_0)>0$.

\par If $\beta_0=0$, then
$$f(\beta_0)=\left(
         \begin{array}{cc}
           0 & b \\
           0 & -b \\
         \end{array}
       \right).$$
Since $\displaystyle\lim_{t\rightarrow\infinity}Z(t)=h(\beta_0)$,
therefore
$\displaystyle\left(\frac{\mathrm{d}z_1}{\mathrm{d}t},\frac{\mathrm{d}z_2}{\mathrm{d}t}\right)^T$
converges to zero. Letting time go to infinity on the both sides of
(\ref{eq:ChFA-local-linear-new}), we obtain the following
equilibrium equations,
\begin{equation}\label{eq:ChFA-local-100-01}
  \left(
    \begin{array}{c}
      0 \\
     0 \\
    \end{array}
  \right)=f(\beta_0)
        \left(
    \begin{array}{c}
      h_1(\beta_0) \\
      h_2(\beta_0) \\
    \end{array}
  \right).
\end{equation}
By the conservation law, $(h_1(\beta_0)+ h_2(\beta_0))^T=M=1$.
Therefore, $(h_1(\beta_0), h_2(\beta_0))^T$ satisfies
%\begin{displaymath}
\begin{equation}\label{eq:ChFA-local-100}
       \left\{\begin{array}{c}
                f(\beta_0)(h_1(\beta_0), h_2(\beta_0))^T=0,\\
                h_1(\beta_0)+h_2(\beta_0)=M=1.
              \end{array}\right.
\end{equation}
%\end{displaymath}
Solving (\ref{eq:ChFA-local-100}), we obtain the unique solution
$(h_1(\beta_0), h_2(\beta_0))^T=(1,0)^T$. Therefore, $h_1(\beta_0)$
is one, and thus $\inf_{\beta\in[0,1]}h_1(\beta)=h_1(\beta_0)>0$.
\end{proof}

\begin{remark}
As $\beta$ tends to $0$,
$$f(\beta)=\left(
         \begin{array}{cc}
           -a\beta & b \\
           a\beta & -b \\
         \end{array}
       \right)
\longrightarrow f(0)=\left(
         \begin{array}{cc}
           0 & b \\
           0 & -b \\
         \end{array}
       \right).
$$
Correspondingly, for the equilibrium $h(\beta)
=(h_1(\beta),h_2(\beta))^T =\left(\frac{bM}{a\beta+b},\frac{a\beta
M}{a\beta+b}\right)^T$ satisfying $f(\beta)h(\beta)=0$ and
$h_1(\beta)+h_2(\beta)=M$, we have $
 \left(\frac{bM}{a\beta+b},
 \frac{a\beta M}{a\beta+b}\right)^T
 \rightarrow (M,0)^T
$ as $\beta$ tends to zero. From the explicit expression, i.e.
$h_1(\beta)=\frac{bM}{a\beta+b}$, the minimum and maximum of
$h_1(\beta)$ are $\frac{bM}{a+b}$ and $M$ respectively, which
correspond to the matrices $f(1)$ and $f(0)$ respectively. In the
context of the PEPA model, $f(1)$ corresponds to a free subsystem
and there is no synchronisation effect on it, i.e. the subsystem of
component type $X$  is independent of $Y$. The matrix $f(0)$
reflects that the subsystem of $X$ has been influenced by the
subsystem of $Y$, i.e. the rates of shared activities are determined
by $Y$, that is, the term $a\min\{x_1,y_1\}$ has been replaced by
$ay_1$. Therefore the exit rates from the local derivative $X_1$
correspondingly become smaller since now $ay_1<ax_1$. In order to
balance the flux, which is described by the equilibrium equation,
the population of $X_1$ must increase. That is why the equilibrium
$h_1(\beta)$ increases as $\beta$ decreases. In short,
synchronisations  can increase the populations in syncrhonised local
derivatives in the steady state.

\par As an application of the above facts, if $h_1(\beta_0)>0$ for
some $\beta_0>0$, then we can claim that $h_1(0)>0$ because
$h_1(0)\geq h_1(\beta_0)>0$.
%Therefore, the proof of
%Proposition~\ref{proposition:ChFA-Fact2-x1-LowerBound} can avoid the
%explicit expression of $f(t)$.
\end{remark}

\par Obviously, Proposition~\ref{proposition:ChFA-Fact2-x1-LowerBound}
has a corollary:
\begin{corollary}\label{corollary:ChFA-hat(x1)>cM}
There exists $c>0$ such that for any $t\in[0,\infinity)$,
$\hat{x}_1(t)\geq cM.$
\end{corollary}

\par Proposition
~\ref{proposition:ChFA-X(t)-hatX(t)-UserProvider-NotExpl} and
Proposition~\ref{proposition:ChFA-Fact2-x1-LowerBound} can lead to
the following lemma.
\begin{lemma}\label{lem:ChFA-x1(t)>cM}
There exists $c>0, T>0$, such that $x_1(t)\geq  cM$ for all $t\geq
T$.
\end{lemma}

\begin{proof}
By Proposition~\ref{proposition:ChFA-Fact2-x1-LowerBound} or
Corollary~\ref{corollary:ChFA-hat(x1)>cM}, there exists $c_1, T_1>0$
such that $\hat{x}_1(t)\geq c_1M$ for any $t>T_1$. By
Proposition~\ref{proposition:ChFA-X(t)-hatX(t)-UserProvider-NotExpl},
$\lim_{t\rightarrow\infinity}|x_1(t)-\hat{x}_1(t)|=0$, which implies
that for any $\epsilon$, there exists $T_2>0$ such that for any
$t>T_2$, $x_1(t)>\hat{x}_1(t)-\epsilon$. Choose $T_2>T_1$, then we
have
$$x_1(t)>\hat{x}_1(t)-\epsilon\geq c_1M-\epsilon.$$  Therefore, there
exist $c,T>0$ such that such that $x_1(t)\geq cM$ for all $t>T$.
\end{proof}

Because $x_1(t)\geq cM$, provided $cM>N$ we  have $x_1(t)\geq
cM>N\geq y_1(t)$, i.e., the system will finally become linear. In
the following we will show how to apply our method to more general
cases.

\subsection{Convergence for two component types and one
synchronisation (II): general case}

This section deals with such an arbitrary PEPA model which has two
component types and one synchronisation. The local action rates of
the shared activity are not assumed to be the same. The main result
of this section is a convergence theorem: as long as the population
of one component type is sufficiently larger than the population of
the other, then the solution of the derived ODEs converges as time
tends to infinity.

\subsubsection{Features of coefficient matrix}

We assume the component types to be $X$ and $Y$. The component type
$X$ is assumed to have local derivatives $X_1,X_2,\cdots,X_m$, while
$Y$ has local derivatives $Y_1,Y_2,\cdots,Y_n$. We use $x_i(t)$ to
denote the population of $X$ in $X_i\;(i=1,\cdots,m)$ at time $t$.
Similarly, $y_j(t)$ denotes the population of $Y$ in
$Y_j\;(j=1,\cdots,n)$ at time $t$. Without loss of generality, we
assume  the synchronisation is associated with the local derivatives
$X_1$ and $Y_1$, i.e. the nonlinear term in the derived ODEs is
$\min\{rx_1(t),sy_1(t)\}$ where $r$ and $s$ are some constants. In
fact, if the synchronisation is associated with $X_i$ and $Y_j$, by
appropriately permuting their suffixes, i.e. $i\rightarrow
1,\;i+1\rightarrow 2,\;\cdots\;, i-1\rightarrow m$, $j\rightarrow
1,\;j+1\rightarrow 2,\;\cdots\;, j-1\rightarrow n$, the
synchronisation will be associated with $X_1$ and $Y_1$. According
to the mapping semantics presented previously, the derived ODEs from
this class of PEPA model are
\begin{equation}\label{eq:ChFA-TwoType-OneSync-ODEs}
\begin{split}
\frac{\mathrm{d}\mathbf{x}}{\mathrm{d}t}=\sum_{l}lf(\mathbf{x},l)
\end{split}
\end{equation}
where $\mathbf{x}=(x_1(t),\cdots,x_m(t),y_1(t),\cdots,y_n(t))^T$.
 For convenience, we denote
$$X(t)=\left(x_1(t),x_2(t),\cdots,x_m(t)\right)^T,$$
$$Y(t)=\left(y_1(t),y_2(t),\cdots,y_n(t)\right)^T.$$

In~(\ref{eq:ChFA-TwoType-OneSync-ODEs}) all terms are linear except
for those containing ``$\min\{rx_1(t),sy_1(t)\}$''. Notice
 $$
\min\{rx_1(t),sy_1(t)\}=I_{\{rx_1(t)\leq
sy_1(t)\}}rx_1(t)+I_{\{rx_1(t)> sy_1(t)\}}sy_1(t).
$$
When $rx_1(t)\leq sy_1(t)$, which is indicated by $I_{\{rx_1(t)\leq
sy_1(t)\}}=1$ and $I_{\{rx_1(t)> sy_1(t)\}}=0$, we can replace
$\min\{rx_1(t),sy_1(t)\}$  by $rx_1(t)$
in~(\ref{eq:ChFA-TwoType-OneSync-ODEs}). Then
(\ref{eq:ChFA-TwoType-OneSync-ODEs}) becomes linear since all
nonlinear terms are replaced by linear terms $rx_1(t)$, so the ODEs
have the following form,
\begin{equation}\label{loc11}
\left(
  \begin{array}{c}
   \frac{\mathrm{d}X}{\mathrm{d}t} \\
  \frac{\mathrm{d}Y}{\mathrm{d}t}\\
  \end{array}
\right)=Q_1 \left(
  \begin{array}{c}
   X \\
  Y\\
  \end{array}
\right),
\end{equation}
where $Q_1$ is a coefficient matrix. Similarly, if $rx_1(t)>
sy_1(t)$,  $\min\{rx_1(t),sy_1(t)\}$ can be replaced by $sy_1(t)$
in~(\ref{eq:ChFA-TwoType-OneSync-ODEs}). Then
(\ref{eq:ChFA-TwoType-OneSync-ODEs}) can become
\begin{equation}
\left(
  \begin{array}{c}
   \frac{\mathrm{d}X}{\mathrm{d}t} \\
  \frac{\mathrm{d}Y}{\mathrm{d}t}\\
  \end{array}
\right)=Q_2 \left(
  \begin{array}{c}
   X \\
  Y\\
  \end{array}
\right),
\end{equation}
where $Q_2$ is another coefficient matrix corresponding to the case
of $rx_1(t)>sy_1(t)$.

\par In short, the derived ODEs (\ref{eq:ChFA-TwoType-OneSync-ODEs}) are just the following
\begin{equation}\label{eq:ChFA-TwoType-OneSync-Q1Q2}
\left(
  \begin{array}{c}
   \frac{\mathrm{d}X}{\mathrm{d}t} \\
  \frac{\mathrm{d}Y}{\mathrm{d}t}\\
  \end{array}
\right)=I_{\{rx_1\leq sy_1\}}Q_1 \left(
  \begin{array}{c}
   X \\
  Y\\
  \end{array}
\right)+I_{\{rx_1>sy_1\}}Q_2 \left(
  \begin{array}{c}
   X \\
  Y\\
  \end{array}
\right).
\end{equation}
The case discussed in the previous section is a special case of this
kind of form. If the conditions $rx_1(t)\leq sy_1(t)$ and
$rx_1(t)>sy_1(t)$ occur alternately, then the matrices $Q_1$ and
$Q_2$ will correspondingly alternately dominate the system, as
Figure~\ref{fig:ChFA-Q1-Q2-Alternative} illustrates.

\begin{figure}[htb]
\begin{center}
  \includegraphics[width=8cm]{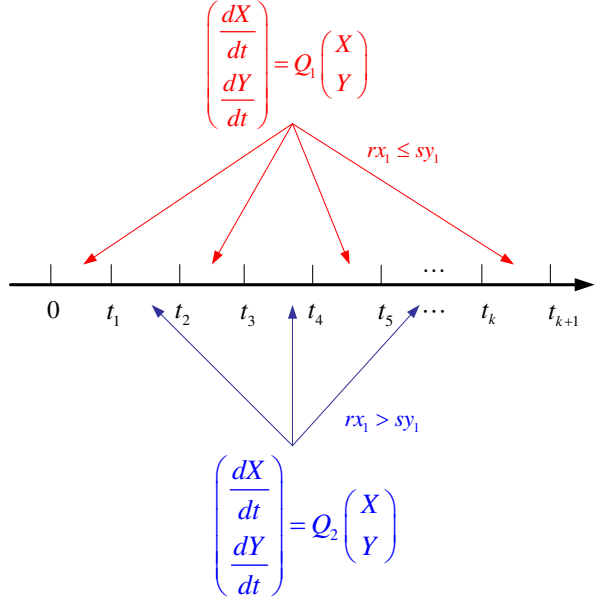}\\
  \caption{Illustration of derived ODEs with component types and one synchronisation}
  \label{fig:ChFA-Q1-Q2-Alternative}
  \end{center}
\end{figure}

Similar to the cases discussed in the previous two sections, the
convergence problem of (\ref{eq:ChFA-TwoType-OneSync-Q1Q2}) can be
divided into two subproblems, i.e. to examine whether the following
two properties hold:
\begin{enumerate}
  \item There exists a time $T$, such that either $x_1\leq y_1,\;\forall t>T$
  or $x_1\leq y_1,\;\forall t>T$.
  \item The eigenvalues of $Q_1$ and $Q_2$ other than zeros have
  negative real parts.
\end{enumerate}
The first item  can guarantee (\ref{eq:ChFA-TwoType-OneSync-Q1Q2})
to eventually have a constant linear form, while the second item
ensures the convergence of the bounded solution of the linear ODEs.
If the answers to these two problems are both positive, then the
convergence of the solution of (\ref{eq:ChFA-TwoType-OneSync-Q1Q2})
will hold. The study of these two problems are discussed in the next
two subsections. In the remainder of this subsection, we first
investigate the structure property of the coefficient matrices $Q_1$
and $Q_2$ in (\ref{eq:ChFA-TwoType-OneSync-Q1Q2}).

\par The structure of the coefficient matrices $Q_1$
and $Q_2$ is determined by the following two propositions, which
indicate that they are either block lower-triangular or block
upper-triangular.
\begin{proposition}\label{proposition: ChFQ-Q1-structure}
$Q_1$ in (\ref{eq:ChFA-TwoType-OneSync-Q1Q2})  can be written as
\begin{equation}
Q_1=\left(
  \begin{array}{cc}
    \hat{Q}_1 & 0 \\
    W & V_{n\times n} \\
  \end{array}
\right)_{(m+n)\times(m+n)},
\end{equation}
where $\hat{Q}_1^T$ is an infinitesimal generator matrix with the
dimension ${m\times m}$, and
\begin{equation}
W_{n\times m}=\left(
  \begin{array}{cccc}
               w_{11} & 0 & \cdots & 0 \\
               w_{21} & 0 & \cdots & 0 \\
               \vdots & \vdots & \vdots & \vdots \\
               w_{n1} & 0 & \cdots & 0
             \end{array}
\right),
\end{equation}
where $w_{11}<0$, $w_{j1}(j=2,\cdots,n)\geq0$ and
$\sum_{j=1}^nw_{j1}=0$. Here $V$ and $W$ satisfy that if we let
\begin{equation}
P=\left(W_1+V_1,V_2,\cdots,V_n\right),
\end{equation}
i.e. $P$'s first column is the sum of $V$'s first column and $W$'s
first column, while $P$'s other columns are the same to $V$'s other
columns, then $P^T$ is also an infinitesimal generator matrix.
\end{proposition}

\begin{proof}
Let
\begin{equation}
  Q_1=\left(
       \begin{array}{cc}
        \hat{Q}_1 & U\\
        W & V
       \end{array}
     \right),
\end{equation}
where $\hat{Q}_1$ and $V$ are $m\times m$ and $n\times n$ matrices
respectively. Suppose $rx_1(t)\leq sy_1(t)$, then
\begin{equation}\label{eq:ChFA-ProProof0}
\begin{split}
\left(
  \begin{array}{c}
   \frac{\mathrm{d}X}{\mathrm{d}t} \\
  \frac{\mathrm{d}Y}{\mathrm{d}t}\\
  \end{array}
\right) =Q_1 \left(
  \begin{array}{c}
   X \\
  Y\\
  \end{array}
\right)=\left(
       \begin{array}{cc}
        \hat{Q}_1 & U\\
        W & V
       \end{array}
     \right)\left(
  \begin{array}{c}
   X \\
  Y\\
  \end{array}
\right).
\end{split}
\end{equation}
So we have
 \begin{equation}\label{eq:ChFA-ProProof}
                \frac{\mathrm{d}X}{\mathrm{d}t}=\hat{Q}_1X+UY.
 \end{equation}
The condition $rx_1(t)\leq sy_1(t)$ implies that all nonlinear terms
$\min\{rx_1(t),sy_1(t)\}$ can be replaced by $x_1(t)$. This means
that the behaviour of the component type $X$ in
(\ref{eq:ChFA-ProProof0}) and (\ref{eq:ChFA-ProProof}) is
independent of the component type $Y$. Thus
in~(\ref{eq:ChFA-ProProof0}) $U$ must be a zero matrix, i.e.
\begin{equation*}
  Q_1=\left(
       \begin{array}{cc}
        \hat{Q}_1 & 0\\
        W & V
       \end{array}
     \right).
\end{equation*}
Moreover, (\ref{eq:ChFA-ProProof}) becomes
\begin{equation}\label{PropProof2}
                \frac{\mathrm{d}X}{\mathrm{d}t}=\hat{Q}_1X,
 \end{equation}
that is, there is no synchronisation in the ODEs corresponding to
the component type $X$ given $rx_1(t)\leq sy_1(t)$. Then by
Proposition~\ref{pro:ChFP-NonSynQmatrix}, $\hat{Q}_1^T$ is an
infinitesimal generator.

\par According to~(\ref{eq:ChFA-ProProof0}),
\begin{equation}\label{PropProof3}
 \begin{split}
         \frac{\mathrm{d}Y}{\mathrm{d}t}&=WX+VY\\
           &=(W_1,W_2,\cdots,W_m)(x_1,x_2,\cdots,x_m)^T+VY\\
           &=x_1W_1+VY+\sum_{i=2}^mx_iW_i,
 \end{split}
 \end{equation}
where $W=(W_1,W_2,\cdots,W_m)$. Notice that the component type $Y$
is synchronised with the component type $X$ only through the term
$\min\{rx_1(t),sy_1(t)\}=x_1(t)$. In other words,
in~(\ref{PropProof3}) $Y$ is directly dependent on only $x_1$ other
than $x_i\;(i\geq 2)$. This implies $W_1\neq 0$ while
$W_i=0\;(i=2,3,\cdots,m)$. Therefore,
\begin{equation}\label{PropProof4}
 \begin{split}
         \frac{\mathrm{d}Y}{\mathrm{d}t}&=x_1W_1+VY=x_1W_1+\sum_{j=1}^ny_jV_j,
 \end{split}
 \end{equation}
where $V_j\;(j=1,2,\cdots,n)$ are the columns of $V$. Denote
$W_{1}=(w_{11},w_{21},\cdots,w_{n1})^T$. Notice that $Y_1$ is a pre
local derivative of the shared activity, and $x_1w_{11}$ represents
the exit rates of the shared activity from $Y_1$. Therefore,
$w_{11}<0$. Moreover, $x_1w_{j1}$ $(j=2,\cdots,n)$ are the
synchronised entry rates for the local derivatives
$Y_j\;(j=2,\cdots,n)$ respectively, so
$w_{j1}\geq0\;(j=2,\cdots,n)$. By the conservation law, the total
synchronised exit rates are equal to the total synchronised entry
rates, i.e. $x_1\sum_{j=1}^nw_{j1}=0$ or $\sum_{j=1}^nw_{j1}=0$.

\par We have known that $x_1$ in~(\ref{PropProof4}) derives from the
synchronised term $\min\{rx_1,sy_1\}$. If the effect of the
synchronisation  on the behaviour of $Y$ is removed, i.e. recover
$y_1$ by replacing $x_1$, then (\ref{PropProof4}) will become
\begin{equation}\label{PropProof5}
 \begin{split}
         \frac{\mathrm{d}Y}{\mathrm{d}t}&=y_1W_1+VY=y_1W_1+\sum_{j=1}^ny_jV_j=PY,
 \end{split}
 \end{equation}
where $P=(W_1+V_1, V_2, \cdots, V_n)$. Since there is no
synchronisation contained in the subsystem of the component type
$Y$, according to Proposition~\ref{pro:ChFP-NonSynQmatrix}, $P^T$ is
the infinitesimal generator.
\end{proof}

Similarly, we can prove
\begin{proposition}\label{proposition: ChFQ-Q2-structure}
$Q_2$ in (\ref{eq:ChFA-TwoType-OneSync-Q1Q2}) can be written as
\begin{equation}
Q_2=\left(
  \begin{array}{cc}
    E_{m\times m} & F \\
    0 & \hat{Q}_2 \\
  \end{array}
\right)_{(m+n)\times(m+n)},
\end{equation}
where $\hat{Q}_2^T$ is an infinitesimal generator matrix with the
dimension ${n\times n}$, and
\begin{equation}
F_{m\times n}=\left(
  \begin{array}{cccc}
               f_{11} & 0 & \cdots & 0 \\
               f_{21} & 0 & \cdots & 0 \\
               \vdots & \vdots & \vdots & \vdots \\
               f_{m1} & 0 & \cdots & 0
             \end{array}
\right),
\end{equation}
where $f_{11}<0$, $f_{j1}(j=2,\cdots,m)\geq0$ and
$\sum_{j=1}^mf_{j1}=0$. Here $F$ and $E$ satisfy that if we let
\begin{equation}
R=\left(F_1+E_1,E_2,\cdots,E_n\right),
\end{equation}
then $R^T$ is also an infinitesimal generator matrix.
\end{proposition}

\subsubsection{Eigenvalues of coefficient matrix}
 In this subsection, we will determine the eigenvalue property
of $Q_1$ and $Q_2$. First, the Perron-Frobenius theorem  gives an
estimation of eigenvalues for  nonnegative matrices.
\begin{theorem}\label{thm:ChFA-Perron-Frobenius Theorem}(\textbf{Perron-Frobenius}).
Let $A=(a_{ij})$ be a real $n\times n$ matrix with nonnegative
entries $a_{ij}\geq0$. Then the following statements hold:
\begin{enumerate}
  \item There is a real eigenvalue $r$ of $A$ such that any other
  eigenvalue $\lambda$ satisfies $|\lambda|\leq r$.
  \item $r$ satisfies $\min\limits_i\sum_{j}a_{ij}\leq r\leq
  \max\limits_i\sum_{j}a_{ij}$.
  \end{enumerate}
\end{theorem}

\begin{remark}\label{remark:ChFA-PF-theorem}
We should point out that in the second property, exchanging $i$
 and $j$ in $a_{ij}$ in the formula,
  we still have
 $\min\limits_i\sum_{j}a_{ji}\leq r\leq
  \max\limits_i\sum_{j}a_{ji}$. In fact, $A^T$ is also a real matrix with non-negative entries.
  Since $A^T$ and $A$ share the same eigenvalues, so $r$ is one of
  the eigenvalues of $A^T$, such that  any other eigenvalue
  $\lambda$ of $A^T$ satisfies $|\lambda|\leq r$. Notice $(A^T)_{ij}=A_{ji}$,
  By applying the Perron-Frobenius theorem to $A^T$,
   we have $$\min\limits_i\sum_{j}a_{ji}\leq r\leq
  \max\limits_i\sum_{j}a_{ji}.$$
%\end{enumerate}
\end{remark}

We cannot directly apply this theorem to our coefficient matrices
$Q_1$ and $Q_2$, since both of them have negative elements, not only
on the diagonal but also in other entries. However, we use  some
well-known techniques in linear algebra, i.e. the following Lemma
\ref{lemma:ChFA-Eigen_Lemma_TrangleMatrix} and
\ref{lemma:ChFA-Eigen_Lemma_V+rI} (which can be easily found in
linear algebra textbooks), to cope with this problem, and thus
derive estimates of their eigenvalues.

\begin{lemma}\label{lemma:ChFA-Eigen_Lemma_TrangleMatrix}
If $E_{m\times m}$ and $F_{n\times n}$ have eigenvalues
$\lambda_i\;(i=1,2,\cdots,m)$ and $\delta_j$\linebreak
$(j=1,2,\cdots,n)$ respectively, then each of
$$H_1=\left(
  \begin{array}{cc}
    E & 0 \\
    G & F \\
  \end{array}
\right),\;  H_2=\left(
  \begin{array}{cc}
    E & G \\
    0 & F \\
  \end{array}
\right)
$$ has eigenvalues $\lambda_i\;(i=1,2,\cdots,m)$ and $
\delta_j\;(j=1,2,\cdots,n)$.
\end{lemma}

\begin{lemma}\label{lemma:ChFA-Eigen_Lemma_V+rI}
If $\lambda$ is an eigenvalue of $V$, then $\lambda+r$ is an
eigenvalue of $V+rI$, where $r$ is a scalar.
\end{lemma}

\begin{theorem}\label{thm:ChFA-Q1Q2-Eigen-NegativeRealParts}
The eigenvalues of both $Q_1$ and $Q_2$  are either zeros or have
negative real parts.
\end{theorem}

\begin{proof}
\par We only give the proof for $Q_1$'s case. By Proposition~\ref{proposition:
ChFQ-Q1-structure},
\begin{equation}
Q_1=\left(
  \begin{array}{cc}
    \hat{Q}_1 & 0 \\
    W & V_{n\times n} \\
  \end{array}
\right)_{(m+n)\times(m+n)}.
\end{equation}
According to Lemma~\ref{lemma:ChFA-Eigen_Lemma_TrangleMatrix}, if
all eigenvalues of $\hat{Q}_1$ and $V$ are determined, then the
eigenvalues of $Q_1$ will be determined. Let us consider $V$ first.

\par Notice that only diagonal elements of $V$ are possibly negative
 (which can be deduced from Proposition~\ref{proposition: ChFQ-Q1-structure}). Let
$r=\sup\limits_i|V_{ii}|>0$, then all the entries of $V+rI$ are
nonnegative. Let $\lambda$ be an arbitrary eigenvalue of $V$, then
by Lemma~\ref{lemma:ChFA-Eigen_Lemma_V+rI}, $\lambda+r$ is an
eigenvalue of $V+rI$.

\par Notice the sum of the elements of any column of
$V$ is zero (because the sum of entry rates equals to the sum of
exit rates), so $V$ has a zero eigenvalue with the corresponding
eigenvector $\mathbf{1}$, i.e. $V\mathbf{1}=\mathbf{0}$. Thus
$r=0+r$ is an eigenvalue of $V+rI$. Moreover,
$$
\min\limits_i\sum_{j}(V+rI)_{ji}=r=
  \max\limits_i\sum_{j}(V+rI)_{ji}.
$$
Applying the Perron-Frobenius theorem (Theorem
\ref{thm:ChFA-Perron-Frobenius Theorem}) and Remark
\ref{remark:ChFA-PF-theorem} to $V+rI$, so
\begin{equation}\label{eq:ChFA-Eigen_Inequality}
   |\lambda+r|\leq r.
\end{equation}
Let $\lambda=a+bi$, then (\ref{eq:ChFA-Eigen_Inequality}) implies
that $a\leq 0$, and if $a=0$ then $b=0$. In other words, $V$'s
eigenvalues are either zeros or have negative real parts.

\par Similarly, $\hat{Q}_1$'s
eigenvalues other than zeros have negative real parts. By
Lemma~\ref{lemma:ChFA-Eigen_Lemma_TrangleMatrix}, the eigenvalues of
$Q_1$ other than zeros have negative real parts. The proof is
complete.
\end{proof}

\subsubsection{Convergence theorem}
\label{subsec:ChFA-twoComType-OneSyn-ConvergenceTheorem} Now we deal
with another subproblem: whether or not after a long time, we always
have $rx_1\geq sy_1$ (or $rx_1< sy_1$). If the population of
 $X$ is significantly larger than the population of $Y$,
intuitively, there will finally be a greater number of $X$ in the
local derivative $X_1$, than the number of $Y$ in $Y_1$.  This will
lead to $rx_1>sy_1$.

\begin{lemma}\label{Lemma_LowerBound}
Under the assumptions in Section~\ref{Subsection_Features of
Matrix}, for the following ODEs
\begin{equation*}
\left(
  \begin{array}{c}
   \frac{\mathrm{d}X}{\mathrm{d}t} \\
  \frac{\mathrm{d}Y}{\mathrm{d}t}\\
  \end{array}
\right)=I_{\{rx_1\leq sy_1\}}Q_1 \left(
  \begin{array}{c}
   X \\
  Y\\
  \end{array}
\right)+I_{\{rx_1>sy_1\}}Q_2 \left(
  \begin{array}{c}
   X \\
  Y\\
  \end{array}
\right),
\end{equation*}
there exists $c_1>0, c_2>0,T>0$, such that $x_1(t)\geq c_1M$,
$y_1(t)\geq c_2N$ for any $t\geq T$, where  $c_1$ and $c_2$ are
independent of $M$ and $N$.
\end{lemma}
\begin{proof}
The proof is essentially the same as the proof of
Lemma~\ref{lem:ChFA-x1(t)>cM}. We only give the sketch of the proof
for $x_1(t)\geq c_1M$. By introducing new two functions $\alpha(t)$
and $\beta(t)$,
\begin{equation*}
\alpha(t)=\left\{\begin{array}{cc}
            \frac{\min\{rx_1(t),sy_1(t)\}}{rx_1(t)}, & x_1(t)\neq0, \\
            1, & x_1(t)=0,
          \end{array}\right.
\end{equation*}
$$
    \beta(t)=\frac1t\int_0^t \alpha(s)ds,
$$
the nonlinear term $\min\{rx_1(t),sy_1(t)\}$ equals
$r\alpha(t)x_1(t)$, and thus the ODEs associated with the subsystem
$X$ can be written as
\begin{equation}\label{eq:ChFA-local52}
\frac{\mathrm{d}X}{\mathrm{d}t}=A(t)X
\end{equation}
where $A(t)$ is related to $\alpha(t)$. The solution of
(\ref{eq:ChFA-local52}) is $X(t)=e^{tB(t)} X(0)$, where $B(t)$ is
defined by $\displaystyle B(t)=\frac1t\int_0^t A(s)ds$, and thus
$B(t)$ is related to $\beta(t)$.

Notice that according to
Theorem~\ref{thm:ChFA-Q1Q2-Eigen-NegativeRealParts} and its proof,
the eigenvalues of $B(t)$ other than zeros have negative real parts
for any $t>0$.  By a similar proof to
Corollary~\ref{corollary:ChFA-Lambda(beta(t))-LowerBound}, we have
\begin{equation}\label{eq:ChFA-General-TwoOne-Lambda-Estimation}
\Lambda=\inf_{t\geq0}\{-\Re(\lambda)\mid\lambda \;\mbox{is $B(t)$'s
non-zero eigenvalue}\}>0.
\end{equation}
This fact will lead to the conclusion that $X(t)$ can be
approximated by $\hat{X}(t)$, where $\hat{X}(t)$ is constructed
similarly to  the one in
Proposition~\ref{proposition:ChFA-X(t)-hatX(t)-UserProvider-NotExpl}.
Because a general $B(t)$ considered here may not be diagonalisable,
so the construction of $\hat{X}(t)$ is a little bit more
complicated. We detail the construction as well as  the proof of the
following result in~\ref{section:Appendix-Jordan-Form}: $$
           \lim_{t\rightarrow \infinity}\|X(t)-\hat{X}(t)\|=0.
$$
Then, by similar arguments to
Proposition~\ref{proposition:ChFA-Fact2-x1-LowerBound} and
Corollary~\ref{corollary:ChFA-hat(x1)>cM}, we can prove that
$\inf_{t>T}\hat{x}_1(t)\geq cM$, where $\hat{x}_1(t)$ is the first
entry of $\hat{X}(t)$. Then, by a similar proof to the proof of
Lemma~\ref{lem:ChFA-x1(t)>cM}, we can conclude that there exists a
number $c_1$ such that $x_1(t)\geq c_1M$ after a time $T$.
\end{proof}

\begin{lemma}\label{OnlyOneDominates}Under the assumptions of Lemma~\ref{Lemma_LowerBound},
if  $M>K_1N$ or $N>K_2M$, where constants $K_1>0$ and $K_2>0$ are
sufficiently large, then there exists $T$ such that $rx_1(t)\geq
sy_1(t),\;\forall t\geq T$ or  $rx_1(t)\leq sy_1(t),\;\forall t\geq
T$ respectively.
\end{lemma}
\begin{proof}
  By the boundedness of solutions, $0\leq x_1(t)\leq M$ and $0\leq y_1(t)\leq N$ for any $t$.
  Suppose $M>K_1N$. By Lemma~\ref{Lemma_LowerBound}, there exists $c,T>0$, $x_1(t)\geq cM,\;\forall t\geq T$.
  Since $K_1$ is assumed to be large enough such that $K_1\geq \frac{s}{rc}$,
  then $rcM>sN$. So for any $t>T$, we have
  $$rx_1(t)\geq rcM\geq sN\geq sy_1(t).$$
If $N>K_2M$, the proof is similar and omitted here.
\end{proof}

\par  Now we state our convergence theorem.
\begin{theorem}
If $M>K_1N$ or $N>K_2M$, where constants $K_1, K_2>0$ are
sufficiently large, then the solution of the derived ODEs
(\ref{eq:ChFA-TwoType-OneSync-Q1Q2}), i.e.
\begin{equation*}
\left(
  \begin{array}{c}
   \frac{\mathrm{d}X}{\mathrm{d}t} \\
  \frac{\mathrm{d}Y}{\mathrm{d}t}\\
  \end{array}
\right)=I_{\{rx_1\leq sy_1\}}Q_1 \left(
  \begin{array}{c}
   X \\
  Y\\
  \end{array}
\right)+I_{\{rx_1>sy_1\}}Q_2 \left(
  \begin{array}{c}
   X \\
  Y\\
  \end{array}
\right),
\end{equation*}
 converges to a finite limit as time goes to
infinity.
\end{theorem}

\begin{proof} Suppose $M>K_1N$, then by Lemma
\ref{OnlyOneDominates}, there exists a time $T>0$, such that after
time $T$, $rx_1(t)\geq sy_1(t)$, so
(\ref{eq:ChFA-TwoType-OneSync-Q1Q2}) becomes
\begin{equation}
\left(
  \begin{array}{c}
   \frac{\mathrm{d}X}{\mathrm{d}t} \\
  \frac{\mathrm{d}Y}{\mathrm{d}t}\\
  \end{array}
\right)=Q_2 \left(
  \begin{array}{c}
   X \\
  Y\\
  \end{array}
\right).
\end{equation}
Since $Q_2$'s eigenvalues other than zeros have strict negative real
parts according to
Theorem~\ref{thm:ChFA-Q1Q2-Eigen-NegativeRealParts}, and the
solution of the above equation is bounded, then by
Corollary~\ref{Corollary D.2.1.}, the solution converges to a finite
limit as time goes to infinity.
\end{proof}

\begin{remark}
We should point out that for a general PEPA model with two component
types and one synchronisation, the limit of the solution of the
derived ODEs is determined by the populations of these two component
types rather than the particular starting state. That is to say,
whatever the initial state is, as long as the total populations of
the components are the same, the limit that the solution converges
to will always be the same. We do not plan to discuss this topic in
detail in this paper.
\end{remark}

\section{Related work}

As we have mentioned in the introduction section, the fluid
approximation approach was first proposed by Hillston
in~\cite{Jane2} to deal with the state space explosion problem in
the context of PEPA. An interpretation as well as a justification of
this approximation approach has been demonstrated by Hayden  in his
dissertation~\cite{Richard-Undergraduate-Thesis}.
In~\cite{Richard-Undergraduate-Thesis,Richard-general-moment-2008},
generation of similar systems of coupled ODEs for higher-order
moments such as variance has been addressed. Additionally, the
dissertation~\cite{Richard-Undergraduate-Thesis} discusses how to
derive stochastic differential equations from PEPA models.

More recently, some extensions of the previous mapping from PEPA to
ODEs have been presented by Bradley \emph{et al}.
in~\cite{WormAttacks}. In particular, passive rates are introduced
into the fluid approximation. In the recent
paper~\cite{Richard-passive-fluid-semantics}, different existing
styles of passive cooperation in fluid models are compared and
intensively discussed. Moreover, a new passive fluid semantics for
passive cooperation, which can be viewed as approximating the first
moments of the component counting processes, has been provided, with
a theoretical justification. The
paper~\cite{SB-Improved-ContinuousApproximation-Epidemiological}
considers the application of this fluid approximation approach with
modifications in the context of epidemiology. In this paper, the
notions of side and self-loops are added to the activity matrix, and
the rates are calculated differently, for the purpose of deriving
from PEPA models the most commonly used ODEs in the context of
epidemiology. In~\cite{Mirco-2009} and \cite{Mirco-Fluid-Semantics}
by Tribastone, a new operational semantics is proposed to give a
compact symbolic representation of PEPA models. This semantics
extends the application scope of the fluid approximation of PEPA by
incorporating all the operators of the language and removing earlier
assumptions on the syntactical structure of the models amenable to
this analysis.

\par The fluid approximation approach has also been applied to timed Petri
nets to deal with the state space explosion
problem~\cite{Silva05-Contunuization-Timed-Petri-Nets,Silva06-continuous-Petri-nets}.
The comparison between the fluid approximation of PEPA models and
timed continuous Petri nets has been demonstrated by Galpin
in~\cite{Vashti08-ODEs-PEPA-Petri}. This paper has established links
between two continuous approaches to modelling the performance of
systems. In the paper, a translation from PEPA models to continuous
Petri nets and \emph{vice versa} has been presented. In addition, it
has been shown that the continuous approximation using PEPA has
infinite server semantics.  The fluid approximation approach has
also been used by Thomas to derive asymptotic solutions for a class
of closed queueing
networks~\cite{Nigel-ODEs-PEPA-Queueing-Networks}. In this paper, an
analytical solution to a class of models, specified using PEPA, is
derived through the ODE approach. It is shown that ``this solution
is identical to that used for many years as an asymptotic solution
to the mean value analysis of closed queueing networks''.

\par Moreover, the relationship between the fluid approximation and the
underlying CTMCs for a special PEPA model has been revealed by
Geisweiller \emph{et al}. in \cite{Jane4}: the ODEs derived from the
PEPA description are the limits of the sequence of underlying CTMCs.
It has been shown in~\cite{Gilmore2005} by Gilmore that for some
special examples the equilibrium points of the ODEs derived from
PEPA models coincide the steady-state probability distributions of
the CTMCs underlying the nonsynchronised PEPA models.

\par In addition, there are several papers which discuss how to derive
response time from the fluid approximation of PEPA models.
In~\cite{Richard-responsetimes-fluidanalysis}, by constructing an
absorption operator for the PEPA language, Bradley \emph{et al}.
allow general PEPA models to be analysed for fluid-generated
response times. Clark \emph{et al}. demonstrate
in~\cite{ResponseTimeODE} how to derive expected passage response
times using Little's law based on averaged populations of entities
in an equilibrium state. This technique has been generalised into
one for obtaining a full response-time profile computing the
probability of a component observing the completion of a response at
a given time after the initiation of the request,
see~\cite{Allan-Response-Time-ODEs}. Moreover, an error in the
passage specification in the approach taken
in~\cite{Richard-responsetimes-fluidanalysis} has been uncovered and
rectified in~\cite{Allan-Response-Time-ODEs} by Clark. The ODE
method associated with the PEPA language has demonstrated successful
application in the performance analysis of large scale
systems~\cite{Harrison-Ubiquitous,Zhaoyishi-Approximate,Zhaoyishi-Fluid,Zhaoyishi-Efficient,JieDing}.

\section{Conclusions}

In this paper, we have demonstrated how to derive the fluid
approximation from a general PEPA model via the numerical
representation of the PEPA language, which extended the current
mapping semantics of fluid approximations.  The fundamental
properties of the fluid approximation such as the existence and
uniqueness, boundedness and nonnegativeness of the solutions of the
derived ODEs have been established. Moreover, the convergence of the
solution under a particular condition for general models has been
verified. This particular condition relates some famous constants of
Markov chains such as the spectral gap and the Log-Sobolev constant.
For the models without synchronisations or with one synchronisation
and two component types, we have determined the convergence.
Furthermore, the consistency between the fluid approximation and the
CTMC in the context of PEPA has been revealed. If a model has no
synchronisations, then the derived ODEs are just the probability
distribution evolution equations of the underlying CTMC except for a
scaling factor. For any general PEPA model, the ODEs can be taken as
the corresponding density dependent CTMC with the concentration
level infinity. In addition, the coefficient matrices of some
derived ODEs were studied: their eigenvalues are either zeros or
have negative real parts.  The structural property of invariance has
been shown to play an important role in the proof of convergence for
some PEPA models. Due to the limitation of pages, we have not shown
how to derive performance measurers from the fluid approximation,
and have not given the numerical comparison between the fluid
approximation and the CTMC in terms of performance measures. For
more details, please refer to~\cite{JieThesis}.

\par In addition to the established results, this paper has also demonstrated
comprehensive techniques and methods to investigate the PEPA
language: not only both theoretical and experimental, probabilistic
and analytic, but also syntactic and numerical (investigating the
models based on the numerical representation of PEPA), qualitative
and quantitative (exploiting the structural property to verify the
convergence). These results and techniques are expected to have more
applications in performance modelling and evaluation of large scale
systems.

\section*{Acknowledgements}

Partial work of the first author was carried out in LFCS, School of
Informatics and IDCom, School of Engineering, The University of
Edinburgh,  when he was a PhD student and supported by the Mobile
VCE (www.mobilevce.com) Programme 4.

\bibliographystyle{elsarticle-num}
\bibliography{JieDing_August2010-NoDOI}

\appendix

%\section{Some theorems}\label{section:Appendix-Some-Theorems}
\section{}\label{section:Appendix-Some-Theorems}

\begin{theorem}(\textbf{Kurtz theorem}~\cite{KurtzBook}, page 456).
Let $X_n$ be a family of density dependent CTMCs with the
infinitesimal generators
$$
     q^{(n)}_{k,k+l}=nf(k/n,l),
$$
where $f(x,l)\;(x\in E\subset R^h,\; l\in \mathbb{Z}^h)$ is a
continuous function, $k$ is a numerical state vector and $l$ is a
transition vector.

\par Suppose $X(t)\in E$ satisfies
 \begin{equation*}
    \frac{\mathrm{d}x}{\mathrm{d}t}=F(x)
\end{equation*}
where $F(x)=\sum_llf(x,l).$ Suppose that for each compact $K\subset
E$,
\begin{equation}
       \sum_l\|l\|\sup_{x\in K}f(x,l)< \infinity
\end{equation}
and there exists $M_K>0$ such that
\begin{equation}\label{eq: Kurtz Lipschitz}
         \|F(x)-F(y)\|\leq M_K\|x-y\|,\quad x,y\in K.
\end{equation}

If $\lim_{n\rightarrow\infinity}\frac{X_n(0)}{n}=x_0$, then for
every $t\geq 0$,
\begin{equation}
        \lim_{n\rightarrow\infinity}\sup_{s\leq
        t}\left\|\frac{X_n(s)}{n}-X(s)\right\|=0\quad  a.s.
\end{equation}
\end{theorem}

The following lemma can be found in any good book on differential
calculus.
\begin{lemma}\label{Lemma:Differential-Inequality} Let $y(t)$ be a differentiable function
defined
    for $t\geq 0$. Suppose $a,b\in \mathbb{R}$, $a\neq0$.
If $y(t)$ satisfies
$\displaystyle\frac{\mathrm{d}y}{\mathrm{d}t}\geq ay(t)+b,\;t>0$,
    then
    $$
                 y(t)\geq e^{at}\left(y(0)+\frac ba\right)-\frac ba.
    $$
Similarly, if $y(t)$ satisfies
$\displaystyle\frac{\mathrm{d}y}{\mathrm{d}t}\leq ay(t)+b,\;t>0$,
    then
    $$
                 y(t)\leq e^{at}\left(y(0)+\frac ba\right)-\frac ba.
    $$
\end{lemma}

\begin{proof} Let $W(t)=y(t)e^{-at}$, then
$\displaystyle\frac{\mathrm{d}W}{\mathrm{d}t}=e^{-at}\left(\frac{\mathrm{d}y}{\mathrm{d}t}-ay\right)\geq
be^{-at}$. Integrating on both sides, so $
        W(t)-W(0)\geq b\int_0^te^{-as}ds$.
Thus $y(t)e^{-at}-y(0)\geq\frac{b}{a}(1-e^{-at})$. \linebreak So $
y(t)\geq e^{at}\left(y(0)+\frac ba\right)-\frac ba$. The second
conclusion can be similarly proved.
\end{proof}

\begin{theorem}\label{Theorem D.1.1.}
(\textbf{Fundamental Inequality}, \cite{DifferentialEquations2},
page 14). If
$\displaystyle\frac{\mathrm{d}\mathbf{x}}{\mathrm{d}t}=\mathbf{f}(\mathbf{x},t)$
is defined on a set $U$ in $\mathbb{R}^n\times \mathbb{R}$ with the
Lipschitz condition
$$
      \|\mathbf{f}(\mathbf{x}_1,t)-\mathbf{f}(\mathbf{x}_2,t)\|<K\|\mathbf{x}_1-\mathbf{x}_2\|
$$
for all $(\mathbf{x}_1,t)$ and $(\mathbf{x}_2,t)$ on $U$, and if for
$\epsilon_i, \delta\in \mathbb{R}$, and $\mathbf{u}_1(t)$ and
$\mathbf{u}_2(t)$ are two continuous, piecewise differentiable
functions on $U$ into $\mathbb{R}^n$ with $$
   \left\|\frac{\mathrm{d}\mathbf{u}_i(t)}{\mathrm{d}t}-\mathbf{f}(\mathbf{u}_i(t),t)\right\|\leq \epsilon_i,
\quad\mbox{and}\quad
    \|\mathbf{u}_1(t_0)-\mathbf{u}_2(t_0)\|\leq \delta,
$$
then
$$
   \|\mathbf{u}_1(t)-\mathbf{u}_2(t)\|\leq \delta e^{K(t-t_0)}+\left(\frac{\epsilon_1+\epsilon_2}{K}\right)
   \left(e^{K(t-t_0)}-1\right).
$$
\end{theorem}

%\section{Spectral gaps and Log-Sobolev constants of Markov
%chains}\label{section:Appendix-spectral-gap-log-sobolev}
\section{}\label{section:Appendix-spectral-gap-log-sobolev}

The material presented here is extracted
from~\cite{SCoste-FiniteMarkovChains}. Let $(K,\pi)$ be a Markov
chain on a finite set $S$, where $K$ is a Markov kernel and $\pi$ is
the  stationary probability distribution associated with $K$. For
any real function $f, g$ on $S$, define an inner product
``$\langle\cdot,\cdot\rangle$'' as
$$
\langle f, g\rangle=\sum_{x\in S}f(x)g(x)\pi(x).
$$
Denote $\|f\|_2=\sqrt{\langle f, f\rangle}$, and
$$
    l^2(\pi)=\{f: \|f\|_2< \infinity\}.
$$
Then $l^2(\pi)$ is a Hilbert space with the norm $\|\cdot\|_2$. We
say that $K^{*}$ is \emph{adjoint} to $K$ if
$$
       \langle Kf,g\rangle=\langle f,K^{*}g\rangle, \quad \forall
       f,g\in L^2(\pi).
$$
It follows that
$$
     K^{*}(x,y)=\frac{\pi(y)}{\pi(x)}K(y,x).
$$
If $K=K^{*}$, then $K$ is called \emph{self-adjoint}. If $K$ is
self-adjoint on $l^2(\pi)$, then $(K,\pi)$ is \emph{reversible}.

For a function in $f\in l^2(\pi)$, denote its mean and variance by
$\pi(f)$ and $Var(f)$ respectively, that is
$$
  \pi(f)=\sum_{x\in S}f(x)\pi(x),\quad Var(f)=\pi((f-\pi(f))^2).
$$
\begin{definition}\label{Definition C.2.1} (\textbf{Dirichlet form}). The form
$$
    \mathcal{E}(f,g)=\langle (I-K)f, g\rangle
$$
is called the Dirichlet form associated with $H_t=e^{-t(I-K)}$.
\end{definition}

\begin{remark}\label{Remark C.2.1}
 The Dirichlet form
$\mathcal{E}$ satisfies
$$
  \mathcal{E}(f,f)=\langle (I-K)f, f\rangle=\left\langle \left(I-\frac{K+K^{*}}{2}\right)f,
  f\right\rangle,
$$
$$
 \mathcal{E}(f,f)=\frac12\sum_{x,y}(f(x)-f(y))^2K(x,y)\pi(x).
$$
\end{remark}

\begin{definition}\label{Definition C.2.2}(\textbf{Spectral gap}). Let $K$ be a
Markov kernel with Dirichlet form $\mathcal{E}$. The spectral gap
$\lambda=\lambda(K)$ is defined by
$$
      \lambda=\min\left\{
         \frac{\mathcal{E}(f,f)}{Var(f)}: Var(f)\neq 0
      \right\}.
$$
\end{definition}

\begin{remark}\label{Remark C.2.2}
 In general $\lambda$ is the smallest non zero
eigenvalue of $I-\frac{K+K^{*}}{2}$. If $K$ is self-adjoint, then
$\lambda$ is the smallest non zero eigenvalue of $I-K$. Clearly, we
also have
$$
      \lambda=\min\left\{\mathcal{E}(f,f)
        : \|f\|_2=1, \pi(f)=0
      \right\}.
$$
\end{remark}

\par The definition of the logarithmic Sobolev (Log-Sobolev)
constant $\alpha$ is similar to that of the spectral gap $\lambda$
where the variance has been replaced by
$$
      \mathcal{L}(f)=\sum_{x\in
      S}f(x)^2\log\left(\frac{f(x)^2}{\|f\|_2^2}\right)\pi(x).
$$

\begin{definition}\label{Definition C.2.2}(\textbf{Log-Sobolev constant}). Let $K$
be an irreducible Markov chain with stationary measure $\pi$. The
logarithmic Sobolev constant $\alpha=\alpha(K)$ is defined by
$$
      \alpha=\min\left\{
         \frac{\mathcal{E}(f,f)}{\mathcal{L}(f)}: \mathcal{L}(f)\neq 0
      \right\}.
$$
\end{definition}

\begin{lemma}\label{Lemma C.2.1} For any finite Markov chain $K$ with
stationary measure $\pi$, the Log-Sobolev constant $\alpha$ and the
spectral gap $\lambda$ satisfy
$$
     \frac{1-2\pi(*)}{\log[1/\pi(*)-1]}\lambda\leq \alpha
     \leq \frac{\lambda}{2},
$$
where $\pi(*)=\min_{x}\pi(x)$.
\end{lemma}

%\section{Eigenvalue of coefficient matrices of
%Model~\ref{model:X-Y}}\label{section:Appendix-Eigenvalues-Matrices}

\section{}\label{section:Appendix-Eigenvalues-Matrices}

\par In this appendix, we claim that all
eigenvalues of $Q_i\;(i=1,2,3,4)$ appearing in
(\ref{eq:ChFA-Model-XY-ODEQ1Q2Q3Q4}) in
Section~\ref{subsec:ChFA-Analytica-proof-InterestingModel} other
than zeros have negative real parts. We do not worry about $Q_1$ and
$Q_4$ since they are lower or upper block triangular matrices and
the eigenvalues of this kind of matrices can be well estimated: all
eigenvalues of $Q_1$ and $Q_4$ are either zeros or have negative
real parts. All that we want to do here is to show that both $Q_2$
and $Q_3$ also have this property.

\par By symbolic calculation using Matlab,  $Q_3$'s eigenvalues are
$$\lambda_{1,2,3}=0 (\mbox{three folds}), \lambda_4=-c_4-c_3,$$
$$
 \lambda_5=-\frac12(a_1+a_2+c_1+c_2)+\frac12\sqrt{(a_1-a_2+c_1+c_2)^2-4a_1c_2},
$$
$$
 \lambda_6= -\frac12(a_1+a_2+c_1+c_2)-\frac12\sqrt{(a_1-a_2+c_1+c_2)^2-4a_1c_2}.
$$
If $(a_1-a_2+c_1+c_2)^2-4a_1c_2<0$, then the real parts of
$\lambda_5$ and $\lambda_6$ are $-\frac12(a_1+a_2+c_1+c_2)$, which
is negative. Otherwise, $$(a_1-a_2+c_1+c_2)^2-4a_1c_2\geq0.$$ In
this case,
$$(a_1-a_2+c_1+c_2)^2-4a_1c_2\leq(a_1-a_2+c_1+c_2)^2<(a_1+a_2+c_1+c_2)^2,$$
so $$
-\frac12(a_1+a_2+c_1+c_2)+\frac12\sqrt{(a_1-a_2+c_1+c_2)^2-4a_1c_2}<0.
$$
This means that $\lambda_5$ and $\lambda_6$ are both negative. Thus
$\lambda_i\;(i=1,2,\cdots,6)$ are either $0$ or have negative real
parts.

\par Similarly, $Q_2$'s eigenvalues are $\delta_{1,2,3}=0$,
$\delta_4=-c_1-c_2$,
$$
 \delta_5=-\frac12(a_1+a_2+c_3+c_4)+\frac12\sqrt{(a_2-a_1+c_3+c_4)^2-4a_2c_3},
$$
$$
 \delta_6=-\frac12(a_1+a_2+c_3+c_4)-\frac12\sqrt{(a_2-a_1+c_3+c_4)^2-4a_2c_3}.
$$
By similar argument, we still have that
$\delta_{i}\;(i=1,2,\cdots,6)$ are either zeros or have negative
real parts.

%\section{Some background theories and results used for analytic approach}
%\label{section:Appendix-Jordan-Form}
\section{}\label{section:Appendix-Jordan-Form}

%\subsection{The Jordan Canonical Form}
\subsection{}

In this subsection, we use $\Re(z)$ and $\Im(z)$ to respectively
represent the real and imaginary parts of a complex number $z$. The
following is mainly extracted from~\cite{DynamcalSystem_good} (page
39$\thicksim$42).

%\par \textbf{Theorem D.2.1. (The Jordan Canonical Form).}
\begin{theorem}\label{Theorem D.2.1.}(\textbf{The Jordan Canonical Form}).
Let $A$ be a real matrix with real eigenvalues $\lambda_j$,
$j=1,\cdots,k$ and complex eigenvalues $\lambda_j=a_j+ib_j$ and
$\bar{\lambda}_j=a_j-ib_j$, \linebreak $j=k+1,\cdots,n$. Then there
exists a basis
$\{\mathbf{v}_1,\cdots,\mathbf{v}_k,\mathbf{v}_{k+1},\mathbf{u}_{k+1},
\cdots,\mathbf{v}_n,\mathbf{u}_n\}$ for $\mathbb{R}^{2n-k}$ where
$\mathbf{v}_j$, $j=1,\cdots,k$ and $\mathbf{w}_j$, $j=k+1,\cdots,n$
are generalized eigenvectors of $A$, \linebreak
$\mathbf{u}_j=\Re(\mathbf{w}_j)$ and
$\mathbf{v}_j=\Im(\mathbf{w}_j)$ for $j=k+1,\cdots,n$, such that the
matrix \linebreak
$P=\{\mathbf{v}_1,\cdots,\mathbf{v}_k,\mathbf{v}_{k+1},\mathbf{u}_{k+1},\cdots,\mathbf{v}_{n},\mathbf{u}_{n}\}$
is invertible and
$$
   P^{-1}AP=\left[
           \begin{array}{ccc}
             B_1 &  &  \\
              & \ddots &  \\
              &  & B_r
           \end{array}
   \right]
$$
where the elementary Jordan blocks $B=B_j$, $j=1,\cdots,r$ are
either of the form
\begin{equation}\label{eq:ChFA-local-Appendix-1}
B=\left[
           \begin{array}{ccccc}
             \lambda & 1 & 0 & \cdots & 0 \\
                0 &\lambda & 1 & \cdots & 0 \\
               \cdots &  &  && \\
             0 & \cdots & & \lambda &1\\
             0 & \cdots & &0 &\lambda
           \end{array}
   \right]
\end{equation}
for $\lambda$ one of the real eigenvalues of $A$ or of the form
\begin{equation}\label{eq:ChFA-local-Appendix-2}
B=\left[
           \begin{array}{ccccc}
             D & I_2 & 0 & \cdots & 0 \\
                0 & D & I_2 & \cdots & 0 \\
               \cdots &  &  && \\
             0 & \cdots & & D & I_2\\
             0 & \cdots & &0 & D
           \end{array}
   \right]
\end{equation}
with
$$
     D=\left[
           \begin{array}{cc}
             a & -b \\
             b & a
           \end{array}
     \right],\quad
   I_2=\left[
           \begin{array}{cc}
             1 & 0 \\
             0 & 1
           \end{array}
     \right]\quad \mbox{and}\quad
  0=\left[
           \begin{array}{cc}
             0 & 0 \\
             0 & 0
           \end{array}
     \right]
$$
for $\lambda=a+ib$ one of the complex eigenvalues of $A$.
\end{theorem}

\par The Jordan canonical form of $A$ yields some explicit
information about the form of $\mathbf{x}=e^{At}\mathbf{x}_0$, i.e.
the solution of the initial value problem
\begin{equation}\label{eq:ChFA-ClassicalODEsTheory}
\left\{
\begin{split}
\frac{\mathrm{d}\mathbf{x}}{\mathrm{d}t}=& A\mathbf{x} \\
\mathbf{x}(0)= & \mathbf{x}_0
\end{split}
\right.
\end{equation}
That is,
\begin{equation}\label{eq:ChFA-Appendix-x(t)=Pexp(Bt)P^{-1}}
   \mathbf{x}(t)=P\,\mathrm{diag}\left[e^{B_jt}\right]P^{-1}\mathbf{x}_0,
\end{equation}
where $B_j$ are the elementary Jordan blocks of $A$, $j=1,\cdots,r$.
Here $\mathrm{diag}\left[e^{B_jt}\right]$ represents
$$
\mathrm{diag}\left[e^{B_jt}\right]=\left(
                                     \begin{array}{cccc}
                                       e^{B_1t} & 0 & \cdots & 0 \\
                                       0 &  e^{B_2t} &\cdots & 0 \\
                                       \vdots & \vdots & \ddots & \vdots \\
                                       0 & 0 & \cdots & e^{B_rt} \\
                                     \end{array}
                                   \right).
$$
In the following, the notation $\mathrm{diag}[\cdot]$ indicates the
similar meaning. If $B_j=B$ is an $m\times m$ matrix of the form
(\ref{eq:ChFA-local-Appendix-1}) and $\lambda$ is a real eigenvalue
of $A$ then
\begin{equation}\label{eq:ChFA-Appendix-e^{Bt}-RealEigenvalue}
e^{Bt}=e^{\lambda t}\left[
\begin{array}{ccccc}
  1 & t & t^2/2! & \cdots & t^{m-1}/(m-1)! \\
  0 & 1 & t & \cdots & t^{m-2}/(m-2)!  \\
  0 & 0 & 1 & \cdots & t^{m-3}/(m-3)!  \\
  \cdots &  &  &  &  \\
  0 & \cdots &  & 1 & t\\
 0 & \cdots &  & 0 & 1\\
\end{array}
\right].
\end{equation}
If $B_j=B$ is an $2m\times 2m$ matrix of the form
(\ref{eq:ChFA-local-Appendix-2}) and $\lambda=a+ib$ is a complex
eigenvalue of $A$, then
\begin{equation}\label{eq:ChFA-Appendix-e^{Bt}-ComplexEigenvalue}
e^{Bt}=e^{at}\left[
\begin{array}{ccccc}
  R & Rt & Rt^2/2! & \cdots & Rt^{m-1}/(m-1)! \\
  0 & R & Rt & \cdots & Rt^{m-2}/(m-2)!  \\
  0 & 0 & R & \cdots & Rt^{m-3}/(m-3)!  \\
  \cdots &  &  &  &  \\
  0 & \cdots &  & R & Rt\\
 0 & \cdots &  & 0 & R\\
\end{array}
\right]
\end{equation}
where $R$ is the rotation matrix
$$
  R=\left[
             \begin{array}{cc}
               \cos bt & -\sin bt \\
               \sin bt & \cos bt
             \end{array}
   \right].
$$

\begin{theorem}\label{Theorem D.2.2.}
If $\mathbf{x}(t)$ satisfies (\ref{eq:ChFA-ClassicalODEsTheory}),
then each coordinate in $\mathbf{x}(t)$ is a linear combination of
functions of the form
$$
      t^{k}e^{at}\cos bt \quad \mbox{or} \quad t^{k}e^{at}\sin bt
$$
where $\lambda=a+ib$ is an eigenvalue of the matrix $A_{n\times n}$
and $0\leq k\leq n-1$.
\end{theorem}

%\par \textbf{Corollary D.2.1.}

\begin{corollary}\label{Corollary D.2.1.}
If the eigenvalues of $A$ are either zeros or have negative real
parts, and $\mathbf{x}(t)$ is bounded in $[0,\infinity)$, then
$\mathbf{x}(t)$ converges to a finite limit as time goes to
infinity.
\end{corollary}
\begin{proof}
The solution is composed of the terms like
 $t^{k}e^{at}\cos bt$ and $t^{k}e^{at}\sin bt$.
If $a<0$, then $t^{k}e^{at}\cos bt$ and $t^{k}e^{at}\sin bt$
converge as time goes to infinity. If $a=b=0$, we will see $k=0$. In
fact, in this case $t^{k}e^{at}\cos bt=t^{k}$. If $k>0$, then this
term $t^{k}$ in the solution will make the solution unbounded as $t$
tends to infinity. So $k$ must be zero in the terms corresponding to
$a=b=0$. Thus, $t^{k}e^{at}\cos bt=1$ and $t^{k}e^{at}\sin bt=0$. So
the solution converges.
\end{proof}

%\subsection{Some obtained results}
\subsection{}

As the above subsections illustrate, the following problem
\begin{equation*} \left\{
\begin{split}
\frac{\mathrm{d}\mathbf{x}}{\mathrm{d}t}=& A\mathbf{x} \\
\mathbf{x}(0)= & \mathbf{x}_0
\end{split}
\right.
\end{equation*}
has solution $\mathbf{x}(t)=e^{At}\mathbf{x}_0$, which equals
\begin{equation}\label{eq:ChFA-Appendix-local-3}
   \mathbf{x}(t)=P\,\mathrm{diag}\left[e^{B_jt}\right]P^{-1}\mathbf{x}_0,
\end{equation}
where $B_j$ are the elementary Jordan blocks of $A$, $j=1,\cdots,r$.
Suppose the rank of $A_{n\times n}$ is $n-1$. This means that zero
is a one fold eigenvalue of $A$.

\par According to (\ref{eq:ChFA-Appendix-local-3}), we construct a
corresponding $\hat{\mathbf{x}}$ in the form
\begin{equation}\label{eq:ChFA-Appendix-local-4}
   \hat{\mathbf{x}}=P\,\mathrm{diag}\left[\hat{B}_j\right]P^{-1}\mathbf{x}_0,
\end{equation}
where $\hat{B}_j(t)$ is defined as follows. If $e^{B_jt}$ in
(\ref{eq:ChFA-Appendix-local-3}) has the form of
(\ref{eq:ChFA-Appendix-e^{Bt}-ComplexEigenvalue}), then $\hat{B}_j$
is defined by
\begin{equation}\label{eq:ChFA-Appendix-Bj(t)-ComplexEigenvalue}
\hat{B}_j=\mathbf{0}_{2m\times 2m}.
\end{equation}

If $e^{B_jt}$ has the form of
(\ref{eq:ChFA-Appendix-e^{Bt}-RealEigenvalue}), and the
corresponding real eigenvalue $\lambda<0$, then $\hat{B}_j$ is
defined by
\begin{equation}\label{eq:ChFA-Appendix-Bj(t)-RealEigenvalue}
\hat{B}_j=\mathbf{0}_{m\times m}
\end{equation}
If $\lambda=0$, we know that  zero is a one fold eigenvalue of $A$
due to its rank $n-1$. So $m=1$. Then,
\begin{equation}\label{eq:ChFA-Appendix-Bj(t)-RealEigenvalue}
\hat{B}_j=1.
\end{equation}
In short, only for the zero eigenvalue is $\hat{B}_j$ set to one,
otherwise it is set to zeros. Clearly, we have
%\par \textbf{Lemma D.2.1}
\begin{lemma}\label{Lemma D.2.1.}
If zero is a one fold eigenvalue of $A$ and all other eigenvalues of
$A$ have negative real parts, then
$$
    \lim_{t\rightarrow\infinity}|\hat{\mathbf{x}}(t)-\mathbf{x}(t)|=0.
$$
\end{lemma}

\begin{proof}
If $B$ is the Jordan block corresponding to the one fold zero
eigenvalue, then according to
(\ref{eq:ChFA-Appendix-Bj(t)-RealEigenvalue}), $e^{Bt}=e^{0}=1$, and
 $e^{Bt}-\hat{B}=1-1=0$. For any non-zero eigenvalue, since $\hat{B}=0$ then
 $e^{Bt}-\hat{B}=e^{Bt}$. Notice that by (\ref{eq:ChFA-Appendix-e^{Bt}-RealEigenvalue}) and
(\ref{eq:ChFA-Appendix-e^{Bt}-ComplexEigenvalue}),
$$\|e^{Bt}\|\leq C_1(t)e^{-\Lambda t},$$
where $C_1(t)$ is a polynomial of $t$ with the maximum order $k$,
and
$$\Lambda=\inf\{-\Re(\lambda)\mid
\lambda \; \mbox{is non-zero eigenvalue of}\; A\}>0.
$$
Therefore,
\begin{equation}\label{eq:ChFA-Appendix-x(t)-hat(x)-Estimation}
     \|\mathbf{x}(t)-\hat{\mathbf{x}}\|\leq
     \sum_{B}\|e^{Bt}-\hat{B}\|\leq C_2(t)e^{-\Lambda t}\longrightarrow
     0
\end{equation}
as $t$ goes to infinity, where $C_2(t)$ is a polynomial function of
$t$.
\end{proof}

\par The above construction can be extended to the  problem of
\begin{equation}\label{eq:ChFA-Appendix-local-10}
\frac{\mathrm{d}X}{\mathrm{d}t}=A(t)X
\end{equation}
with the initial value $X(0)$, which is discussed in
Section~\ref{subsec:ChFA-twoComType-OneSyn-ConvergenceTheorem}. The
solution is
$$X(t)=e^{tC(t)} X(0),$$
where $\displaystyle C(t)=\frac1t\int_0^t A(s)ds$. We can similarly
define a function $f$ such that $C(t)=f(\beta(t))$, where $\beta(t)$
is similarly defined according to $A(t)$. Therefore,
$$X(t)=e^{tC(t)} X(0)=e^{tf(\beta(t))}X(0).$$

\par For a fixed
$\beta$, %$X(t)$ can be similarly written as
\begin{equation}\label{eq:ChFA-Appendix-local-4-0}
e^{tf(\beta)}X(0)=P(\beta)\,\mathrm{diag}\left[e^{B_j(\beta)t}\right]P(\beta)^{-1}X(0),
\end{equation}
where $B_j(\beta)$ are the elementary Jordan blocks of $f(\beta)$,
$j=1,\cdots,r(\beta)$. Repeating the previous construction process
with $B_j(\beta)$ for each $j$, we obtain the constructed matrix
$\hat{B}(\beta)_j$. We define
\begin{equation}\label{eq:ChFA-Appendix-local-4}
   h(\beta)=P(\beta)\,\mathrm{diag}\left[\hat{B}(\beta)_j\right]P(\beta)^{-1}X(0).
\end{equation}

For convenience, suppose the dimension of $A(t)$ in
(\ref{eq:ChFA-Appendix-local-10}) is $n\times n$. We should point
out that for any $t$, the rank of $A(t)$ is $n-1$, and thus for any
$t$, $A(t)$'s zero eigenvalue is one fold. In fact, the rank of any
infinitesimal generator with dimension $n\times n$ is $n-1$. This
implies that any $n-1$ columns or rows of this generator are
linearly independent. According to the definition of $A(t)$ in
Section~\ref{subsec:ChFA-twoComType-OneSyn-ConvergenceTheorem},
$A(t)$ is an infinitesimal generator if $\alpha(t)\neq 0$. If
$\alpha(t)=0$, $A(t)$ is also a generator after one column is
modified (see Proposition~\ref{proposition: ChFQ-Q1-structure}
and~\ref{proposition: ChFQ-Q2-structure}), which means that the
other $n-1$ columns are linearly independent. So whatever $t$ is,
the rank of $A(t)$ is $n-1$. Thus, the zero eigenvalue is one fold.
Therefore, $f(\beta)$'s zero eigenvalue is also one fold for any
$\beta$. So each entry of all blocks $\hat{B}(\beta)_j$ is zero,
except for the one corresponding to the zero eigenvalue, in which
case this block is a scalar one. This implies that for any $\beta$
all entries of the matrix $\mbox{diag}[\hat{B}(\beta)_j]$ are zeros,
except for a diagonal entry with one.

\par By permutation, $\mathrm{diag}\left[\hat{B}(\beta)_j\right]$ can always be
transformed into the form
$$
     \left[
                \begin{array}{cc}
                  1 & \mathbf{0} \\
                  \mathbf{0} & \mathbf{0}
                \end{array}
     \right].
$$
Correspondingly, $P(\beta)$ is permuted into $U(\beta)$. Therefore,
the formulae (\ref{eq:ChFA-Appendix-local-4}) can be written as
\begin{equation}\label{eq:ChFA-Appendix-local-5}
 h(\beta)=U(\beta) \left[
                \begin{array}{cc}
                  1 & \mathbf{0} \\
                  \mathbf{0} & \mathbf{0}
                \end{array}
     \right]U(\beta)^{-1}X(0),
\end{equation}

\par Now we prove a proposition which is used in
Section~\ref{subsec:ChFA-twoComType-OneSyn-ConvergenceTheorem}.

%\par \textbf{Proposition D.2.1}
\begin{proposition}\label{Proposition D.2.1.}
 For the $e^{tf(\beta)}X(0)$  in (\ref{eq:ChFA-Appendix-local-4-0})
and $h(\beta)$ in (\ref{eq:ChFA-Appendix-local-4}), we have
$$
    \lim_{t\rightarrow\infinity}\|e^{tf(\beta)}X(0)-h(\beta)\|=0.
$$
\end{proposition}
\begin{proof}
By a similar estimation as
(\ref{eq:ChFA-Appendix-x(t)-hat(x)-Estimation}), we have
$$\|e^{tf(\beta)}X(0)-h(\beta)\|\leq C(t)e^{-\Lambda_1t}.$$
where $C(t)$ is a polynomial of $t$. By a similar proof to
Lemma~\ref{lemma:ChFA-lambda(beta)-UniformLowerBound}, we have
$$
\Lambda_1=\inf_{\beta\in[0,1]}\{-\Re(\lambda)|\lambda \;\mbox{is
$f(\beta)$'s non-zero eigenvalue}\}>0.
$$
Then
$$\|e^{tf(\beta)}X(0)-h(\beta)\|\leq C(t)e^{-\Lambda_1t}\longrightarrow 0$$
as $t$ goes to infinity.
\end{proof}

Let $\hat{X}(t)=h(\beta(t))$, where $\beta(t)\in[0,1]$, and notice
$X(t)=f(\beta(t))$. As a consequence of this proposition, we have
%\par \textbf{Corollary D.2.1}
\begin{corollary}\label{Corollary D.2.1-2}
 Let $X(t)$ be the solution of
$\frac{\mathrm{d}X}{\mathrm{d}t}=A(t)X$ which is discussed in
Section~\ref{subsec:ChFA-twoComType-OneSyn-ConvergenceTheorem}, and
let $\hat{X}(t)=h(\beta(t))$, then
$$
    \lim_{t\rightarrow\infinity}\|X(t)-\hat{X}(t)\|=0.
$$
\end{corollary}
%\begin{proof}
%By a similar estimation as
%(\ref{eq:ChFA-Appendix-x(t)-hat(x)-Estimation}), we have
%$$\lim_{t\rightarrow\infinity}|X(t)-\hat{X}(\beta(t))|\leq C'e^{-\Lambda(\beta(t))t}f(t).$$
%where $C'$ is independent of $t$, and $f(t)$ is a polynomial of $t$
%with the maximum order $k$ independent of $t$. By
%(\ref{eq:ChFA-General-TwoOne-Lambda-Estimation}), i.e.
%$$
%\Lambda=\inf_{t\geq0}\{-\Re(\lambda)|\lambda \;\mbox{is $C(t)$'s
%non-zero eigenvalue}\}>0,
%$$
%We have
%$$\lim_{t\rightarrow\infinity}\mid X(t)-\hat{X}(\beta(t))|\leq C'e^{-\Lambda(\beta(t))t}f(t)
% \leq C'e^{-\Lambda t}f(t)\longrightarrow 0$$
%as $t$ goes to infinity.
%\end{proof}

%\subsection{A proof of
%Lemma~\ref{lemma:ChFA-lambda(beta)geq0}}\label{section:Appendix-Proof-Lemma}
\subsection{}\label{section:Appendix-Proof-Lemma}
This subsection presents a proof of
Lemma~\ref{lemma:ChFA-lambda(beta)geq0} in
Section~\ref{section:Proof-Not-Rely-On-Explicit-Expressions}.  Let
$\lambda(\beta)$ be a nonzero eigenvalue of the following
$f(\beta)$:
 \begin{equation*}%\label{eq:ChFA-local-f(beta)}
 f(\beta)=\left(
         \begin{array}{cc}
           -a\beta & b \\
           a\beta & -b \\
         \end{array}
       \right).
\end{equation*}
We will prove the following lemma which states that the real part of
$\lambda(\beta)$ is negative. The proof given here does not rely on
the explicit expression of the eigenvalue.

\par \textbf{Lemma}~\ref{lemma:ChFA-lambda(beta)geq0}:
For any $\beta\in [0,1]$, $\Re(\lambda(\beta))<0$, where
$\lambda(\beta)$ is a nonzero eigenvalue of $f(\beta)$.

\begin{proof}
After a shift $\max\{a\beta,b\}I$, $f(\beta)$ becomes
$\tilde{f}(\beta)=f(\beta)+\max\{a\beta,b\}I$, which is a
nonnegative matrix. Then similarly to the proof of
Theorem~\ref{thm:ChFA-Q1Q2-Eigen-NegativeRealParts}, which is based
on the Perron-Frobenious theorem
(Theorem~\ref{thm:ChFA-Perron-Frobenius Theorem}), we can conclude
that the eigenvalue other than zero has negative real part.
\end{proof}

%\bibliographystyle{elsarticle-harv}

%\bibliographystyle{elsarticle-num}
%\bibliography{JieDing_April2010}






% \end{thebibliography}

\end{document}